\documentclass[12pt]{article}
\usepackage[utf8]{inputenc}
\usepackage{amsfonts,amssymb,amsmath,amsthm}
\usepackage{mathrsfs}
\usepackage{fullpage}
\usepackage[colorlinks=true,linkcolor=blue,citecolor=blue,breaklinks=true]{hyperref}
\usepackage{breakcites}
\usepackage{algorithm}
\usepackage{algorithmicx}
\usepackage[noend]{algpseudocode}
\usepackage{bm}
\usepackage{dsfont}
\usepackage{graphicx, float, tikz}
\usepackage{colonequals}
\usepackage{color}
\usepackage{multicol}
\usepackage{tikz}
\usepackage{fancybox}
\usepackage[breakable, theorems, skins]{tcolorbox}
\usetikzlibrary{matrix,arrows}
\usetikzlibrary{positioning}
\usetikzlibrary{decorations}
\usetikzlibrary{decorations.pathreplacing}
\usetikzlibrary{shapes.geometric, chains, calc, fit, matrix}
\tikzstyle{system}=[rectangle,draw,fill=lightgray,minimum height=0.8cm,minimum
            width=0.8cm,thick]
\tikzstyle{BC}=[system]
\tikzstyle{resource}=[system]
\tikzstyle{RO}=[resource, minimum width=1cm]
\tikzstyle{protocol}=[circle, inner sep=0.7mm, draw]
\tikzstyle{simulator}=[circle, inner sep=0.7mm, draw]
\tikzstyle{memory}=[resource]
\tikzstyle{distinguisher}=[resource,fill=white,minimum width=3.5cm,
minimum height=1.2cm]
\tikzstyle{link}=[]

\usepackage{cleveref}

\usepackage{paralist}
\usepackage{enumerate}
\usepackage{enumitem}
\usepackage[n,
advantage,
operators,
sets,
adversary,
landau ,
probability,
notions,
logic,
ff,
mm,
primitives,
events,
complexity,
asymptotics,
keys]{cryptocode}

\usetikzlibrary{arrows,automata}
\usepackage{quantikz}

\usepackage{braket}

\usepackage[framemethod=tikz]{mdframed}
\usepackage{subcaption}

\usepackage{algpseudocode}

\usepackage{cleveref}

\newcommand{\bbN}{\mathbb{N}}

\newcommand{\calA}{\mathcal{A}}
\newcommand{\calB}{\mathcal{B}}

\newcommand{\calD}{\mathcal{D}}
\newcommand{\calE}{\mathcal{E}}
\newcommand{\calF}{\mathcal{F}}

\newcommand{\calH}{\mathcal{H}}

\newcommand{\calP}{\mathcal{P}}
\newcommand{\calQ}{\mathcal{Q}}
\newcommand{\calR}{\mathcal{R}}
\newcommand{\calS}{\mathcal{S}}
\newcommand{\calT}{\mathcal{T}}

\newcommand{\calW}{\mathcal{W}}
\newcommand{\calX}{\mathcal{X}}
\newcommand{\calY}{\mathcal{Y}}
\newcommand{\calZ}{\mathcal{Z}}

\newtheorem{theorem}{Theorem}
\newtheorem{corollary}{Corollary}
\newtheorem{definition}{Definition}
\newtheorem{lemma}{Lemma}
\newtheorem{claim}{Claim}
\newtheorem{remark}{Remark}

\theoremstyle{definition}
\newtheorem{construction}{Construction}
\crefname{theorem}{Theorem}{Theorems}
\crefname{corollary}{Corollary}{Corollaries}
\crefname{definition}{Definition}{Definitions}
\crefname{lemma}{Lemma}{Lemmas}
\crefname{claim}{Claim}{Claims}
\crefname{construction}{Construction}{Constructions}

\newcommand{\Accept}{\mathsf{Acc}}
\newcommand{\Reject}{\mathsf{Rej}}

\newcommand{\lang}{\mathcal{L}}
\newcommand{\NP}{\mathsf{NP}}

\newcommand{\QPT}{\mathsf{QPT}}
\newcommand{\PPT}{\mathsf{PPT}}

\renewcommand{\secpar}{\lambda}

\newcommand{\Hyb}{\mathsf{Hyb}}
\newcommand{\Sim}{\mathsf{Sim}}

\newcommand{\Com}{\mathsf{Com}}
\newcommand{\Enc}{\mathsf{Enc}}
\newcommand{\Dec}{\mathsf{Dec}}

\newcommand{\KeyGen}{\mathsf{KeyGen}}

\DeclareMathOperator{\Tr}{\mathsf{Tr}}

\newcommand{\ketbra}[1]{\ket{#1}\!\!\bra{#1}}

\newcommand{\CZ}{\mathsf{CZ}}

\newcommand{\email}[1]{\href{mailto:#1}{#1}}

\newcommand{\Setup}{\mathsf{Setup}}
\newcommand{\pubparams}{\mathsf{pp}}
\newcommand{\td}{\mathsf{td}}
\newcommand{\mode}{\mathsf{mode}}
\newcommand{\lossy}{\mathsf{lossy}}
\newcommand{\injective}{\mathsf{injective}}
\newcommand{\rand}{\mathsf{rand}}

\newcommand{\Inv}{\mathsf{Inv}}
\newcommand{\Good}{\mathsf{Good}}

\newcommand{\aux}{\mathsf{aux}}
\newcommand{\ct}{\mathsf{ct}}

\newcommand{\csh}{\mathsf{csh}}
\newcommand{\qsh}{\mathsf{qsh}}
\newcommand{\com}{\mathsf{com}}
\newcommand{\decom}{\mathsf{d}}

\newcommand{\cert}{\mathsf{cert}}
\newcommand{\verkey}{\mathsf{vk}}

\newcommand{\Gen}{\mathsf{Gen}}
\newcommand{\Ver}{\mathsf{Ver}}
\newcommand{\Commit}{\mathsf{Commit}}
\newcommand{\Reveal}{\mathsf{Reveal}}
\newcommand{\Eval}{\mathsf{Eval}}
\newcommand{\Share}{\mathsf{Share}}
\newcommand{\Reconstruct}{\mathsf{Reconstruct}}
\newcommand{\Del}{\mathsf{Del}}
\newcommand{\DecDel}{\mathsf{Dec}\text{-}\mathsf{Del}}

\newcommand{\PKE}{\mathsf{PKE}}
\newcommand{\FHE}{\mathsf{FHE}}
\newcommand{\SQSS}{\mathsf{SQSS}}
\newcommand{\COM}{\mathsf{COM}}

\newcommand{\CDGame}{\mathsf{CD\text{-}Game}}
\newcommand{\CollapseGame}{\mathsf{Collapse}}

\newcommand{\SPNP}{\WPNP}
\newcommand{\WPNP}{\mathsf{WP\text{-}NP}}

\newcommand{\Advtg}{\mathsf{Advtg}}

\newcommand{\Ext}{\mathsf{Ext}}
\newcommand{\IndExt}{\mathsf{IndExt}}

\newcommand{\RSP}{\mathsf{RSP}}

\newcommand{\LWE}{\mathsf{LWE}}

\newcommand{\execution}[1]{\left\langle #1 \right\rangle}
\newcommand{\INDCPA}{\mathsf{IND}\text{-}\mathsf{CPA}}

\title{Semi-Quantum Cryptography with Certified Deletion}

\author{Yael Tauman Kalai\thanks{MIT. \email{tauman@mit.edu}} \and Justin Raizes\thanks{NTT Research. \email{justin.raizes@ntt-research.com}}}

\date{}
\begin{document}
\maketitle

\begin{abstract}
    Certified deletion allows a client to upload encrypted data to a server as a quantum state, then later request that the server delete their data and detect whether the server complies. If verification passes, then the data on the server will remain hidden even if the decryption key is later leaked. In publicly verifiable certified deletion, the verification key is published so that anyone may check for deletion compliance.

In this work, we give a general compiler to publicly verifiable certified deletion that allows the client to initially upload the \emph{quantum} ciphertext using purely \emph{classical} communication, assuming the post-quantum hardness of LWE in the plain model. Our approach can be applied to a wide variety of primitives, including commitments and public-key, attribute-based, identity-based, and fully-homomorphic encryption.

Moreover, our constructions allow the client to classically manage the encrypted data beyond just verifying its deletion. The classical client may non-destructively audit the ciphertext via a Proof of No Intrusion to check whether it has been leaked to a third party. 
The classical client may also retrieve the data while \emph{simultaneously} verifying its deletion from the server. Thus, they do not have to choose between retrieving the data and protecting it against future key leakage.

As our core technical contribution, we give a classical-communication protocol and simulation technique that allows adapting purification-based security arguments for BB84-style states which are prepared through classical interaction. 
It makes the required purification available in a hybrid experiment, despite the classical transcript uniquely determining the prepared state in a real execution.

\end{abstract}

\clearpage
\tableofcontents
\clearpage

\pagenumbering{arabic} 

\section{Introduction}

When encrypted data reaches the end of its lifespan, its owner may request that the server storing it erase it to protect against future exposure of the decryption key.
However, erasure is impossible to enforce for classical ciphertexts because a malicious server can always undetectably copy the ciphertext.
Broadbent and Islam~\cite{TCC:BroIsl20} proposed the use of quantum ciphertexts to combat this impossibility. In encryption with certified deletion, erasure of the ciphertext produces a verifiable certificate. A valid certificate proves that the holder of the ciphertext can no longer recover the data, \emph{even if the decryption key is subsequently revealed}. 

Certified deletion thus provides a cryptographic mechanism to enforce data-management requests. 
Subsequent works developed this approach across a range of settings, allowing clients to obtain deletion guarantees for software~\cite{EC:BGKMRR24}, proofs of NP statements~\cite{AC:AbbKat25}, and many other cryptographic objects~\cite{C:HMNY22,C:BarKhu23,EC:AKNYY23,EC:HKMNPY24,C:BarRai24}.
Certified deletion is also useful for management tasks that do not require discarding the stored data.
For example, ~\cite{TCC:CGLR24} uses publicly verifiable certified deletion to detect whether a ciphertext has been stolen.

\paragraph{Importance of a Classical Client.}
Unfortunately, exercising these management capabilities typically requires quantum operations at the client. Preparing and uploading the initial ciphertext requires quantum communication, while intrusion detection may require coherent access to the stored ciphertext itself. 
While a large company may find it worthwhile to invest in their own quantum infrastructure, it is cost-prohibitive for individual users to do so. 

\paragraph{Classical Management of Quantum Data.}
A few works have made specific certified deletion-based management tasks accessible to fully classical clients. 
For classical upload, \cite{AC:HMNY21} constructs public-key encryption (PKE) with certified deletion either in the quantum random oracle model or assuming one-shot signatures and extractable witness encryption. However, a plain-model construction from standard assumptions remains open.
For auditing, \cite{EC:GoyRai26} introduced proofs of no intrusion, which allow a client to check for theft through classical interaction without destroying the stored ciphertext. Their construction requires an initial quantum upload, however, and explicitly leaves classical upload as an open problem.

Classical clients storing data on quantum servers also raises another question: how does the client retrieve their data? Many existing certified deletion schemes support the conversion from a quantum ciphertext to a classical ciphertext, which can then be sent to the client for decryption. However, this leaves the converted ciphertext as a permanent security risk in case of future key leakage. Ideally, the client should be able to retrieve their data while also erasing it from the server.

\begin{quote}
    \begin{center}
        \emph{Can a classical client upload, audit, and erase or retrieve encrypted data from a quantum server while protecting against future key leakage?}
    \end{center}
\end{quote}

\paragraph{Our Results.}
We present a simple compiler for converting any semantically-secure cryptographic primitive into the same primitive with certified deletion where the initial upload is \emph{classical}, assuming only the post-quantum hardness of LWE. Concretely, for any $X$ in $\{$secret-sharing, public-key encryption, attribute-based or identity-based encryption, fully-homomorphic encryption, commitment$\}$, we obtain the corresponding primitive by plugging $X$ into our compiler. The deletion certificates are publicly verifiable, which enables non-destructive auditing via \cite{DBLP:journals/iacr/KalaiKR26}'s generic compiler.
Moreover, we extend our compiler to generically add classical retrieval of the encrypted data while simultaneously deleting it from the server.

\paragraph{New Techniques.} 
Our main technical contribution is a method for applying purification-based security arguments to BB84-style states which are prepared through classical interaction. 
These arguments replace preparation of random $\{\ket{0/1}, \ket{+/-}\}$ states by half of an EPR pair where the reduction retains the purifying half. We develop a classical preparation protocol together with a simulation technique that makes the required purification available in a hybrid experiment. The analysis can then carry out the entanglement-based argument despite the fact that the classical transcript \emph{uniquely determines} the prepared state in a real execution. 

Thus, we are able to lift \cite{C:BarKhu23}'s entanglement-based argument for certified deletion to the classical upload case.
Given the ubiquity of purification-based arguments for arguing security of quantum cryptographic constructions, we expect the new technique to be of independent interest.

\subsection{Results}

As our main result, we give a general compiler for certified deletion where the initial upload uses only classical interaction, assuming the post-quantum hardness of LWE.
As compared to \cite{AC:HMNY21}, our constructions use far weaker assumptions and span a wider range of primitives. 
In particular, our compiler does not rely on complex primitives such as receiver non-committing encryption, which is incompatible with fully-homomorphic encryption~\cite{PKC:KatThiZho13}. 
Moreover, the compiled ciphertexts may be decrypted using only the secret keys from the original scheme, thus ensuring compatibility with existing public-key infrastructure.

\begin{theorem}\label{thm:sq-cd-primitives-intro}
    Assuming the post-quantum hardness of LWE and the existence of $X\in \{$secret-sharing, PKE, ABE, IBE, commitments$\}$, there exist semi-quantum $X$ with publicly-verifiable certified deletion in the plain model.
\end{theorem}

Our constructions also support auditing of the ciphertexts after they have been (classically) uploaded to the quantum server, answering an open problem from~\cite{EC:GoyRai26}. Proofs of no-intrusion (PoNIs) allow a quantum server to convince a classical client that no third party could decrypt the message \emph{even given the decryption key}. Crucially, the ciphertext remains intact after the PoNI, preventing the need to store an additional copy of the data elsewhere.
\cite{DBLP:journals/iacr/KalaiKR26} show a generic compiler from any PKE with \emph{publicly verifiable} certified deletion to a PKE with PoNIs. To test a ciphertext, the verifier asks for a non-destructive argument of knowledge that the server knows a valid certificate. This approach can be combined with any of the primitives constructed in this paper.

\begin{corollary}
    Assuming the post-quantum hardness of LWE and the existence of $X\in \{$secret-sharing, PKE, ABE, IBE, commitments$\}$, there exist semi-quantum $X$ with PoNIs. Search security holds in the plain model and decisional security holds in the quantum random oracle model (QROM).
\end{corollary}

Finally, we show an interactive classical decryption procedure for our constructions which allows the classical client to retrieve its data from the server while simultaneously deleting it.

\begin{theorem}[Informal]
    Assuming the post-quantum hardness of LWE, the encryption schemes from \Cref{thm:sq-cd-primitives-intro} can be equipped with an interactive classical protocol $\DecDel$ between a classical decryptor holding the decryption key and quantum server holding the ciphertext such that at the end of the protocol:
    \begin{itemize}
        \item The classical decryptor learns the encrypted message.
        \item The message is hidden from the server even given the decryption key.
    \end{itemize}
\end{theorem}

\section{Technical Overview}

We begin the overview of our techniques by discussing the leakage model and desired security guarantees in more detail.

\paragraph{Certified Deletion Security.} Security for certified deletion requires that after a ciphertext has been deleted, the message is hidden even given the secret decryption key. More explicitly, no quantum polynomial-time adversary should be able to distinguish the message $m\in \{0,1\}$ in the following game:
\begin{enumerate}
    \item Sample a fresh key pair $(\sk, \pk)$.
    \item Upload an encryption of $m$ to the adversarial server by interacting with it. The challenger obtains a verification key $\verkey$.  
    \item The adversarial server outputs a certificate $\cert$.
    \item If $\cert$ passes verification under $\verkey$, send the decryption key $\sk$ to the adversary and the adversary outputs a guess $m'$. Otherwise, the adversary does not get to guess.
\end{enumerate}

\paragraph{Leakage Modeling.} 
Several works have also considered a stronger notion of deletion called ``everlasting security'', wherein the message is hidden after deletion even against unbounded adversaries. This notion is unlikely to be achievable with classical upload because the classical transcript (which cannot be deleted) must information-theoretically determine the prepared state.
Nonetheless, the randomness used for encryption may be immediately discarded, leaving the decryption key $\sk$ as the only persistent piece of data that can be leaked.

\subsection{Ciphertext Design}\label{sec:overview:ciphertext-design}

Both our construction and \cite{AC:HMNY21}'s prior construction in the QROM start from a common trick for preparing states of the form
\[
    \ket{0, x_0} + \ket{1, x_1}.
\]
First, the sender samples public parameters and a trapdoor $(\pubparams, \td)$ for a trapdoor claw-free function (TCF). A TCF $f_\pubparams$ is a randomized 2-to-1 function where is it hard to find a collision (or a ``claw'') $f_\pubparams(0; x_0) = f_\pubparams(1; x_1) = y$ given only $\pubparams$, but any image $y$ can be efficiently inverted given $\td$. The receiver can prepare superpositions over collisions by evaluating $f_\pubparams$ on a uniform superposition then measuring the output $y$, thereby collapsing the state to a single collision. Afterwards, the sender can learn the prepared state by using $\td$ to recover both preimages of $y$.

Both of our constructions also make use of the adaptive hardcore bit property~\cite{BCMVV18}. Roughly, this property guarantees it is computationally hard to both produce a preimage $(b, x_b)$ and an inner product $\vec{d}\cdot (x_0 \oplus x_1)$ with a $\vec{d}$ of the adversary's choice. If the adversary does find a preimage, they should not be able to also guess an adaptive inner product with better than $1/2 + \negl$ probability.
Our construction additionally uses a collapsing property: a $\QPT$ adversary who provides a superposition over preimages $\ket{b, x_b}$ cannot detect if it is measured in the computational basis before being returned.

\paragraph{HMNY21's QROM Construction.} 
\cite{AC:HMNY21}'s construction begins by preparing the claw state, then using the two preimages to derive random masks $H(x_0)$ and $H(x_1)$ for $m$.
The resulting ciphertext looks like
\[
    \ket{0, x_0} + \ket{1, x_1},\qquad \Enc(H(x_0) \oplus m,\ H(x_1)\oplus m)
\]

Deletion works by measuring the claw superposition in the Hadamard basis to obtain a string $d$ such that $d\cdot (x_0 \oplus x_1) = 0$. If one produces such a string $d$, then the adaptive hardcore bit property of $f_\pubparams$ implies they cannot also find one of the preimages $x_0$ or $x_1$, except with probability $1/2$. To contradict the adaptive hardcore bit property, the crux of the argument is to show that any successful distinguisher must ``know'' one of $x_0$ or $x_1$ after deletion. This is precisely where a random oracle (or extractable witness encryption) becomes useful: any algorithm which knows $H(x_0)$ must also know $x_0$.

Their final construction uses a more sophisticated approach to amplify security while avoiding a blow-up in the ciphertext size due to the number of possible $x_b$ combinations, but its analysis follows from similar principles and is thus limited to random oracles or extractable witness encryption.

\paragraph{Moving $m$ to the Phase.} 
The need for knowledge assumptions stems from using preimages $x_0$ or $x_1$ to hide the message, while the Hadamard basis is used to encode the deletion certificate. 
We use an approach from \cite{C:BarKhuPor23} which flips this relation:
encode the message in the phase and use preimages as the certificate. As a welcome benefit, their approach is specifically designed for \emph{public verifiability} and does not need an amplification step.
More explicitly, our eventual ciphertexts will be of the form
\[
    \ket{0, x_0} + (-1)^p \ket{1, x_1},\quad \Enc_{\pk}((p\oplus m, \td))
\]
To decrypt using $\sk$, one measures the collision superposition in the Hadamard basis to obtain a string $(b', d)$ such that $b' \oplus d\cdot (x_0\oplus x_1) = p$. Then, $\Enc_{\pk}(\td)$ can be decrypted using $\sk$ and the trapdoor $\td$ subsequently used to recover the preimages $(x_0, x_1)$, and therefore $p$. Finally, $p$ is used to unmask $m$.

Deletion works by measuring the collision superposition to get either $x_0$ or $x_1$. These can be publicly checked using the published image $y = f_\pubparams(0,x_0) = f_\pubparams(1,x_1)$. 
Intuitively, measuring $\ket{x_0} + (-1)^p \ket{x_1}$ in the computational basis destroys the phase $p$, which prevents decrypting $p\oplus m$ forevermore. 
The actual analysis is more complicated and depends on the preparation protocol, which we will describe shortly hereafter.

\paragraph{Other Primitives (\Cref{sec:applications}).} 
The above construction also allows replacing $\Enc(\bullet)$ by any distribution which semantically hides its input against QPT adversaries. Thus we may implement \cite{C:BarKhu23}'s formula for achieving various primitives with certified deletion by encapsulating $(p\oplus m, \td)$ under the desired primitive.
For example, by encrypting $(p\oplus m, \td)$ under any post-quantum Attribute-Based Encryption (ABE), we achieve semi-quantum ABE with certified deletion; the message may only be decrypted if one possesses a key with the correct attributes. Another example is to commit to $(p\oplus m, \td)$, whereupon the sender becomes bound to the message. 

In the technical sections, we abstract out the base functionality as a secret-sharing with a quantum share $\ket{\qsh} = \ket{0, x_0} + (-1)^p \ket{1, x_1}$ and a classical share $\csh = (p\oplus m, \td)$. We detail our other applications in \Cref{sec:applications}, including Identy-Based Encryption (IBE) and Fully-Homomorphic Encryption (FHE).

\subsection{Abstract Framework for Preparation of Hadamard States}
\label{sec:overview:abstract-framework}

Since the ciphertext includes a quantum state, we still need a safe method to prepare it with only classical communication. The generation of collision superpositions $\ket{0,x_0} + \ket{1,x_1}$ is well-established, but the ciphertext described above \emph{also} requires a relative phase. 

In preparation for the full construction, this section abstracts out our core technique to focus on the preparation of the ``canonical'' relative phase states: $\ket{+}$ and $\ket{-}$.

\paragraph{Prior Approaches.}
Prior works~\cite{FOCS:GheVid19,AC:CCKW19,ICALP:GheMetPor23,FOCS:Zhang22,TCC:CheHerVu23,ITCS:Zhang25,C:BarKhu25} commonly build on the following basic approach to preparing Hadamard basis states:
\begin{enumerate}
    \item The server prepares $\secpar$-many claw states $\Ket{0, x_{0}^{(i)}} + \Ket{1, x_{1}^{(i)}}$ with images $y_i$.
    \item The client asks the server to reveal preimages for a random set $S \subset \{y_i\}_i$. The server measures those claw states and sends the corresponding preimages $\left(b_i, x_{b_i}^{(i)}\right)$.
    \item The client picks one of the remaining indices $i^*$ and has the server measure the $x$ part of the claw $\Ket{0, x_{0}^{(i^*)}} + \Ket{1, x_{1}^{(i^*)}}$ to obtain $d$ and $\ket{0} + (-1)^{p}\ket{1}$ where $p = d\cdot \left(x_0^{(i^*)} \oplus x^{(i^*)}_1\right)$, up to global phase.\footnote{Technically $d$ also needs to satisfy some basic requirements such as $d \neq 0$. We omit this from the overview for simplicity of exposition.}
    \item The client computes the phase $p = d\cdot (x_0 \oplus x_1)$ and the server outputs the prepared state $\ket{0} + (-1)^{p}\ket{1}$.
\end{enumerate}
The middle two cut \& choose steps are effectively an \emph{argument of knowledge} of preimages for each of the $y_i$.
In particular, it guarantees that the server knows a preimage for the chosen $y_{i^*}$, except with probability $1/\secpar$ (the ``knowledge error'').

The weakness of cut-and-choose is that it has high knowledge error scaling with the number of states tested, which directly translates to a $1/\poly$ loss in security. 
In the $1/\poly$-probability case where the server does not have to know a preimage, the adaptive hardcore bit property does not apply. Then it becomes difficult to rule out the server forcing preparation of, say, $\ket{+}$ by somehow obliviously sampling $d$ such that $d \cdot (x_0 \oplus x_1) = 0$, without ever knowing a preimage.

\paragraph{Non-Destructive Testing Allows Lower Knowledge Error.}
The limiting factor for knowledge error in prior arguments is that they \emph{destroy} the tested indices. In other words, the output can only be an \emph{untested} index. \cite{DBLP:journals/iacr/KalaiKR26}'s recent witness-preserving argument of knowledge allows us to bypass this limitation by preserving the tested state. Thus, it may achieve negligible knowledge error for the output state. Concretely, to prepare an $H\ket{p}$ state:
\begin{enumerate}
    \item The server prepares a \emph{single} claw state $\ket{0, x_0} + \ket{1, x_1}$ with image $y$.
    \item The server gives a witness-preserving argument of knowledge for a preimage such that $f_\pubparams(b, x_b) = y$. Crucially, they still hold $\ket{0, x_0} + \ket{1, x_1}$ at the end of the argument.
    \item The server measures the $x$ part of the claw to obtain $d$ and $\ket{0} + (-1)^p \ket{1}$.
    \item The client computes $p = d\cdot (x_0 \oplus x_1)$ to identify the prepared state.
\end{enumerate}

\paragraph{Peeking Ahead at Purification.}
One of the most common and useful techniques for analyzing protocols which send $\ket{+}$ or $\ket{-}$ states at random is to purify the choice of which state to send. The reduction instead sends one half of an EPR pair and keeps the other half, often for usage later in the protocol. It would be useful to adapt such arguments to states which are prepared via classical communication.

The difficulty of doing so is that the transcript \emph{uniquely determines} the prepared state -- sending the classical messages means deciding on the state. The key to adapting the purification technique is to 
allow the same classical transcript to prepare both possible states.
We can do so using the following high-level template:
\begin{enumerate}
    \item The adversary declares an image $y$.
    \item
    Coherently extract a witness $(b, x_b)$ during the argument of knowledge. \textbf{Change the prepared state by applying $Z$ to the first qubit $b$.}
    \item The adversary completes the preparation by declaring $d$.
\end{enumerate}
The $Z$ operation switches between preparing $\ket{+}$ and $\ket{-}$. Thus, the choice of which state is prepared may be purified by purifying the choice of whether to apply $Z$.

We discuss the mechanics of how to implement this template for the specific case of certified deletion in \Cref{sec:overview:analysis}. As part of that implementation, we will need the $Z$ phase injection to be indistinguishable \emph{even if the choice of state is later revealed}. Of course, if the challenger tells the adversary it prepared $\ket{+}$ after changing the state to $\ket{-}$, the adversary could immediately detect this by measuring the state. 
We show in the Phase Injection Lemma (\Cref{sec:plug-in-lemma}) that the change is indistinguishable so long as the challenger also updates the revealed phase to $p\oplus 1$.
This indistinguishability is not obvious, even from the TCF's collapsing or hardcore adaptive bit properties. We discuss more details in 
\Cref{sec:overview:analysis}.

\subsection{Classical Upload}

Next, we specialize the abstract approach to prepare our  specific ciphertext states
\[
    \Ket{0, x_0} + (-1)^{p}\Ket{1, x_1}.
\]
The abstract approach effectively adds a phase $p$ to the first qubit by entangling it with a TCF collision, then measuring the colliding preimages in the Hadamard basis. 
Adapting this to our ciphertext structure, we use two TCF instances $A$ and $B$.
$A$ provides the base structure of the ciphertext and its preimages $(b, x_b^{(A)})$ will serve as the deletion certificate. The other instance $B$ is used during preparation to add a phase to the $A$ claw.
This yields our preparation protocol:
\begin{enumerate}
    \item The \textbf{sender} samples two sets of public parameters and trapdoors $(\td_A, \pubparams_A)$ and $(\td_B, \pubparams_B)$.
    \item The \textbf{receiver} evaluates and measures $y^{(A)} = f_{\pubparams_A}\left(b, x^{(A)}\right)$ and $y^{(B)} = f_{\pubparams_B}\left(b, x^{(B)}\right)$ on a uniform superposition $\sum_{b, x^{(A)}, x^{(B)}} \Ket{b, x^{(A)}, x^{(B)}}$ to obtain an entangled collision superposition 
    \[
        \Ket{0, x_0^{(A)}, x_0^{(B)}} + \Ket{1, x_1^{(A)}, x_1^{(B)}}
    \]
    It then sends the images $y^{(A)}$ and $y^{(B)}$ to the sender.
    \item \label{step:overview-SPNP} The \textbf{receiver} uses a witness-preserving argument of knowledge for NP~\cite{DBLP:journals/iacr/KalaiKR26} to prove knowledge of a string $(b, x_b^{(A)}, x_b^{(B)})$ such that 
    \[
        y^{(A)} = f_{\pubparams_{A}}\left(b, x_b^{(A)}\right)
        \quad \text{and} \quad
        y^{(B)} = f_{\pubparams_{B}}\left(b, x_b^{(B)}\right).
    \]
    Crucially, the state-preservation property ensures that the superposition over witnesses is not disturbed.
    \item The \textbf{receiver} measures the register containing the $B$ preimage $x_b^{(B)}$ in the Hadamard basis to obtain a string $d_B$. After factoring out a global phase $(-1)^{d_B\cdot x_0^{(B)}}$, this leaves the state
    \begin{align*}
        (-1)^{d_B\cdot x_0^{(B)}}\left(\Ket{0, x_0^{(A)}} + (-1)^{d_B\cdot \left(x_0^{(B)}\oplus x_1^{(B)}\right)}\Ket{1, x_1^{(A)}}\right) 
    \end{align*}
    \item Finally, the \textbf{sender} inverts the $B$ preimages $(x_0^{(B)}, x_1^{(B)})$ and uses them to compute $p_B = d_B\cdot \left(x_0^{(B)}\oplus x_1^{(B)}\right)$, and sends back 
    \[
        \Enc((p_B \oplus m, \td_A))
    \]
\end{enumerate}

\paragraph{Importance of Entanglement.}
The argument of knowledge in step \ref{step:overview-SPNP} is even more important for this specific construction than alluded to in the abstract framework.
By asking that the receiver knows preimages $x_b^{(A)}$ and $x_b^{(B)}$ using \emph{the same bit $b$}, the sender forces the preimages to be entangled. 
Hadamard measuring the $B$ preimages therefore adds a phase $p$ to the $A$ preimages, which requires $\td_A$ to recover. 

There is even an explicit attack if the adversary were able to avoid entangling their superpositions. Suppose a malicious receiver instead prepared the collision superpositions in tensor:
\[
    \left(\Ket{0, x_0^{(A)}} + \Ket{1, x_1^{(A)}}\right) \otimes 
    \left(\Ket{0, x_0^{(B)}} + \Ket{1, x_1^{(B)}}\right)
\]
Now if the adversary measures their whole $B$ register, including the first bit, in the Hadamard basis, they get $(p_B, d_B)$ such that $d_B \cdot \left(x_0^{(B)}\oplus x_1^{(B)}\right) = p_B$. 
The adversarial receiver can simply send $d_B$ to the sender and receive an encryption of $p_B \oplus m$ back. 
Separately, the $A$ register can be measured and the resulting $(b, x_b^{(A)})$ exchanged for the decryption key. Finally, they can use this to decrypt $p_B \oplus m$ and unmask the message using the previously obtained $p_B$.
Forcing the $A$ and $B$ registers to be entangled is exactly what prevents $p_B$ from being directly revealed.

\subsection{Analysis: Entanglement Over Classical Channels}
\label{sec:overview:analysis}

The analysis of our construction adapts \cite{C:BarKhuPor23}'s original proof to the classical preparation case.
Unfortunately, the most difficult part to adapt is also the most critical part of their analysis: the ability to send a ``purified'' ciphertext which is in superposition between encryptions of $m = 0$ and $m=1$.
Since the message is masked by the phase $p$, this is achieved by purifying the phase:
\[
    \sum_{p\in \{0,1\}} \left(\Ket{0, x_{0}^{(A)}} + (-1)^{p_B}\Ket{1, x_{1}^{(A)}}\right)_{\calA} \otimes \ket{p_B}_{\calP}
\]
Here, the challenger sends register $\calA$ to the adversary and keeps register $\calP$. 
Unfortunately, the phase $p$ is completely determined by the classical transcript. How can we argue that the server's view is indistinguishable from the purified version?

The main technical contribution of this work is a general technique for applying BB84-style purification arguments when preparing states over classical channels.
In this section, we discuss why purification is so critical to BKP23's strategy and develop a general technique to adapt purification arguments for classically prepared states. 
Then, we show how to use the new technique to complete the adaptation of BKP23's proof for our construction.

\paragraph{Quantum-Communication Strategy.} To explain why the purification is so crucial to BKP23's strategy, we give a high-level overview of the whole strategy. Roughly, it proceeds in 5 steps:
\begin{enumerate}
    \item Delay the dependence on the message by encrypting a random $\tilde{m}$ instead of $m\oplus p_B$, then checking at the end of the experiment if $\tilde{m} = m\oplus p_B$ and aborting if not. Any adversary's distinguishing advantage in this experiment is precisely $1/2$ of its original advantage.
    \item Purify $p_B$ by instead preparing the state
    \[
        \sum_{p_B\in \{0,1\}} \left(\Ket{0, x_{0}^{(A)}} + (-1)^{p_B}\Ket{1, x_{1}^{(A)}}\right)_{\calA} \otimes \ket{p_B}_{\calP}
    \]
    and sending register $\calA$ to the adversary. The abort condition at the end of the experiment is decided by measuring register $\calP$ in the computational basis to obtain $p_B$ and checking whether $\widetilde{m} = m\oplus p_B$.
    
    \item Argue that if the adversary returns a valid certificate $\cert = (b, x_{b}^{(A)})$, then the phase purification register $\calP$ contains $H\ket{b}$ with overwhelming probability. This step relies on the hardness of finding both $x_{0}^{(A)}$ and $x_1^{(A)}$.
    
    \item Observe that now the abort step is \emph{independent} of $m$ because $p_B$ is the result of measuring a Hadamard basis state in the computational basis.The rest of the experiment is also independent of $m$, so $m$ cannot be guessed with any advantage. Returning to the main experiment, the original advantage can only be at most $2 \cdot \negl$ which is still negligible.
\end{enumerate}
The purification moves the dependence on $m$ to be coherently recorded in the phase purification register $\calP$.
$\calP$ is then forcefully decohered by measuring it in two maximally incompatible bases: the Hadamard basis (from the certificate) followed by the computational basis (from the abort check).

\paragraph{Separate Roles for $A$ and $B$.}
The two most load-bearing steps in BKP23's strategy are the phase purification in step 2 and the argument that $\calP$ contains $H\ket{b}$ in step 3. 
Similarly to their roles as phase ($B$) and structure ($A)$ in the classical upload protocol, $B$'s security will be used to argue phase purification and $A$'s security will be used to argue $\calP$'s property during deletion. We remark that it is particularly difficulty to use $A$'s security outside of the very careful usage in BKP23's original proof, because $\td_A$ is eventually revealed.

\paragraph{Dissociating the Phase and Transcript.}
The trouble with purifying the phase is that each classical transcript uniquely identifies a single phase $p_B$. We cannot purify $p_B$ by leaving the transcript in superposition because the adversary expects a classical transcript and will eventually measure it, collapsing the phase.

Instead, we will allow the same transcript to be associated with both $p_B=0$ and $p_B=1$ by ``injecting'' a phase without changing the transcript. Once the phase is independent of the transcript, the choice of $p_B$ can be purified.
The phase injection lemma, discussed previously in \Cref{sec:overview:abstract-framework}, provides a general template for switching the phase on the same transcript: coherently extract the preimage $\Ket{b, x_{b}^{(A)}}$, then apply $Z$ to the first qubit $b$.

\paragraph{Temporarily Exposing the Witness.}
The argument of knowledge for a witness $\left(b, x_{b}^{(A)}, x_b^{(B)}\right)$ provides the needed extraction property.
\cite{FOCS:LomMaSpo22} show a general strategy to extract a witness from a (coherent) accepting transcript with overwhelming probability. Moreover, their extractor is compatible with \cite{DBLP:journals/iacr/KalaiKR26}'s witness-preserving NP argument. Once the witness is exposed, a new phase $p'_B$ can be ``injected'' by applying $Z^{p'_B}$ to the first qubit $b$. 

However, there is a tension between extraction and witness preservation. 
Extracting the witness would disturb the prover's copy, but the prover expects their witness superposition to be undisturbed.
It is important to ``return'' the extracted witness before the end of the argument so that the prover regains it -- otherwise the prover's view is distinguishable.

If the extractor were a \emph{unitary} $U_\Ext$, then it could be inverted to return the witness. The following strategy thus resolves the extraction-versus-preservation tension:
\begin{enumerate}
    \item Coherently execute the argument.
    \item Run the unitary extractor $U_{\Ext}$.
    \item Check (but do not otherwise disturb) the exposed witness.
    \item Invert the unitary extractor by running $U_{\Ext}^\dagger$.
\end{enumerate}

We call this strategy ``indistinguishable-extraction'' because it both extracts the witness and returns the adversarial prover's view to an indistinguishable state. 
Unlike \cite{FOCS:LomMaSpo22}'s ``state-preserving extraction'',\footnote{Not to be confused with the witness-preserving property from \cite{DBLP:journals/iacr/KalaiKR26}.} which also extracts while generating an indistinguishable adversarial prover view, indistinguishable-extraction is compatible with witness preservation because it does not unnecessarily measure the extracted witness.

The caveat is that the unitary dilation of \cite{FOCS:LomMaSpo22}'s guaranteed extractor is arbitrarily large. 
Fortunately, they also observe that their guaranteed extractor can be $\epsilon$-approximated by a size $\poly[1/\epsilon]$ unitary $U_{\Ext, \epsilon}$ for any $\epsilon = 1/\poly$. Thus, the above strategy is only $\epsilon$-indistinguishable from a real argument.
Since $\epsilon$ can be scaled to arbitrarily low $1/\poly$ without changing the real protocol, the $\epsilon$-approximation is sufficient for proving falsifiable security properties such as certified deletion. 

\begin{em}
\paragraph{\emph{Remark: The Role of $\epsilon$.}} 
It is important to differentiate between the roles of $\epsilon$ in our proof and in existing remote state preparation (RSP) protocols.
Previously, we mentioned that existing RSP protocols are unsuitable for our purposes because they incur an $\epsilon = 1/\poly$ error. There, the protocol description is a function of $\epsilon$; once $\epsilon$ is fixed, $\RSP(\epsilon)$ will fail to have $\epsilon'$ security for $\epsilon'<\epsilon$. 
In contrast, the role of $\epsilon$ in our analysis is to determine the reduction runtime. The same protocol is secure for both $\epsilon$ and $\epsilon' < \epsilon$, but the reduction for $\epsilon'$-security is a longer computation. We show that the security error in our single protocol is $<1/\secpar^c$ for all polynomials $\secpar^c$, so it must be negligible.
\end{em}

\paragraph{Indistinguishability of Phase Injection.}
Temporarily extracting $(b, x_b^{(A)}, x_b^{(B)})$ gives a window of opportunity to inject a chosen phase $p'$, but it is not a-priori clear that injecting the phase is indistinguishable. After all, the adversary eventually sees both $\td_A$ and $p_B = d_B \cdot (x_0^{(B)}\oplus x_1^{(B)})$ in the deletion game. 
It is especially difficult to rely on any security of $A$ because its full trapdoor is later revealed. Revealing $p_B$ is safer, but still requires some care to rely on the security of $B$.

The phase injection lemma guarantees that changing the phase from $p_B$ to $p_B \oplus p'$ is indistinguishable even when $p_B \oplus p'$ is revealed:
\begin{enumerate}
    \item Sample fresh $\pubparams_B$, then the adversary sends an evaluation $y_B$ and an arbitrary superposition over its images $\left(b, x_b^{(B)}\right)$.
    \item If the preimages are valid, inject the phase $p'$ by applying $Z^{p'_B}$ to the first qubit $b$. Then, return the state.
    \item The adversary chooses $d_B$.
    \item Compute $p_B = d_B \cdot (x_0^{(B)}\oplus x_1^{(B)})$, then send $p_B \oplus p'$ to the adversary.
\end{enumerate}
In the broader context of the analysis, the phase is injected during the temporary extraction in between $U_\Ext$ and $U_\Ext^\dagger$. Then, $p_B \oplus p'$ is revealed and used at the end of the experiment to decide whether to abort.
The phase injection lemma thus allows indistinguishably changing the phase by $p'$. Then, we can achieve our purification goal by purifying the choice of $p'$.

\paragraph{Proof of the Phase Injection Lemma.}
Although revealing $p_B$ is less dangerous than revealing a full trapdoor, it too presents some difficulties. Hiding the injected phase $p'$ would be trivial if we could undetectably measure the extracted $x_b^{(B)}$, since computational basis measurement is equivalent to injecting a random phase. However, revealing $p_B$ disrupts $\pubparams_B$'s collapsing property. Moreover, we cannot even replace $p_B$ by a random bit using adaptive hardcore bits because the adversary never produces a preimage $x_b^{(B)}$ of $y_B$ \emph{at the same time} that it would see $p_B$. Measuring $x_b^{(B)}$ during extraction and keeping it for this purpose would return to the collapsing property -- a circularity.

To break the circularity between collapsing and adaptive hardcore bits, we consider 4 hybrid experiments. In experiment $\Hyb_{ab}$, we inject $a$, but return $p_B \oplus b$. The mixed distribution $\Hyb_{a*} = 1/2(\Hyb_{a0} + \Hyb_{a1})$ now inserts a fixed phase $a$, but returns a random bit at the end. This decouples the phase from the reported bit so that we may prove that the phase is random when $p_B$ is not reported, i.e. $\Hyb_{0*} \approx \Hyb_{1*}$. Similarly, we can show that the reported bit is indistinguishable from random when we insert an \emph{independently random} phase, i.e. $\Hyb_{*0} \approx \Hyb_{*1}$. Finally, some algebra allows us to derive the desired cross-term relation: $\Hyb_{00} \approx \Hyb_{11}$.

\paragraph{Entanglement Over Classical Channels.}
Putting everything together, we purify the phase over a classical channel as follows:
\begin{enumerate}
    \item Prepare the public parameters $(\pubparams_A, \pubparams_B)$ and receive two evaluations $(y_A, y_B)$ from the adversary.
    \item Use a unitary $U_\Ext$ to coherently extract a superposition over entangled preimages $(b, x_b^{(A)}, x_b^{(B)})$ of $y_A$ and $y_B$.
    \item Prepare a $\ket{+}$ state in register $\calP$, then perform a controlled $Z$ from register $\calP$ to the first qubit $b$ of the extracted preimages.
    \item ``Return'' the preimages by inverting $U_{\Ext}$.
    \item Encrypt a random bit $\tilde{m}$ as in step 1 of BK23's strategy.
    \item At the end of the deletion hybrid experiment, measure the purification register $\calP$ to obtain $p'$ and compare $p_B \oplus p'$ to $\tilde{m}$ for the abort check.
\end{enumerate}

\paragraph{Finishing the Deletion Proof.}
Armed with the purification technique, the rest of BKP23's strategy can be adapted smoothly. The primary remaining step is to argue that if the adversary returns a preimage $(b, x_{b}^{(A)})$, then the phase purification register contains $H\ket{b}$ with overwhelming probability. This uses a similar approach to BK23's original, with some careful modifications to account for inserting the phase to a pre-prepared state instead of preparing the whole state from scratch.

First, observe that the predicate ``$(b, x_{b}^{(A)})$ matches the contents of $\calP$'' can be efficiently checked \emph{before} $\sk$ is revealed. Thus, using the semantic security of $\Enc$ we may replace $\Enc((p_{B}\oplus m, \td_A))$ by $\Enc(0)$ without noticeably affecting whether the predicate holds. Without $\td_A$, the collapsing property of TCFs ensures it is now indistinguishable to \emph{measure} the extracted witness $\left(b', x_{b'}^{(A)}, x_{b'}^{(B)}\right)$ after injecting the phase, in between $U_\Ext$ and $U_\Ext^\dagger$. Then the certificate $\left(b, x_{b}^{(A)}\right)$ must match $\left(b', x_{b'}^{(A)}\right)$, i.e. $b = b'$, or else the adversary has found two preimages and violated the claw-free property of the TCF.

The final step is to argue that the purification register $\calP$ matches the preimage $(b', x_{b}^{(A)})$ which was extracted between $U_\Ext$ and $U_\Ext^\dagger$. Recall that the phase injection was done by applying a $Z$ operation to $b'$, controlled on the phase register $\calP$. Controlled $Z$ is symmetric between the control and the target, i.e. $\CZ_{1\rightarrow 2} = \CZ_{2\rightarrow 1}$, so it is also equivalent to apply $Z$ to $\calP$ controlled on $b'$ instead. Finally, $\calP$ was initialized to $\ket{+}$, so after the controlled $Z$ it contains $H\ket{b'}$.

\subsection{Proofs of No-Intrusion and Classical Decryption with Deletion}\label{sec:overview:poni-dec}

\paragraph{Proofs of No-Intrusion (\Cref{sec:poni}).} 
A proof of no-intrusion (PoNI)~\cite{EC:GoyRai26} allows a classical client to interactively test that a quantum ciphertext does not exist outside of the proving server, in the sense that no third party can decrypt the message even given the decryption key. Moreover, this test preserves the ciphertext so that it may be later decrypted or tested again.
Kalai, Khurana, and Raizes~\cite{DBLP:journals/iacr/KalaiKR26} gave a generic compiler to add PoNIs to any encryption scheme with \emph{publicly verifiable} deletion and classical certificates: the server coherently computes a superposition over certificates, then simply gives a witness-preserving argument of knowledge that it knows a valid certificate. Since the certificate superposition is not disturbed, the server can then uncompute the certificate superposition to return the ciphertext to its original state.
This compiler may be generically applied to our semi-quantum encryption scheme with certified deletion because the certificates are classical and publicly verifiable.

\paragraph{Decryption with Deletion (\Cref{sec:deldec}).} 
Finally, we describe how a classical client holding a decryption key $\sk$ may interact with a server holding a ciphertext $\ket{\ct}$ from the above scheme to learn the plaintext while also deleting $\ket{\ct}$ from the server. 
$\ket{\ct}$ can easily be decrypted over a classical channel by first measuring it in the Hadamard basis to convert it to a classical ciphertext $\ct'$, then sending $\ct'$ to the client for decryption. However, such an approach leaves the server also able to decrypt the converted $\ct'$ should it later learn the decryption key.

The trick is to rerandomize $\ket{\ct}$ \emph{before conversion}. The goal is that the converted $\ct'$ should instead decrypt to $m\oplus p'$ for a fresh $p'$ which can be discarded after decryption.
Thus, even if the server later learns the decryption key, $m$ is still protected by unknown randomness.

The client may remotely rerandomize the ciphertext using a familiar trick: entangle the quantum part of the ciphertext
\[
    \ket{0, x_0^{(A)}} + (-1)^{p} \ket{1, x_1^{(A)}}
\]
with a third pair of preimages $(x_0^{(C)}, x_1^{(C)})$, prove entanglement via a witness-preserving argument of knowledge, then measure the $C$ preimages in the Hadamard basis to add a new phase $p_C$ to the ciphertext. 

To show security, observe that neither $\td_C$ nor $p_C$ are ever revealed as a result of interactive decryption (unlike how $p_B$ can be obtained from $\Enc(m\oplus p_B)$ and $\sk$ in the basic deletion game). Thus we may rely on $\pubparams_C$'s collapsing property and the forced entanglement with $x_b^{(A)}$ to undetectably measure the $A$ preimage $(b, x_{b}^{(A)})$ when extracting the witness. Finally, $(b, x_{b}^{(A)})$ can be used as the certificate in the reduction to certified deletion.

\subsection{Discussion and Related Works}

\paragraph{Limitations of Our Preparation Technique.}
The gold standards for remote state preparation are simulation-security and verifiability, both of which this technique does not achieve. Simulation-security says that the adversary's view could be generated given the prepared state directly, whereas verifiability requires that the adversary's registers actually contain the generated state.\footnote{Up to a local isometry that depends only on the adversary, not on the prepared state.} 
Simulation-security is especially flexible because it would allow generically replacing quantum communication in cryptographic protocols~\cite{FOCS:GheVid19,ICALP:GheMetPor23}.

Unfortunately, with our techniques an adversary could always evaluate $f_\pubparams$ purely classically and pass the test. Adding a Hadamard test is also difficult because verifying a Hadamard basis measurement requires the use of the secret trapdoor, but \cite{DBLP:journals/iacr/KalaiKR26}'s witness-preserving argument only supports publicly-verifiable statements.
Nonetheless, our techniques provide a simple and efficient method for a wide class of protocols whose security can be based on BB84 purification arguments. 

\paragraph{Previously Mentioned Prior Work.} 
Hiroka, Morimae, Nishimaki, and Yamakawa previously constructed PKE with certified deletion where the ciphertexts may be classically uploaded, assuming either quantum random oracles or extractable witness encryption and one-shot signatures. Their extractable witness encryption construction is publicly verifiable. See \Cref{sec:overview:ciphertext-design}.

Goyal and Raizes proposed and constructed PoNIs, assuming the post-quantum hardness of LWE. They achieve search security in the plain model and decisional in the quantum random oracle model. Their construction has classical communication during the protocol, but requires quantum communication to upload the original ciphertext. Later, Kalai, Khurana, and Raizes gave a generic compiler to add PoNIs to any encryption scheme with publicly-verifiable certified deletion, assuming LWE (and QROM for decisional security). See \Cref{sec:overview:poni-dec}.

\paragraph{Classical Upload for Secure Key Leasing.}
Secure key leasing (SKL)~\cite{EC:AKNYY23,TCC:AnaPorVai23,EC:BGKMRR24,TCC:AnaHuHua24} allows a lessor to temporarily loan out a decryption key to a leasee as a quantum state. Later, the key can be verifiably revoked so that the leasee cannot decrypt future ciphertexts. Recently, \cite{EC:CGJL25,CRYPTO:KLYY26} showed constructions which allow the lessor to be entirely classical, ``uploading'' the quantum key using classical interaction.

SKL with classical upload also gives a natural approach to semi-quantum encryption with certified deletion, although it comes with some significant restrictions. 
To encrypt a message, the sender can lease a fresh key $(\ket{\sk_{\mathsf{SKL}}}, \pk_{\mathsf{SKL}})$, then encrypt their message under the SKL public key to obtain $\ct \gets \Enc(\pk_{\mathsf{SKL}}, m)$,{} and finally encrypt $\ct_{\mathsf{SKL}}$ under the recipient's public key to get $\ct'$. Intuitively, once $\ket{\sk_{\mathsf{SKL}}}$ is revoked, the encapsulated ciphertext $\ct_{\mathsf{SKL}}$ becomes un-decryptable and so the leakage of the recipient's decryption key does not matter.

Although we believe this is secure under certain circumstances, it comes with some limitations. First, the recipient \emph{must} use a PKE scheme which allows equivocation of ciphertexts (receiver non-committing). This is because SKL only provides security for ciphertexts generated \emph{after} revocation, whereas encryption with certified deletion generates the ciphertext \emph{before} deletion. Thus, the recipient's PKE must be able to equivocate the existing $\ct'$ to reveal an encryption of the challenge ciphertext which is generated later. Equivocation also precludes fully-homomorphic encryption~\cite{PKC:KatThiZho13}.
Second, this approach cannot achieve publicly verifiable certified deletion unless the SKL scheme supports publicly verifiable revocation. Any such SKL would immediately imply public-key quantum money, which is currently only known from strong assumptions such as indistinguishability obfuscation.

\paragraph{Classical Upload for Deletable NIZKs.}
In a non-interactive zero-knowledge argument (NIZK) with certified deletion, a prover can prepare a quantum NIZK which can later be deleted by the verifier. After producing a valid certificate, the verifier cannot produce another NIZK for the same statement without knowing the witness themselves. \cite{AC:AbbKat25} proposed this notion and also showed how a classical prover can ``send'' a quantum NIZK to the verifier using classical interaction. Certified deletion for NIZKs is a search-style security game, as opposed to the indistinguishability-style required for encryption.

\paragraph{Classical Upload for Copy Protected Programs.}
\cite{TCC:CheHerVu23} dequantized several copy-protection constructions which rely on coset states, assuming the sub-exponential hardness of LWE and the sub-exponential hardness of indistinguishability obfuscation. 
Their technical centerpiece is a remote state preparation protocol for coset states which ensures two search-style security notions for the constructed states: coset state monogamy of entanglement and direct product hardness.

\section{Preliminaries}

\paragraph{Cryptographic Notation.} We denote that two distributions $D_0$ and $D_1$ are \textbf{computationally indistinguishable} against $\QPT$ algorithms by $D_0 \approx_c D_1$. We denote \textbf{statistical indistinguishability} against unbounded quantum adversaries by $D_0\approx_s D_1$.
A distribution $D(\bullet)$ is \textbf{semantically hiding} if $D(z) \approx_c D(z')$ for all $z$ and $z'$.
We denote that two distributions are (computationally) distinguishable with at most $\epsilon$ advantage by $D_0 \overset{\epsilon}{\approx} D_1$.

\paragraph{Quantum Notation.} A \textbf{unitary dilation} of a quantum algorithm $Q$ acting on registers $\calQ$ is a unitary $U$ acting on register $\calQ$ and an ancilla register $\calA$ such that applying $U$ with $\calA$ initialized to $\ket{0}$ implements $Q$:
\[
    Q(\rho) = \Tr_\calA\left[U\left(\rho_{\calQ} \otimes \ketbra{0}_{\calA} \right)U^\dagger\right].
\]

\subsection{Trapdoor Claw-Free Functions}

We recall the definition of trapdoor claw-free functions using the syntax from \cite{DBLP:journals/iacr/KalaiKR26}.

\begin{definition}[Dual-Mode Trapdoor Function Family.]\label{def:dual-tcf}
    A \textbf{dual-mode trapdoor function family} is a family of (potentially randomized) functions
    \[
        \calF = \{f_{\pubparams} : \calX \times \calR \rightarrow \calY\}_{\pubparams}
    \]
    satisfying the following properties.\footnote{Formally, $\calX, \calR, \calY$ all depend on the security parameter with respect to which $\pubparams$ was generated.  This is omitted from the notation for the sake of simplicity.}
    \begin{itemize}
        \item \textbf{Efficient Generation.} There exists a probabilistic poly-time algorithm $\Setup$ distribution, which takes as input security parameter $1^\secpar$ along with a flag $\mode \in \{\lossy, \injective\}$, and outputs 
        \[
            (\pubparams, \td) \gets \Setup(1^\secpar, \mathsf{mode}).
        \]
        
        \item \textbf{Dual-Mode Indistinguishability.}
        \[
            \{\pubparams: (\pubparams, \td) \gets \Setup(1^\secpar, \injective)\}
            \approx_c
            \{\pubparams: (\pubparams, \td) \gets \Setup(1^\secpar, \lossy)\}
        \]

        \item {\textbf{Efficient Evaluation and Range Superposition.}}  There is a poly-time computable evaluation function that given any $\pubparams$, in the image of  $\Setup$ and given any pair $(x,r)\in(\calX\times \calR)$ outputs $f_\pubparams(x,r)$. 
        
        Furthermore, $f_\pubparams$ is associated with a distribution $\calD_{\pubparams}$ over randomness $\calR$ such that the state $\ket{\rand_\pubparams} = \sum_{r\in \calR}\sqrt{\calD_\pubparams(r)}\ket{r}$ is efficiently preparable.
        This state is a superposition over randomness used to evaluate $f_\pubparams$.

        \item \textbf{Trapdoor Inversion.}
        There exists a poly-time algorithm $\Inv: \calX\cup \{\bot\} \times \calY \mapsto \calR$ such that 
        for all $\mode \in \{\lossy, \injective\}$, all $(\pubparams, \td) \in \mathsf{SUPP}(\Setup(1^\secpar, \mode))${} and all input/randomness pairs $(x, r) \in \calX \times \calR$, $\Inv(\td, x, y)$ finds the $r$ which $y$ associates with $x$:
        \footnote{Note that this property implies that with overwhelming probability over $(\pubparams,\td) \gets \Setup(1^\secpar,\mode)$ and $r \gets \mathcal{D}_\pubparams$, $r$ has no siblings $r'$ such that $f_\pubparams(x;r) = f_\pubparams(x;r')$. That is, $f_\pubparams(x;\cdot)$ is effectively injective.}
        \[
            \Pr[(x, r) \gets \Inv(\td,x, f_\pubparams(x;r))] = 1
        \]
        Moreover, if $\mode = \injective$, then $\Inv(\td, \bot, y)$ also finds a unique $x$:
        \[
            \Pr[(x, r) \gets \Inv(\td,\bot, f_\pubparams(x;r))] = 1
        \]

    \end{itemize}
\end{definition}

\begin{definition}[Adaptive Hardcore Bits] \label{def:adapative-hcb}
    A dual-mode trapdoor function family for $\calX = \{0,1\}$ has \textbf{adaptive hardcore bits} if there exists a $\PPT$ algorithm $\Good((r_0, r_1), d)$ deciding membership in a set $\Good_{r_0, r_1}$ such that 
    \[
        \Pr_{d\gets \{0,1\}^n}[d \not\in \Good_{r_0, r_1}] = \negl
    \]
and for all $\QPT$ algorithms $\adv$,
    \[
        \left|
        \operatornamewithlimits{\Pr}_{\substack{(\pubparams, \td) \gets \Setup(1^\secpar, \lossy)  \\
            ((x, r), y, d, v) \gets \adv(\pubparams)}}\left[
        \begin{array}{cc}
             &f_{\pubparams}(x;r) = y  \\
             \land & d\in \Good_{r_0, r_1}\\
             \land & d\cdot (r_0 \oplus r_1) = v\\
        \end{array}
        \right]
        -
        \operatornamewithlimits{\Pr}_{\substack{(\pubparams, \td) \gets \Setup(1^\secpar, \lossy)  \\
            ((x, r), y, d, v) \gets \adv(\pubparams)}}\left[
        \begin{array}{cc}
             &f_{\pubparams}(x;r) = y  \\
             \land & d\in \Good_{r_0, r_1}\\
             \land & d\cdot (r_0 \oplus r_1) \neq v\\
        \end{array}
        \right]
        \right| \leq \negl
    \]
where $(0, r_0)$ and $(1, r_1)$ are the two preimages of $y$.
    
\end{definition}

It will be useful to deal with TCFs where the probability of getting output $y$ is roughly equivalent for all inputs $x$, over the choice of randomness. This ensures that evaluating in superposition and measuring the output $y$ yields roughly even amplitudes on the two preimages, for example $\frac{1}{\sqrt{2}}\left(\ket{x, r_x} + \ket{x', r_x'}\right)$, rather than having a higher amplitude on one of them.

\begin{definition}[Balance]\label{def:tcf-balance}
    A trapdoor-claw-free function is balanced if there exists a $\nu(\secpar) = \negl[\secpar]$ such that for any $x, x' \in \calX$, 
    \[
        \Pr\left[ 
            \begin{array}{c}
                 \exists !\ r' \text{ s.t. } f_\pubparams(x'; r') = y  \\
                 \wedge \ \Big|
                 \frac{\mathcal{D}_\pubparams(r)- \calD_\pubparams(r')}{\mathcal{D}_\pubparams(r)  + \mathcal{D}_\pubparams(r')} 
                 \Big| \leq \nu(\secpar) 
            \end{array}
            :
            \begin{array}{c}
                 (\pubparams,\td) \gets \Setup(1^\secpar, \lossy) 
                 \\ r \gets \mathcal{D}_\pubparams 
                 \\ y \coloneqq f_\pubparams(x;r)
            \end{array}
        \right] \geq 1-\negl[\secpar]
    \]
\end{definition}

\begin{theorem}\cite{BCMVV18,Mahadev18}\label{thm:dual-tcf}
    Assuming the quantum hardness of LWE, there exists a (balanced) dual-mode claw-free trapdoor function family with \textbf{adaptive hard-core bits}.
\end{theorem}

\begin{definition}[Collapsing]
    Let $\calF = \{f_{\pubparams}:\calX \times \calR \rightarrow \calY\}_\pubparams$ be a function family associated with a sample $\Setup$. We say that it is collapsing if every $\QPT$ adversary $\adv$ wins in the following game with probablity $\leq 1/2+\negl$:
    \begin{enumerate}
        \item \textbf{Challenger:} Sample $\pubparams \gets \Setup(1^\secpar)$ and sample a random bit $b\gets\{0,1\}$.
        
        Send $\pubparams$ to $\adv$.
        \item $\adv$ sends back a classical value $y\in \calY$ and a register $\calS$.
        \item \textbf{Challenger:} Measure $\calS$ with respect to the projective measurement $\{\Pi_y, I - \Pi_y\}$ where $\Pi_y$ projects onto $\ket{x, r} $ such that $f_\pubparams(x;r) = y$. If the result is $I - \Pi_y$ then $\calA$ needs to guess $b$ (without being given any information about~$b$), and the game is over.
        \item \textbf{Challenger:} If $b = 1$, measure register $\calS$ in the computational basis. Then, regardless of $b$, send $\calS$ back to $\adv$.
        \item $\adv$ outputs a bit $b'$ and wins if $b' = b$.
    \end{enumerate}
\end{definition}

\begin{lemma}\label{lem:dual-mode-tcfs-are-collapsing}
    Any dual-mode trapdoor function family $\calF = \{f_{\pubparams}\}_{\pubparams \in {\sf Supp}(\Setup(1^\secpar, \lossy)}${} is collapsing.
    Moreover, $\lossy$ mode is claw-free.
\end{lemma}
\begin{proof}
    Consider the hybrid game where $\pubparams$ is generated in injective mode instead. Note that the success probability in this hybrid game is negligibly close to its success probability in the actual game since since the challenger's role can be played using an externally sampled $\pubparams$, without the use of $\td$. 
    Furthermore, in the hybrid game the success probability is exactly $\frac12$ since $\Pi_{y}$ is a projection onto a computational basis vector when $\pubparams$ is in injective mode. 

    Moreover, in injective mode there do not exist claws, so any algorithm which finds one in $\lossy$ mode immediately distinguishes, violating the indistinguishability of the two modes.
\end{proof}

\subsection{Witness-Preserving Arguments for NP} \label{sec:witness-pres-np}

\cite{DBLP:journals/iacr/KalaiKR26} showed how a quantum prover holding a superposition $\sum_w \alpha_w \ket{w}$ of NP witnesses for $x\in \lang$ can convince a classical verifier that $x\in \lang$ \emph{without disturbing} their witness superposition. We recall the notion and their result in this section. Then, in the next section we observe that techniques from \cite{FOCS:LomMaSpo22} can make the protocol ``indistinguishably-extractable''; that is, the prover's view can be indistinguishably simulated in such a way that the witness is temporarily extracted, then returned to the prover to obtain the final view.

\begin{definition}\cite{DBLP:journals/iacr/KalaiKR26}\label{def:state-preserving-NP}
    A \emph{witness-preserving argument} for an $\NP$ language $\lang$ is an interactive protocol between a $\QPT$ prover $P$ and a ${\sf PPT}$ verifier $V$ that satisfies the following guarantees.

    \begin{itemize}
        \item \textbf{Completeness and Witness Preservation.} 
        
        For every $x\in \lang$ and every superposition $\ket{\psi} = \sum_{w} \alpha_{w} \ket{w}_{\calW}\otimes \ket{\phi_w}_{\calB}$ over witnesses for $x$ (encoded in the computational basis and potentially entangled with an external state $\ket{\phi_w}$), there exist negligible functions $\mu_1$ and $\mu_2$ such that for every $\secpar\in\mathbb{N}$,
        \[
\Pr\left[
                \left(\frac{1}{2}\big\|\ketbra{\psi} - \rho \big\|_1 \leq \mu_1(\secpar)\right)
                ~~\land~~ \left(b = \Accept\right):~
                (\rho, b) \gets \langle P(\ket{\psi}), V\rangle(1^\secpar, x)
            \right] \geq 1- \mu_2(\secpar),
\]
        where $P$ acts only on register $\calW$ and where the notation $(\rho, b) \gets \langle P(\ket{\psi}), V\rangle(1^\secpar, x)${} means that $V$ outputs $b\in\{\Accept, \Reject\}$ and $P$'s residual state is $\rho$.

        When we wish to be more precise about the repair guarantees, we say that the argument is $(1-\mu_1)$-witness-preserving.
       
        \item \textbf{Computational Soundness.} 
        For every polynomial $\ell=\ell(\secpar)$, there exists a negligible function $\mu$ such that for every $\secpar \in \bbN$, every $x\notin \lang$ of size $\leq \ell(\secpar)$ and every $\QPT$ cheating prover $P^*$ with auxiliary input $\ket{\psi}$ consisting of at most ${\sf poly}(\ell(\secpar))$ qubits
        \[
            \Pr[\langle P^*(\ket{\psi}, V\rangle(1^\secpar, x) = \Accept] \leq \mu(\secpar)
        \]

    \end{itemize}
\end{definition}

\begin{theorem}\label{thm:witness-preserving-NP}\cite{DBLP:journals/iacr/KalaiKR26}
    There exists a  $(1-2^{-\secpar})$-witness-preserving argument for $\NP$, assuming the existence of dual-mode claw-free trapdoor function family with state recovery (which in turn can be constructed from the post-quantum hardness of {\sf LWE}~\cite{BCMVV18}).
\end{theorem}

\subsection{Quantum Extraction for Arguments of Classical Knowledge}
\label{sec:prelim:extraction}

We next define the notion of indistinguishable-extraction, which is a variation on \cite{FOCS:LomMaSpo22}'s ``state-preserving extraction''.\footnote{We use the ``indistinguishably-extractable'' to differentiate from \cite{DBLP:journals/iacr/KalaiKR26}'s ``witness-preserving'' property which refers to preserving the prover's state in a \emph{real} execution.} State-preserving extraction allows generating an indistinguishable transcript and adversarial state while simultaneously extracting a witness. However, obtaining the extracted witness and transcript simultaneously necessarily means collapsing the prover's witness, which is immediately distinguishable in a witness-preserving argument.

In indistinguishable-extraction, we deliberately avoid collapsing the witness in this way. Instead, we build an interface which temporarily exposes a coherent witness, then ``returns'' the witness to the adversary without collapsing it. The resulting interface is quite similar to \cite{FOCS:LomMaSpo22}'s guaranteed extraction: there is a unitary extractor $\Ext$ which uses a superposition over accepting transcripts to extract a valid witness with overwhelming probability. The difference from guaranteed extraction is that we require $\Ext$ to be \emph{unitary}. Thus, it may be inverted afterwards, returning the witness.

In what follows, we denote by  
\[
{\sf View}_{P^*}\langle {P^*}(\ket{\psi}, V\rangle(1^\secpar, x)
\]
the view of an adversarial prover $P^*$ after interacting in the protocol $\langle {P^*}(\ket{\psi}), V\rangle(1^\secpar, x)$,{} which includes the transcript in register $\calT$ as well as the residual state of $P^*$ in register $\calA$.

\begin{definition}[Indistinguishably-Extractable Arguments]\label{def:indist-ext}
    An NP argument is said to be \emph{indistinguishably-extractable} if there exist a QPT algorithm $\Sim$ and uniform family of $\poly[\secpar]$-sized unitaries $\{\Ext_\secpar\}_{\secpar\in \bbN}$ satisfying the following for every sufficiently large security parameter $\secpar\in \bbN$, every adversarial QPT prover $P^*$ and every $\NP$ statement $x$:
    \begin{itemize}
        \item \textbf{Syntax.} $\Sim(P^*)$ outputs three registers $(\calT, \calA, \calS)$ and a decision bit $b$ denoting whether $\calT$ contains an accepting transcript. $\calT$ contains a superposition over classical transcripts $T$, $\calA$ contains the adversary's state, and $\calS$ is a work register for the simulator.\footnote{$\Sim$ does not receive restricted information like in a zero-knowledge argument. For example, $\Sim$ could simply run the argument coherently and measure whether the verifier would accept. We also allow other strategies for generating transcripts.}

        $\Ext$ operates on registers $(\calT, \calA, \calS)$ and reorganizes them into registers $(\calW, \calE)$. $\calW$ contains the extracted witness and $\calE$ is a work register.
        
        \item \textbf{Simulation.} $\Sim$ produces indistinguishable adversarial views given oracle access to the prover $P^*$:
        \[
            \Tr_{\calS}\left[\ket{\phi}_{\calT,\calA, \calS}\right] \approx {\sf View}_{P^*}\langle {P^*}(\ket{\psi}), V\rangle(1^\secpar, x)
        \]
        where $(\ket{\phi}_{\calT, \calA, \calS}, b)\gets \Sim^{P^*} (1^\secpar,x,\ket{\psi})$. If the indistinguishability is computational (resp. statistical), we say the argument is computationally (resp. statistically) indistinguishably-extractable.
        
        \item \textbf{Extraction.} Let $\Pi_{\Ext, \secpar} = \Ext_\secpar^\dagger \Pi_\Ver \Ext_{\secpar}$ be the projector onto $\Ext_\secpar$ successfully extracting an accepting witness for $x$ to register $\calW$.
        Consider the following experiment $\IndExt_{\secpar}$:
        \begin{enumerate}
            \item Run $(\ket{\phi}_{\calT, \calA, \calS}, b_{\Sim}) \gets \Sim^{P^*}$.
            \item If $b_\Sim = \Accept$: measure $\ket{\phi}_{\calT, \calA, \calS}$ according to $\{\Pi_{\Ext, \secpar}, I - \Pi_{\Ext, \secpar}\}$. Call the measurement result (representing successful extraction) $b_\Ext$.
        \end{enumerate}
        Then the probability of generating an accepting transcript but failing to extract a valid witness is negligible:
        \[
            \Pr_{\IndExt_{\secpar}}[b_\Sim = \Accept \land b_\Ext = \Reject] = \negl
        \]
    \end{itemize}
\end{definition}

It will also be useful to define an approximate version, where the extractor is allowed to take $\poly[1/\epsilon]$ time to extract a witness with probability $1-\epsilon$, for every $\epsilon$.

\begin{definition}[$\epsilon$-Indistinguishable-Extractability]\label{def:eps-indist-ext}
    We say an argument is $\epsilon$-Indistinguishably-Extractable if for every uniform sequence of precisions $\{\epsilon(n) \in (0, 1]\}_{n\in \bbN}$,
    there exists a QPT $\Sim$ and uniform family of $\poly[\secpar, 1/\epsilon(n)]$-sized unitaries $\{\{\Ext_{\lambda, \epsilon(n)}\}_{n\in \bbN}\}_{\secpar\in \bbN}$ such that
    \begin{itemize}
        \item \textbf{Simulation.} $\Sim^{P^*}$ produces indistinguishable adversarial views.
        \item \textbf{$\epsilon$-Extraction.} Let $\IndExt_{\secpar, \epsilon(n)}$ the be experiment from \Cref{def:indist-ext} using $\Ext_{\secpar, \epsilon(n)}$. For all $n\in \bbN$ and sufficiently large $\secpar\in \bbN$, the probability of a successful transcript but unsuccessful extraction is at most $\epsilon$:
        \[
            \Pr_{\IndExt_{\secpar, \epsilon(n)}}[b_\Sim = \Accept \land b_\Ext = \Reject] \leq \epsilon(n)
        \]
    \end{itemize}
\end{definition}

The almost-as-good-as-new lemma~\cite{TOC:Aar05} (closely related to the gentle measurement lemma~\cite{TIT:Win99}) limits the damage that indistinguishable extraction can do to a simulated transcript.

\begin{lemma}\label{lem:iext-close}
    For all adversarial provers $P^*$,
    \[
        \frac{1}{2}\left\| \Sim(P^*) - \IndExt_{\secpar, \epsilon(n)} \right\|_1 \leq \sqrt{\epsilon(n)}
    \]
    Combining this with the indistinguishability of simulation,
    \[
        \Tr_{\calS}\left[\IndExt_{\secpar, \epsilon(n)^2}\right] \overset{\epsilon(n)}{\approx} {\sf View}_{P^*}\langle {P^*}(\ket{\psi}), V\rangle(1^\secpar, x)
    \]
    Moreover, if $\IndExt_{\secpar, \epsilon(n)^2}$ behaves arbitrarily in the case of $b_\Sim = \Accept \land b_\Ext = \Reject$, the distance is at most $\epsilon(n) + \epsilon(n)^2$.
\end{lemma}

Lombardi, Ma, and Spooner show a general class of protocols which have guaranteed extraction (i.e. $\Ext$ is allowed to be an expected QPT algorithm instead of a unitary).
Although the exact conditions are somewhat complex, they specifically list \cite{FOCS:GolMicWig86}'s 3-coloring protocol as satisfying these properties, with a small tweak: commit to the witness first, then run a GMW argument which is amplified in parallel~\cite[Section 5.3 and Corollary 8.1]{FOCS:LomMaSpo22}. 
Moreover, Lombardi, Ma, and Spooner show that there exists series of unitary dilations $\{U_{\epsilon(n)}\}_{n\in \bbN}$ which $\epsilon$-approximate the guaranteed extractor~\cite[Claim 9.4]{FOCS:LomMaSpo22}.\footnote{We briefly mention that Lai, Spooner, and Tromanhauser~\cite{eprint:LST26} recently pointed out a bug in \cite{FOCS:LomMaSpo22}'s definition of coherent-runtime simulation. However, the other results are intact, including guaranteed extraction and the specific unitary dilation used to instantiate the faulty coherent-runtime simulators.}

\begin{theorem}
    For every (computationally or statistically) collapsing commitment $\Com$, the following protocol is statistically $\epsilon$-indistinguishably-extractable.
    \begin{enumerate}
        \item Commit to the witness $w$ using $\Com$.
        \item Prove knowledge of an opening to a witness $w$ using \cite{FOCS:GolMicWig86}'s 3-coloring protocol, instantiated with $\Com$.
    \end{enumerate}
\end{theorem}

Finally, we observe that this theorem can be combined with \cite{DBLP:journals/iacr/KalaiKR26}'s construction of \emph{witness}-preserving arguments. Their construction is a parallel repetition of \cite{FOCS:GolMicWig86}'s 3-coloring protocol using a specially chosen commitment, followed by a recovery procedure. An additional commitment and recovery step can be added at the start and end of their protocol. 
The following corollary is obtained by applying the above theorem to the portion of the protocol before recovery and observing that the recovery step afterwards cannot increase the \emph{statistical} distance.\footnote{If the simulation indistinguishability were computational, it is possible that revealing additional information during recovery allows breaking the indistinguishability.}

\begin{corollary}\label{coro:wp-np-eps-iext}
    There exists a  witness-preserving argument for $\NP$ with approximate indistinguishable-extraction, assuming the existence of dual-mode claw-free trapdoor function family with state recovery (which in turn can be constructed from the post-quantum hardness of {\sf LWE}~\cite{BCMVV18}).
\end{corollary}

\section{Phase Injection Lemma}\label{sec:plug-in-lemma}

In this section, we prove a lemma that will be useful for proving semi-quantum certified deletion.
At a high level, the lemma considers a scenario where an adversary evaluates a trapdoor claw-free function $f_\pubparams(b,x_b)$ while knowing a superposition over preimages $(b, x_b)$. Later, they are allowed to ask for an inner product $d\cdot (x_0 \oplus x_b)$ with the XOR of the two preimages. The following lemma shows that when $f_\pubparams$ has adaptive hardcore bits, the adversary cannot detect the insertion of a phase $p$ into its initial superposition over preimages,
so long as the inner product revealed at the end is also adjusted by $p$.

\begin{lemma}[Phase Injection]\label{lem:phase-injection}
    Let $\calF$ be a dual-mode trapdoor claw-free function with adaptive hardcore bits. Then every  $\QPT$ adversary $\adv$ wins in the following game with probability at most $1/2+\negl$.
    \begin{enumerate}
        \item \textbf{Challenger:} Sample $(\pubparams, \td) \gets \KeyGen(1^\secpar, \lossy)$ and sample a random bit $p\gets\{0,1\}$. Send $\pubparams$ to $\adv$.
        \item $\adv$ sends back a classical value $y$ and a register $\calR$ (which should store a preimage of $y$).
        \item \textbf{Challenger:} Perform the projective measurement $\{\Pi_y, I - \Pi_y\}$ on $\calR$ where $\Pi_y$ projects onto $(b, r_b)$ such that $f_{\pubparams}(b;r_b) = y$. If the result is $I - \Pi_y$, then $\calA$ needs to guess $p$ (without being given any information about~$p$), and the game is over.
        \item \textbf{Challenger:} Otherwise, if $p=1$ then apply a $Z$ gate to the first qubit of $\calR$.  
        
        Send $\calR$ back to $\adv$.
        \label{step:phase-inj-phase}
        \item $\adv$ sends a string $d$ to the challenger. \label{step:phase-inj-get-d}
        \item \textbf{Challenger:} Compute $(0, r_0) \gets \Inv(\td,0, y)$ and $(1, r_1) \gets \Inv(\td,1, y)$. 
        If $d\in \Good_{r_0, r_1}$, send $(d\cdot (r_0 \oplus r_1)) \oplus p$ to $\adv$. Otherwise send a random bit $b'$.

        \label{step:phase-inj-inner-product}
        \item $\adv$ outputs a bit $p'$.
    \end{enumerate}
   $\adv$ wins if and only if $p'=p$.
\end{lemma}

\begin{proof}
    We define four hybrid experiments $\Hyb_{ab}$ for $a,b\in \{0,1\}$. $\Hyb_{ab}$ is the same experiment described in the lemma statement, except for two modifications. First, apply $Z^a$ in step \ref{step:phase-inj-phase}. Second, send $d\cdot (r_0\oplus r_1) \oplus b$ in \ref{step:phase-inj-inner-product}. Observe that the original game with bit $p$ is exactly $\Hyb_{pp}$. 

    Consider the distributions $\Hyb_{a*} = 1/2(\Hyb_{a0} + \Hyb_{a1})$ for $a\in \{0,1\}$.
    These are the distributions resulting from inserting a fixed phase $a$ in step \ref{step:phase-inj-phase}, but sending a uniformly random bit in step \ref{step:phase-inj-inner-product} (\emph{regardless of whether $d\in \Good_{(0,r_0), (1, r_1)}$}).
    \begin{claim}
         \begin{equation}\label{eq:phase-inj-hyb-phase}
        \left|\Pr_{\Hyb_{0*}}[p' = 1] - \Pr_{\Hyb_{1*}}[p' = 1]\right| = \negl
        \end{equation}
    \end{claim}
    \begin{proof}
        Define the sub-hybrid $\Hyb'_{a*}$ where register $\calR$ is \emph{measured} in the computational basis in step \ref{step:phase-inj-phase}. Neither $\Hyb_{a*}$ needs the trapdoor since step \ref{step:phase-inj-inner-product} always reveals a random bit regardless of whether $d\in \Good_{(0,r_0), (1, r_1)}$, so $\Hyb_{a*}$ is indistinguishable from $\Hyb'_{a*}$ by the collapsing property of dual-mode trapdoor functions (\Cref{lem:dual-mode-tcfs-are-collapsing}). Furthermore, $\Hyb'_{0*} = \Hyb'_{1*}$ since the phase is destroyed by measuring in the computational basis.
    \end{proof}
   
    Now consider the distributions $\Hyb_{*b} = 1/2(\Hyb_{0b} + \Hyb_{1b})$ for $b\in \{0,1\}$.
    These are the distributions resulting from inserting a random phase in step \ref{step:phase-inj-phase}, but sending $d\cdot (r_0 \oplus r_1) \oplus b$ for a fixed $b$ in step $\ref{step:phase-inj-inner-product}$.
    \begin{claim}
        \begin{equation}\label{eq:phase-inj-hyb-inner-prod}
            \left|\Pr_{\Hyb_{*0}}[p' = 1] - \Pr_{\Hyb_{*1}}[p' = 1]\right| = \negl
        \end{equation}
    \end{claim}
    \begin{proof}
    Conditioned on the $\{\Pi_y, I- \Pi_y\}$ measurement returning $\Pi_y$, the adversary's view after receiving register $\calR$ following step \ref{step:phase-inj-phase} is an equal mixed state over
    \begin{align*}
        &\alpha \ket{\phi_0} \otimes \ket{0, r_0}_{\calR} + \beta \ket{\phi_1} \otimes \ket{1, r_1}_{\calR}
        \\
        \text{and}\quad &\alpha \ket{\phi_0} \otimes \ket{0, r_0}_{\calR} - \beta \ket{\phi_1} \otimes \ket{1, r_1}_{\calR}
    \end{align*}
    This is precisely the mixed state 
    \[
        |\alpha|^2 \ketbra{\phi_0} \otimes \ketbra{0, r_0}_{\calR} + |\beta|^2 \ketbra{\phi_1} \otimes \ketbra{1, r_1}_{\calR} 
    \]
    Thus, it is equivalent in the adversary's view to \emph{measure} register $\calR$ in the computational basis in step \ref{step:phase-inj-phase}. 
    
    Now suppose $\Hyb_{*0}$ were distinguishable from $\Hyb_{*1}$, with $\calR$ measurement, with advantage $\epsilon$.We construct an adversary for the adaptive hardcore bit property with advantage $\epsilon$. Takes as input public parameters $\pubparams$, then use them to run $b'\gets \Hyb_{*0}$, except report a random bit $b$ instead instead of the inner product $d\cdot(r_0 \oplus r_1)$. As part of this, obtain $x$ from measuring $\calR$ and $d$ from $\calA$ in step \ref{step:phase-inj-get-d} and obtain $(x, y)$ such that $f_\pubparams(x) = y$ from the measurement of $\calR$. Output $((x,y), (d, b\oplus b'))$.
    
    Observe that if $b = d\cdot (r_0 \oplus r_1)$ and $d$ is good, then this corresponds to $H_{*0}$; if $b \neq d\cdot (r_0 \oplus r_1)$ and $d$ is good then it corresponds to $H_{*,1}$. If $d$ is not good, it is identically distributed to both.
    Since the adversary successfully guesses $b' = b\oplus (d\cdot (r_0 \oplus r_1))$ with $\epsilon$ advantage and their advantage conditioned on invalid $x$ is $0$, the reduction has advantage $\epsilon$ in guessing a valid $x$ and correct inner product $(x, d, d\cdot (r_0 \oplus r_1))$ versus guessing a valid $x$ and incorrect inner product $(x, d, 1\oplus (d\cdot (r_0 \oplus r_1)))$, violating the adaptive hardcore bit property.
    \end{proof}

    Finally, we show that $\adv$ guesses $1$ with almost the same probability in both the $p=0$ and $p=1$ case by combining \Cref{eq:phase-inj-hyb-phase} with \Cref{eq:phase-inj-hyb-inner-prod}.
\begin{align*}
        &\left|\Pr_{\Hyb_{00}}[p' = 1] - \Pr_{\Hyb_{11}}[p' = 1]\right|
        \\
        &= \left|\begin{array}{cc}
            &1/2\left(\Pr_{\Hyb_{00}}[p' = 1] + \Pr_{\Hyb_{01}}[p' = 1]\right) + 1/2\left(\Pr_{\Hyb_{00}}[p' = 1]) + \Pr_{\Hyb_{10}}[p' = 1]\right) 
            \\
            - &1/2\left(\Pr_{\Hyb_{11}}[p' = 1] + \Pr_{\Hyb_{01}}[p' = 1]\right) - 1/2\left(\Pr_{\Hyb_{11}}[p' = 1] + \Pr_{\Hyb_{10}}[p' = 1]\right) 
            \end{array}
            \right|
        \\
        &\leq \left|\Pr_{\Hyb_{0*}}[p' = 1]- \Pr_{\Hyb_{1*}}[p' = 1])\right| + \left|\Pr_{\Hyb_{*0}}[p' = 1] - \Pr_{\Hyb_{*1}}[p' = 1]\right|
        \\
       & =\negl
    \end{align*}
\end{proof}

\section{Semi-Quantum Secret Sharing with Certified Deletion}\label{sec:sqss}

In this section, we give the core construction for semi-quantum cryptography with certified deniability. At its core, the construction generates a classical ``share'' $\csh$ and a quantum ``share'' $\ket{\qsh}$ of a random bit $b$, so that the quantum share may be certifiably deleted, or both shares can be combined to recover $b$.
This already closely matches the syntax of secret sharing, so we abstract it as such.

\subsection{Definition}
\begin{definition}\label{def:SQSS}
A 2-out-of-2 semi-quantum secret sharing scheme with (publicly verifiable) certified deletion consists of the following ingredients:
\begin{itemize}
    \item 
    $\Share\langle S(m), R\rangle(1^\secpar) \rightarrow ((\verkey, \csh), \ket{\qsh})$ is an interactive protocol between a classical ${\sf PPT}$ sender $S$ and a quantum $\QPT$ receiver $R$. The sender inputs a secret message $m \in \{0,1\}$ and both parties input the security parameter $1^\secpar$.
    At the end of the protocol, the sender outputs a tuple $(\verkey, \csh)$ consisting of the deletion verification key $\verkey$ and a classical share $\csh$. The receiver outputs a quantum share $\ket{\qsh}$.

    \item 
    $\Reconstruct(\csh, \ket{\qsh}) \rightarrow m$ is a $\QPT$ algorithm which takes as input a classical share $\csh$ and a quantum share $\ket{\qsh}$, then outputs a classical message $m$.

    \item 
    $\Del(\ket{\qsh}) \rightarrow \cert$ is a $\QPT$ algorithm which takes as input a quantum share $\ket{\qsh}$ then outputs a classical certificate.\footnote{A quantum certificate is less desirable for semi-quantum schemes because the other communication is classical.} 

    \item 
    $\Ver(\verkey, \cert)\rightarrow \Accept/\Reject$ is a $\PPT$ algorithm which takes as input a verification key $\verkey$ and a certificate $\cert$, then outputs accept or reject.
\end{itemize}
It satisfies the following properties:
\begin{itemize}
    \item {\bf Correctness.}  
    There exists a negligible function $\mu$ such that for every $m\in\{0,1\}$ and every $\secpar\in\mathbb{N}$, reconstruction outputs the shared message with overwhelming probability:
    \[
\Pr\left[m \gets \Reconstruct(\csh, \ket{\qsh}): ((\verkey, \csh), \ket{\qsh}) \gets \Share\langle S(m), R\rangle(1^\secpar)\right]
        =1-\mu(\secpar)
\]

    Moreover, there exists a negligible function $\mu$ such that for every $\secpar\in\mathbb{N}$, $\Ver$ accepts the certificates generated by $\Del$:
    \[
        \Pr\left[\Accept \gets \Ver(\verkey, \cert) : \begin{array}{c}
             ((\verkey, \csh), \ket{\qsh}) \gets \Share\langle S(m), R\rangle(1^\secpar) \\
             \cert \gets \Del(\ket{\qsh}) 
        \end{array}
        \right]
        = 1-\mu(\secpar)
    \]

    \item \textbf{Privacy.} 
    For every $\QPT$ adversarial receiver $R^*$ outputting a register $\calR_{\qsh}$,
    \begin{gather*}
        \left\{(\calR_{\qsh}, \verkey) : ((\verkey, \csh), \calR_{\qsh}) \gets \Share\langle S(0), R^*\rangle (1^\secpar) \right\}
        \\ \approx_c \\
        \left\{(\calR_{\qsh}, \verkey): ((\verkey, \csh), \calR_{\qsh}) \gets \Share\langle S(1), R^*\rangle (1^\secpar) \right\}
    \end{gather*}
    where $\calR$ is a register containing the adversary's state after the sharing protocol.\footnote{One could also ask that the classical share $\csh$ separately does not leak information about the secret $m$. We omit this property here because it is implied by the certified deletion property.}

     \item \textbf{(Publicly Verifiable) Deletion Security.} For every $\QPT$ adversary $(R^*, \Del^*)$,
     \[
        \CDGame(1^\secpar, 0, (R^*, \Del^*)) \approx_c \CDGame(1^\secpar, 1, (R^*, \Del^*))
     \]
     where the game $\CDGame(1^\secpar, m, (R^*, \Del^*))$ is played as follows:
     \begin{enumerate}
        \item Share $((\verkey, \csh), \calR_{\qsh}) \gets \Share\langle S(m), R^*\rangle(1^\secpar)$.
        \item Compute $(\calR_\Del, \calR_{\cert}) \gets \Del^*(\calR_{\qsh})$. In the publicly verifiable version, $\Del^*$ is also given $\verkey$.
        \item Output
        \[
            \begin{cases}
                (\calR_\Del, \csh) &\text{if } \Ver(\verkey, \calR_{\cert}) \text{ accepts}
                \\
                \bot
                &\text{otherwise}
            \end{cases}
        \]
     \end{enumerate}
  \end{itemize}
\end{definition}

\cite{AC:HMNY21} gave a definition which combines privacy with certified deletion by outputting $(\calR_\Del, \bot)$ when the certificate is invalid. This allows the adversary to guess without the secret key (similar to the privacy case) when they fail verification. Their definition is equivalent to the split definition above. The implication from combined to split is easy to see. For the other direction, observe that the adversary's advantage in the combined game is $\leq$ its advantage conditioned on passing verification plus its advantage conditioned on failing verification, weighted by the probability of each. Its advantage conditioned on passing directly translates to advantage in the split certified deletion game since the rejection case contributes $0$ advantage in the split game; its advantage conditioned on failing directly translates to advantage in the privacy game.

\paragraph{Auxiliary Information.}
It is also useful to be able to give out some information depending on $\csh$ before deletion. For example, imagine encrypting $\csh$ and giving the resulting ciphertext $
\widetilde{\csh}$ to the adversary. Since $\csh$ is still hidden without the decryption key, we might hope that $m$ remains protected even if the decryption key for $\widetilde{\csh}$ is revealed after deletion. 
The following definition formalizes this property.

\begin{definition}[$\calZ$-Deletion Security]
    Let $\calZ = \left(\calZ_1(1^\secpar),\ \calZ_2(1^\secpar,\ \bullet)\right)$ be a two-stage distribution using shared randomness.
    A semi-quantum secret sharing scheme has \textbf{(publicly verifiable) $\calZ$-deletion security} if for every $\QPT$ adversary $(R^*, \Del^*)$,
     \[
        \CDGame_\calZ(1^\secpar, 0, (R^*, \Del^*)) \approx_c \CDGame_\calZ(1^\secpar, 1, (R^*, \Del^*))
     \]
     where the game $\CDGame_\calZ(1^\secpar, m, (R^*, \Del^*))$ is played as follows:
     \begin{enumerate}
        \item Sample $\aux_1 \gets \calZ_1(1^\secpar; r)$ using randomness $r$.
        \item Share $((\verkey, \csh), \calR_{\qsh}) \gets \Share\langle S(m), R^*(\aux_1)\rangle(1^\secpar)$.
        \item Sample $\aux_2 \gets \calZ_2(1^\secpar, \csh; r)$ using the same randomness $r$ as in $\calZ_1$. Let $\aux = (\aux_1, \aux_2)$.
        \item Compute $(\calR_\Del, \calR_{\cert}) \gets \Del^*(\calR_{\Gen}, \aux)$. In the publicly verifiable version, $\Del^*$ is also given $\verkey$.
        \item Output
        \[
            \begin{cases}
                (\calR_\Del, \csh, r) &\text{if } \Ver(\verkey, \calR_{\cert}) \text{ accepts}
                \\
                \bot
                &\text{otherwise}
            \end{cases}
        \]
     \end{enumerate}
     For simplicity of notation, we often omit the security parameter input $1^\secpar$ from $\calZ_1$ and $\calZ_2$ when it is clearly implied.
\end{definition}
The above definition separates the auxiliary information into two stages so that an adversarial receiver $R^*$ may use the first stage $\aux_1$ during $\Share$. This allows more flexibility in applying the definition. For example, consider $\aux_1 = \pk$ to be a public key for a PKE scheme, and $\aux_2 = \Enc_\pk(\csh)$ to be an encryption of $\csh$ under that public key. This matches PKE security closely because there, the adversary is allowed to see $\pk$ before even choosing the message to be encrypted, let alone participating in the encryption protocol.

\begin{remark}
    A single-bit protocol can be extended to sharing $\ell$ bits by simply running $\Share$ on each bit of the message sequentially. In this case, privacy and certified deletion now quantify over all $m_0, m_1 \in \{0,1\}^{\ell}$ and replace $0$ by $m_0$ and $1$ by $m_1$ in the security games. 
    Security of the $\ell$-bit protocol follows from a standard hybrid argument which changes each bit of the message one at a time.
\end{remark}

\subsection{Construction}
The construction uses the following tools:
\begin{itemize}
    \item $\calF$ is a dual-mode trapdoor claw-free function family with adaptive hardcore bits (\Cref{def:adapative-hcb}). 
    \item $\SPNP$ is a witness-preserving argument of knowledge for $\NP$ with $\epsilon$-indistinguishable-extraction (\Cref{def:eps-indist-ext,def:state-preserving-NP}).
\end{itemize}
Both primitives exist assuming the post-quantum hardness of LWE (\Cref{thm:dual-tcf,coro:wp-np-eps-iext}).

We give the key generation algorithm in \Cref{constr:sq-cd-keygen}, the reconstruction algorithm in \Cref{constr:sq-cd-recon}, and the deletion and verification algorithms in \Cref{constr:sq-cd-del-ver}.

\begin{construction}[Sharing $\Share\langle S(m), R\rangle(1^\secpar)$]\label{constr:sq-cd-keygen}
    The classical sender and classical receiver interact as follows.
    \begin{enumerate}
        \item \textbf{Classical Sender.}
        Sample two sets of public parameters
        \begin{align*}
            (\pubparams_{A}, \td_{A}) &\gets \Setup_{\calF}(1^\secpar, \lossy)
            \\
            (\pubparams_{B}, \td_{B}) &\gets \Setup_{\calF}(1^\secpar, \lossy)
        \end{align*}
        Send $(\pubparams_{A}, \pubparams_{B})$ to the quantum receiver.

        \item \textbf{Quantum Receiver.} 
        Prepare the state 
        \[
\frac{1}{\sqrt{2}} \sum_{b\in \{0,1\}} \sum_{x^{(A)}, x^{(B)}} \alpha_{x^{(A)}} \alpha_{x^{(B)}}\ket{b}\otimes \Ket{x^{(A)}}_{\calA}\otimes \Ket{x^{(B)}}_{\calB} \otimes \Ket{f_{\pubparams_{A}}\left(b; x^{(A)}\right), f_{\pubparams_{B}}\left(b; x^{(B)}\right)}
\]
        by evaluating $f_{\pubparams_{A}}$ and $f_{\pubparams_{B}}$ on a $\ket{+}$ state together with a superposition over randomness for $\calF$. Then, measure the output register to get $y_{A}$, $y_{B}$ and the state
        \[
            \frac{1}{\sqrt{2}} \left(\ket{0, x^{(A)}_0, x^{(B)}_0}  + \ket{1, x^{(A)}_1, x^{(B)}_1}\right)
        \]
        where $f_{\pubparams_{A}}(0; x^{(A)}_0) = f_{\pubparams_{A}}(1; x^{(A)}_1) = y_{A}$ and $f_{\pubparams_{B}}(0; x^{(B)}_0) = f_{\pubparams_{B}}(1; x^{(B)}_1) = y_{B}$.

        Send $(y_A, y_B)$ to the classical party.

        \item \textbf{Both.} Define the language
        \[
            \lang_{\pubparams_A, \pubparams_B} 
            = 
            \left\{(y_A, y_B): 
                \exists (b, x^{(A)}, x^{(B)}) \text{ such that } \begin{array}{cc}
                     &f_{\pubparams_{A}}\left(b; x^{(A)}\right) = y_{A} \\
                    \land &f_{\pubparams_{B}}\left(b; x^{(B)}\right) = y_{B} 
                \end{array}
            \right\}
        \]

        Perform a witness-preserving NP argument of knowledge for the language $\lang_{\pubparams_A, \pubparams_B} $ and statement $(y_A, y_B)$, where the classical sender acts as the verifier and the quantum receiver acts as the prover. 
        If the verifier rejects, abort the protocol.
        \label{step:SQEnc-wp-np}

        \item \textbf{Quantum Receiver.} 
        Measure register $\calB$ in the Hadamard basis to obtain a string $d_{B}$ and the leftover state
        \[
            \ket{\qsh} 
            \coloneqq
            (-1)^{d_{B}\cdot x_{0}^{(B)}} \left(\ket{0, x_{0}^{(A)}} + (-1)^{d_{B}\cdot s_{B}} \ket{1, x_{1}^{(A)}} \right)
        \]
        where $s_B = x_{0}^{(B)}\oplus x_1^{(B)}$.
        Send $d_B$ to the receiver and output $\ket{\qsh}$.

        \item \textbf{Classical Sender.}
        Verify that $d_B \in \Good_{x_0^{(B)}, x_1^{(B)}}$. If not, set $p_B$ to be a random bit, and if so set $p_B$ as follows.
        Compute $s_{B} = x_{0}^{(B)}\oplus x_{1}^{(B)}$ by using $\td_{B}$ to invert $x_{b}^{(B)} = \Inv(\td_{B},b, y_B)$ for both $b=0$ and $b=1$. 
        Compute
        \[
            p_B = d_{B} \cdot s_{B}.
        \]
        Then output the classical share and verification key
        \[
            \csh = (p_B\oplus m, \td_A)
            \qquad \text{and}\qquad
            \verkey = (\pubparams_A, y_A)
        \]
    \end{enumerate}
\end{construction}

\begin{construction}[Reconstruction $\Reconstruct(\csh, \ket{\qsh})$]\label{constr:sq-cd-recon}
    To reconstruct the secret, do the following.
    \begin{enumerate}
        \item Parse $\csh = (\tilde{m}, \td_A)$ and coherently evaluate $f_{\pubparams_A}\left(b, x_{b}^{(A)}\right)$ on $\ket{\qsh}$ to obtain $y_A$.\footnote{Alternatively, $y_A$ can be included in either share.}
        \item Compute $x_0^{(A)}\gets \Inv(y_A, \td_A, 0)$ and $x_1^{(A)} \gets \Inv(y_A, \td_A, 1)$ . Compute $s_A = \left(x_0^{(A)} \oplus x_1^{(A)}\right)$.
    \item Measure $\ket{\qsh}$ in the Hadamard basis to obtain a string $d_A$.
    \item Output $(d_A \cdot (1\concat s_A)) \oplus \tilde{m}$.
    \end{enumerate}
\end{construction}

\begin{construction}[Deletion and Verification]\label{constr:sq-cd-del-ver}
    Deletion and verification are as follows.
    \begin{itemize}
        \item \underline{$\Del(\ket{\psi})$.} To delete $\ket{\psi}$, measure it in the computational basis to obtain $(b, x_b^A)$. Output $\cert = (b, x_b^A)$.
        \item \underline{$\Ver(\verkey, \cert)$.} Parse $\verkey = (\pubparams_A, y_A)$ and $\cert = (b, x_b^A)$. Accept if $f_{\pubparams_A}(b; x_b^A) = y_A$.
    \end{itemize}
\end{construction}

\noindent We show that so long as $\calZ$ semantically hides an adaptively-chosen $\csh$, the above scheme satisfies $\calZ$-deletion. That is, 
the following experiment should be computationally indistinguishable between $b=0$ and $b= 1$ for every QPT adversary $A$.
\begin{enumerate}
    \item Sample $\aux_1 \gets \calZ_{1}(1^\secpar; r)$ using fresh randomness $r$.
    \item The adversary $A(\aux_1)$ chooses a string $x \in \{0,1\}^n$.
    \item Using the same randomness $r$ as $\calZ$, sample
    \[
        \aux_2 \gets 
        \begin{cases}
            \calZ_2(1^\secpar, 0^n; r) &\text{if } b=0
            \\
            \calZ_2(1^\secpar, x; r) &\text{if } b=1
        \end{cases}
    \]
    \item Output $(\aux_1, \aux_2)$.
\end{enumerate}
We refer to this notion as \textbf{adaptive semantic hiding}.

\begin{theorem}\label{thm:sq-cd}
    The scheme given in \Cref{constr:sq-cd-keygen,constr:sq-cd-recon,constr:sq-cd-del-ver} is a semi-quantum encryption scheme with publicly verifiable $\calZ$-deletion for every 
    $\calZ$ satisfying adaptive semantic hiding.
\end{theorem}

\subsection{Analysis}

\begin{proof}
    We first show correctness. 
    Deletion correctness is obvious from the construction of $\ket{\qsh}$ as a superposition over preimages of $y_A$.
    For decryption correctness, observe that $\ket{\qsh}$ is negligibly far from the claimed state $\frac{1}{2}\left(\ket{0,x_0} + (-1)^{p_B}\ket{1,x_1}\right)$,{} except with negligible probability, by the balance property of $\calF$ (\Cref{def:tcf-balance}). Moreover the witness-preserving argument affects this state negligibly.
    The claimed state can be 
    interpreted as a coset state for the subspace $S = \{0, 1\concat s_A\}$ with computational basis offset $x_{0}^{(A)}$ and Hadamard basis offset $\Delta$ such that $\Delta \cdot (1\concat s_A) = p_B$. Therefore it can be written in the Hadamard basis as
    \[
        H \ket{\qsh}
        =
        \sum_{d' \in S^\perp} (-1)^{\left(0\concat x_{0}^{(A)}\right) \cdot d} \ket{d \oplus \Delta}
    \]
    Measuring this gives an element $d_A = d \oplus \Delta$ where $d \cdot (1\concat s_A) = 0$ and $\Delta \cdot (1\concat s_A) = p_B$. 
    Therefore $d_A \cdot (1\concat s_A) = p_B$. Since $\widetilde{m} = p_B \oplus m$, XORing it with $d_A \cdot (1\concat s_A)$ recovers the message.

    We prove security next. For each of the security proofs, we will consider a modified sharing protocol $\Share'_{\epsilon}$ which uses the indistinguishable simulator-extractor of $\SPNP$. Explicitly, let ($\Sim$, $\Ext_\epsilon$) be the QPT simulator and size $\poly[1/\epsilon]$ unitary guaranteed by $\epsilon$-indistinguishable-extraction. Then $\Share'_{\epsilon}$ does the following:
    \begin{enumerate}
        \item Run $\Share$ until the witness-preserving argument in step \cref{step:SQEnc-wp-np}.
        \item Run $\Sim$.
        \item Set $\gamma = \epsilon^2$. Run $\Ext_\gamma$ to extract a witness register $\calW$, measure whether the witness register contains a valid witness, then run $\Ext_\gamma^\dagger$. If the extraction fails, abort.
        \item Complete $\Share$.
    \end{enumerate}
    The modified protocol $\Share'_\epsilon$ is $\sqrt{\gamma} = \epsilon$-indistinguishable from the output of $\Share$ by \Cref{lem:iext-close}.

    \paragraph{Privacy.}
    Suppose some receiver violated privacy with probability $\epsilon$.
    We reduce to the adaptive hardcore bit property of $\pubparams_B$ using $\epsilon$-indistinguishable extraction and the collapsing property of $\pubparams_A$.
    Consider a hybrid experiment where $\Share$ is replaced by $\Share'_{\epsilon/2}$. This is $\epsilon/2$-indistinguishable and thus reduces the adversary's advantage by at most $\epsilon/2$. Next, consider the hybrid experiment where the extracted witness $(b, x_b^{(A)}, x_b^{(B)})$ is measured in the computational basis to obtain $(b, x_b^{(B)})$. This is indistinguishable by the collapsing property of $\pubparams_A$ because $\td_A$ is never revealed. Finally, the adversary's advantage in guessing $m$, and therefore the adaptive hardcore bit $d_B \cdot s_B$, is at least $\epsilon/2$. Since this hybrid also extracts $(b, x_b^{(B)})$ before the guess, the adversary therefore contradicts the adaptive hardcore bit property of $\pubparams_B$.

    \paragraph{Deletion Security.}
    We show that for every $\epsilon = \poly[1/\secpar]$, no QPT adversary can have $\epsilon$ advantage in distinguishing the message during the deletion game.
    Fix $\epsilon$ and consider the following sequence of hybrid experiments:
    \begin{itemize}
        \item $\Hyb_0(m) = \CDGame_\calZ(1^\secpar, m, (R^*, \Del^*))$ is the certified $\calZ$-deletion game which shares $m\in \{0,1\}$.
        
        \item $\Hyb_1(m)$ changes $\Share$ to $\Share'_{\epsilon/4}$, which runs the indistinguishable-extractor for $\SPNP$. This exposes the extracted witness in between running $\Ext_{\gamma}$ and $\Ext_{\gamma}^\dagger$ and aborts if extraction fails.
        
        \item $\Hyb_2(m)$ changes $\Hyb_1(m)$ by inserting a random phase $p'\in \{0,1\}$ when the witness is extracted and setting the masked message to $p_B\oplus m\oplus p'$ instead of $p_B \oplus m$.
        Explicitly, to insert the phase $p'$, after extracting a witness register supported on $\ket{b, x^{(A)}_b, x^{(B)}_b)}$, apply $Z^{p'}$ to the first qubit $b$.
        
        \item $\Hyb_3(m)$ purifies hybrid $\Hyb_2$ by adding a register $\calP$ containing the inserted phase $p'$. Explicitly, prepare $\ket{+}$ in $\calP$, then apply a controlled $Z$ from $\calP$ to the extracted witness register. To compute the masked message $p_B \oplus m\oplus p'$ for $\calZ$, measure $\calP$ in the computational basis to get $p'$.
        
        \item $\Hyb_4(m)$ changes $\Hyb_3(m)$ by substituting a random bit $\widetilde{m}$ for $p_B \oplus m \oplus p'$ in $\csh$. At the end of the game, it measures register $\calP$ in the computational basis to get $p'$ and checks whether $\tilde{m} = p_B\oplus m \oplus p'$. If not, the experiment outputs $\bot$.
        
        \item $\Hyb_5(m)$ performs an additional measurement on the phase register $\calP$ after checking the certificate $\cert$. It parses $\cert = \left(b, x_b^{(A)}\right)$, then measures $\calP$ in the Hadamard basis. If the result is \emph{not} $H\ketbra{b}H$, the experiment immediately outputs $\bot$.
        Otherwise, it continues as before: measure $\calP$ in the \emph{computational} basis to obtain $p'$ and check whether $\tilde{m} = p_B \oplus m \oplus p'$.
        
        \item $\Hyb_6$ is a restatement of $\Hyb_5$ without using $m$. Run $\Share'$ and insert a random phase $p'$. Send a random ciphertext $\tilde{m}$, then receive and check $\cert$. If $\cert$ is invalid or $\calP$ does not match it, output $\bot$. Otherwise, output $\bot$ with probability $1/2$ anyway.
    \end{itemize}

    Define $\Advtg(\Hyb_i)$ to be the maximum advantage of any $\QPT$ distinguisher $D$ in $\Hyb_i$:
    \begin{equation}
\Advtg(\Hyb_i) \coloneqq
        \max_{\QPT\ D} \left|\Pr\big[D(\Hyb_i(1^\secpar, 0, (R^*, \Del^*))\big] - \Pr\big[D(\Hyb_i(1^\secpar, 1, (R^*, \Del^*))\big]\right|.
\end{equation}

    $\epsilon$-indistinguishable-extraction and \Cref{lem:iext-close} implies that the $\Hyb_{0}(b)$ and $\Hyb_{1}(b)$ are indistinguishable up to $\epsilon/4 + (\epsilon/4)^2< \epsilon/3$ probability. This holds for both $b \in \{0,1\}$ and all $\epsilon \in (0,1]$.
    Thus, the adversary's advantage in distinguishing $\Hyb_0(0)$ from $\Hyb_0(1)$ changes by at most $2\epsilon/3$ in $\Hyb_1$:
    \begin{equation}\label{eq:cd-proof-hyb-0-1}
        |\Advtg(\Hyb_1) - \Advtg(\Hyb_0)| \leq 2\epsilon/3.
    \end{equation}
    The phase injection lemma (\Cref{lem:phase-injection}) applied to $\pubparams_B$ implies that the advantage also changes negligibly going to $\Hyb_2$:\footnote{The phase injection lemma cannot be applied to $\pubparams_A$ because $\td_A$ is later revealed. $\td_B$ is never revealed, only $d_B \cdot s_B$, which is precisely the scenario which \Cref{lem:phase-injection} deals with.}
    \begin{equation}\label{eq:cd-proof-hyb-1-2}
        |\Advtg(\Hyb_2) - \Advtg(\Hyb_1)| = \negl.
    \end{equation}
    $\Hyb_3$ is equivalent to $\Hyb_2$, so
    \begin{equation}\label{eq:cd-proof-hyb-2-3}
        \Advtg(\Hyb_3) = \Advtg(\Hyb_2).
    \end{equation}
    Since $\Hyb_4(m)$ is identically distributed to the distribution which outputs $\Hyb_3(m)$ with probability $1/2$ and otherwise outputs $\bot$,
    \begin{equation}\label{eq:cd-proof-hyb-3-4}
        \Advtg(\Hyb_4) \geq \frac{1}{2}\Advtg(\Hyb_3).
    \end{equation}
    We show that the adversary's distinguishing advantage changes negligibly from $\Hyb_4$ to $\Hyb_5$ in \Cref{claim:cd-proof-hyb4-hyb5}:
    \begin{equation}\label{eq:cd-proof-hyb-4-5}
        |\Advtg(\Hyb_5) - \Advtg(\Hyb_4)| = \negl.
    \end{equation}
    In $\Hyb_5$, $p'$ is the result of measuring a Hadamard basis state in the computational basis, and so is uniformly random. Therefore the $\widetilde{m} = p_B\oplus m \oplus p'$ check causes the experiment to abort with probability $1/2$, independently of $m$. This is precisely $\Hyb_6$, so
    \begin{equation}\label{eq:cd-proof-hyb-5-6}
        \Advtg(\Hyb_6) = \Advtg(\Hyb_5).
    \end{equation}
    Finally, $\Hyb_6$ is independent of $m$, so
    \begin{equation}\label{eq:cd-proof-hyb-6}
        \Advtg(\Hyb_6) = 0.
    \end{equation}

    Putting all these together,
    \[
        \Advtg(\Hyb_0) \leq (2\epsilon/3 + \negl) < \epsilon
    \]

    It remains to be shown that the adversary's distinguishing advantage changes negligibly from $\Hyb_4$ to $\Hyb_5$.

    \begin{claim}\label{claim:cd-proof-hyb4-hyb5}
        For all $\QPT$ adversaries $(R^*, \Del^*)$,
        \[
             |\Advtg(\Hyb_5) - \Advtg(\Hyb_4)| = \negl
        \]
    \end{claim}
    \begin{proof}
        By the gentle measurement lemma, it suffices to show that the measurement on $\calP$ outputs $H\ketbra{b}H$ with overwhelming probability.

        Consider the following truncated hybrid experiments:
        \begin{itemize}
            \item $\Hyb_5'$ which runs $\Hyb_5$ until the measurement on $\calP$. The probability of the measurement outputting $H\ketbra{b}H$ is identical to in $\Hyb_5$. Note that $\csh$ and $r$ are not revealed in this hybrid.
            \item $\Hyb_5''$ is the same as $\Hyb_5'$, except it replaces $\aux_2\gets \calZ_2(\csh; r)$ by $\aux_2 \gets \calZ_{2}(0; r)$.
            \item $\Hyb_5'''$ is the same as $\Hyb_5''$, except when the witness is extracted, it applies $\CZ$ in the other direction from the witness's $b$ register to register $\calP$, resulting in $\calP$ containing $H\ket{b}$. Moreover, it measures the witness to obtain a value $(b, x_b^{(A)}, x_b^{(B)})$.
        \end{itemize}
        In $\Hyb_5'''$, the probability that $\cert$ is accepting and $\cert \neq x_b^{(A)}$ is negligible. Otherwise, we have contradicted the claw-free property of $\calF$ by finding two preimages for $y_A$ under $\pubparams_A$ with noticeable probability. 
        Note that this relies on the fact that $\csh$, and therefore $\td_A$, is not revealed in $\Hyb_5'''$ (either directly or indirectly via $\calZ$ and $r$).
        Moreover, whenever $\cert = x_b^{(A)}$, the copied witness is $x_b^{(A)}$, so $\calP$ contains $H\ket{b}$ as we just mentioned.

        Next, we show that $\Hyb_5''' = \Hyb_5'' \approx_c \Hyb_5'$. Since it is efficient to check whether $\calP$ contains $H\ket{b}$ such that $\cert = x_b^{(A)}$, this implies that $\Hyb_5'$, and therefore $\Hyb_5$ satisfies this predicate with overwhelming probability.

        $\Hyb_5'''$ is identically distributed to $\Hyb_5''$ because controlled $Z$ is direction invariant: $\CZ_{1\rightarrow 2} = \CZ_{2\rightarrow 1}$. Since register $\calP$ initially contains $\ket{+}$, after inserting the phase it contains $H\ket{b}$ controlled on the witness register containing $b$. The only other operation on $\calP$ is to subsequently measure it in the Hadamard basis, which can be commuted to this step and results in a single $H\ket{b}$.
        Moreover, $x_b^{(A)}$ is uniquely determined by $b$, $y_A$, and $\pubparams_A$, so the value $x_b^{(A)}$ can also be copied into the challenger's registers from the witness register at the same time without affecting the witness register any further. The registers containing the copies are disjoint from the rest of the experiment, so measuring them in the computational basis does not affect the distribution of the rest of the experiment.

        Finally, $\Hyb_5''$ is computationally indistinguishable from $\Hyb_5'$ by the adaptive semantic hiding of $\calZ$. This relies on fact that the randomness $r$ for $\calZ$ is not revealed in $\Hyb_5'$.
    \end{proof}
\end{proof}

\section{Applications}
\label[appendix]{sec:applications}

In this section, we show how to achieve several other primitives with semi-quantum certified deletion by building on the secret sharing scheme from \Cref{sec:sqss}.
Each of these primitives uses a similar approach:
\begin{enumerate}
    \item Start with a semantically-hiding primitive without certified deletion (for example, public-key encryption).
    \item Share the message  $m$ into classical share $\csh$ and quantum share $\ket{\qsh}$.
    \item Encode $\csh$ using the starting primitive (for example, encrypt it under the public key).
\end{enumerate}
This general framework applies to commitments, public-key encryption, and advanced encryption schemes such as attribute-based, identity-based, and fully-homomorphic encryption.

\subsection{Public-Key Encryption}\label{sec:sqenc}

In this section, we show how to generically add semi-quantum, publicly-verifiable certified deletion to any public-key encryption scheme. The compiler can also be generically applied to advanced encryption schemes such as Attribute-Based Encryption or Identity-Based Encryption.

\begin{definition}\label{def:sq-encryption-cd}
A semi-quantum public-key encryption scheme with (publicly verifiable) certified deletion consists of the following algorithms and protocol:
\begin{itemize}
    \item \textbf{$\KeyGen(1^\secpar)\rightarrow(\pk,\sk)$} is a classical $\PPT$ key-generation algorithm that takes as input the security parameter $1^\secpar$, then outputs a classical public key $\pk$ and a classical secret key $\sk$.

    \item \textbf{$\Enc\langle S_{\mathsf C}(\pk,m),R_{\mathsf Q}\rangle(1^\secpar)\rightarrow(\verkey,\ket{\ct})$} is an interactive encryption protocol between a classical $\PPT$ sender $S_{\mathsf C}$ and a quantum $\QPT$ receiver $R_{\mathsf Q}$, using only classical communication.
    The sender inputs the public key $\pk$ and a message $m$.
    At the end of the protocol, the sender outputs a classical deletion verification key $\verkey$, and the receiver outputs a quantum ciphertext $\ket{\ct}$.

    \item \textbf{$\Dec(\sk,\ket{\ct})\rightarrow m$} is a $\QPT$ decryption algorithm that takes as input the secret key $\sk$ and a quantum ciphertext $\ket{\ct}$, then outputs a classical message $m$.  

    \item \textbf{$\Del(\ket{\ct})\rightarrow\cert$} is a $\QPT$ deletion algorithm that takes as input a quantum ciphertext $\ket{\ct}$, then outputs a classical deletion certificate $\cert$.

    \item \textbf{$\Ver(\verkey,\cert)\rightarrow\Accept/\Reject$} is a classical $\PPT$ verification algorithm that takes as input the verification key $\verkey$ and a certificate $\cert$, then outputs accept or reject.
\end{itemize}
It satisfies the following properties:
\begin{itemize}
    \item \textbf{Correctness.}
    There exists a negligible function $\mu$ such that for every $m\in\{0,1\}$ and every $\secpar\in\mathbb{N}$,
    \[
        \Pr\left[m\gets\Dec(\sk,\ket{\ct}):\begin{array}{c}
            (\pk,\sk)\gets\KeyGen(1^\secpar)\\
            (\verkey,\ket{\ct})\gets\Enc\langle S_{\mathsf C}(\pk,m),R_{\mathsf Q}\rangle(1^\secpar)
        \end{array}\right]
        =1-\mu(\secpar).
    \]
    Moreover, honestly generated deletion certificates are accepted with overwhelming probability:
    \[
        \Pr\left[\Accept\gets\Ver(\verkey,\cert):\begin{array}{c}
            (\pk,\sk)\gets\KeyGen(1^\secpar)\\
            (\verkey,\ket{\ct})\gets\Enc\langle S_{\mathsf C}(\pk,m),R_{\mathsf Q}\rangle(1^\secpar)\\
            \cert\gets\Del(\ket{\ct})
        \end{array}\right]
        =1-\mu(\secpar).
    \]

    \item \textbf{Ciphertext Privacy.}
    For every $\QPT$ message-chooser $A$ and every $\QPT$ adversarial receiver $R^*$ outputting a register $\calR_{\ct}$,  
    \[
        \INDCPA(0, (A, R^*)) \approx_c \INDCPA(1, (A, R^*))
    \]
    where $\INDCPA(b, (A, R^*))$ is the following game.
    \begin{enumerate}
        \item Sample $(\pk,\sk)\gets\KeyGen(1^\secpar)$.
        \item The adversary receives $\pk$ and chooses messages $(m_0, m_1) \gets A(\pk)$.
        \item Encrypt message $m_b$ as $(\verkey,\calR_{\ct})\gets\Enc\execution{S_{\mathsf C}(\pk,m_b),\ R^*(\pk)}(1^\secpar)$.
        \item Output $(\pk, \verkey, \calR_{\ct})$.
    \end{enumerate}

    \item \textbf{Publicly Verifiable Deletion Security.}{}
    For every $\QPT$ adversary $(A, R^*,\Del^*)$,
    \[
        \CDGame_{\mathrm{enc}}(1^\secpar,0,(R^*,\Del^*))
        \approx_c
        \CDGame_{\mathrm{enc}}(1^\secpar,1,(R^*,\Del^*)),
    \]
    where the game $\CDGame_{\mathrm{enc}}(1^\secpar, b,(R^*,\Del^*))$ is played as follows:
    \begin{enumerate}
        \item Sample $(\pk,\sk)\gets\KeyGen(1^\secpar)$.
        \item $A(\pk)$ gets $\pk$ and chooses two messages $m_0$ and $m_1$. It also passes a register $\calR_{\adv}$ to $R^*$.
        \item Run $(\verkey,\calR_{\ct})\gets\Enc\langle S_{\mathsf C}(\pk,m_b),R^*(\calR_{\adv})\rangle(1^\secpar)$.
        \item Compute $(\calR_{\Del},\calR_{\cert})\gets\Del^*(\calR_{\ct},\pk,\verkey)$.\footnote{If $\Del^*$ does not receive $\verkey$ during the game, we omit ``publicly verifiable'' and instead call the property ``Certified Deletion Security''.}
        \item If $\Ver(\verkey,\calR_{\cert})$ accepts, output $(\calR_{\Del},\sk)$. Otherwise, output $\bot$.
    \end{enumerate}
\end{itemize}
\end{definition}

We remark that since the encryption process is interactive, the folklore transformation from bitwise selective security to multi-bit adaptive security requires some additional care. This is an issue of timing. The folklore reduction switches the bits one at a time by receiving $\Enc(b)$, then giving $\pk$ to the adversary and inserting the true challenge into the adaptive challenge at index $i$. However, here encryption requires the cooperation of the adversary and cannot be done at all before receiving $\pk$. Thus, we directly consider and prove multi-bit adaptive security.

The construction uses the following two primitives:
\begin{itemize}
    \item $\SQSS=(\SQSS.\Share,\SQSS.\Reconstruct,\SQSS.\Del,\SQSS.\Ver)$ is a 2-out-of-2 semi-quantum secret sharing scheme with publicly verifiable $\calZ$-deletion security for every 2-stage distribution $\calZ$ which has adaptive semantical hiding.
    \item $\PKE=(\PKE.\KeyGen,\PKE.\Enc,\PKE.\Dec)$ is a post-quantum semantically secure public-key encryption scheme. We abuse notation by considering $\PKE$'s message space to be variable since this is implied by indistinguishable chosen-plaintext attack security for PKE.
\end{itemize}

\begin{construction}[Semi-Quantum Public-Key Encryption with Certified Deletion]\label{constr:sq-pke-cd}
The message space $\{0,1\}^n$ is the space of all $n$-bit messages.
The algorithms and interactive encryption protocol are defined as follows.
\begin{itemize}
    \item \underline{$\KeyGen(1^\secpar)$.}
    Compute
    \[
        (\pk,\sk)\gets\PKE.\KeyGen(1^\secpar)
    \]
    and output $(\pk,\sk)$.

    \item \underline{$\Enc\langle S_{\mathsf C}(\pk,m),R_{\mathsf Q}\rangle(1^\secpar)$.}
    The classical sender and quantum receiver proceed as follows.
    \begin{enumerate}
        \item For each bit $m_i$ of $m$, run the sharing protocol
        \[
            ((\verkey_i,\csh_i),\ket{\qsh})
            \gets
            \SQSS.\Share\langle S_{\mathsf C}(m_i),R_{\mathsf Q}\rangle(1^\secpar).
        \]
        
        \item The sender computes
        \[
            \widetilde{\csh}\gets\PKE.\Enc(\pk,(\csh)_{i\in [n]})
        \]
        and sends $\widetilde{\csh}$ to the receiver.
        
        \item The sender outputs $\verkey = (\verkey_i)_{i\in [n]}$, and the receiver outputs the quantum ciphertext
        \[
            \ket{\ct}\coloneqq\left(\bigotimes_{i\in [n]}\ket{\qsh_i},\ \widetilde{\csh}\right).
        \]
    \end{enumerate}

    \item \underline{$\Dec(\sk,\ket{\ct})$.}
    Parse $\ket{\ct}=(\ket{\qsh},\widetilde{\csh})$, compute
    \[
        (\csh_i)_{i\in [n]} \gets\PKE.\Dec(\sk,\widetilde{\csh}).
    \]
    If $\PKE.\Dec$ outputs $\bot$, output $\bot$. Otherwise, reconstruct $m_i \gets \SQSS.\Reconstruct(\csh_i,\ket{\qsh_i})${} for each $i\in [n]$ and output $m = (m_i)_{i\in [n]}$.

    \item \underline{$\Del(\ket{\ct})$.}
    Parse $\ket{\ct}=(\ket{\qsh},\widetilde{\csh})$, compute
    \[
        \cert_i \gets\SQSS.\Del(\ket{\qsh_i}),
    \]
    for each $i\in [n]$
    and output $\cert = (\cert_i)_{i\in [n]}$.

    \item \underline{$\Ver(\verkey,\cert)$.}
    Output $\SQSS.\Ver(\verkey,\cert)$.
\end{itemize}
\end{construction}

\begin{remark}
The same generic compiler can also be applied to advanced encryption schemes such as Attribute-Based Encryption or Identity-Based Encryption. By encrypting the classical share under the ABE or IBE scheme, the share, and therefore the message, is only revealed when the attributes or identity conditions are satisfied.
This yields the corresponding semi-quantum ABE or IBE primitive with (publicly verifiable) certified deletion.
\end{remark}

\begin{theorem}\label{thm:sq-pke-cd}
Assume there exists a post-quantum semantically secure public-key encryption scheme and a 2-out-of-2 semi-quantum secret sharing scheme with publicly verifiable $\calZ$-deletion security for every $\calZ$ which has adaptive semantic hiding.

Then there exists a semi-quantum public-key encryption scheme with publicly verifiable certified deletion.
In particular, the scheme in \Cref{constr:sq-pke-cd} satisfies \Cref{def:sq-encryption-cd}.
\end{theorem}
\begin{proof}
Correctness follows from the correctness of $\PKE$ and $\SQSS$, while deletion correctness follows directly from that of $\SQSS$.
For ciphertext privacy, post-quantum semantic security of $\PKE$ lets us replace $\widetilde{\csh}=\PKE.\Enc(\pk,\csh)$ by an encryption of $0^{|\csh|}$.
The encrypted component is then independent of $m$, so privacy of $\SQSS$ for the receiver's full view, including $\verkey$, allows us to change $m$ from $m_0$ to $m_1$; reversing the first replacement proves ciphertext privacy.

Next we show deletion security by reducing to the general $\calZ$-deletion security of $\SQSS$. Consider the series of $n$ hybrid experiment where we switch from encrypting $m_0$ to $m_1$ one bit at a time (in the $i$'th hybrid, the first $i$ bits are according to $m_1$ and the rest are according to $m_0$).
Suppose some hybrid $i-1$ were distinguishable from hybrid $i$ by a QPT adversary $(A, R^*, \Del^*)$ and QPT distinguisher $D$. 

We describe the reduction together with the 2-stage $\calZ$ distribution.
\begin{enumerate}
    \item $\pk \gets \calZ_1(r)$ samples a key pair $(\pk, \sk) \gets \PKE.\KeyGen(1^\secpar)$ using $r$, then outputs the public key $\pk$.
    
    \item \textbf{Reduction Sharing Step.} 
    The reduction computes the adversary's message choices $(m^{(0)}, m^{(1)}, \calR_\adv) \gets A(\pk)$.
    
    It participates in sharing as follows. It internally acts as the sender to execute $\Share$ with $R^*$ for the first $i-1$ bits $m_j^{(1)}$ of $m^{(1)}$. 
    Then, it runs $\Share$ with the challenger, acting as the receiver and using $R^*$'s code.
    Finally, it internally acts as the sender to execute $\Share$ with $R^*$ for the last $n-i$ bits $m_j^{(0)}$ of $m^{(0)}$.

    \item $\ct \gets \calZ_2(\csh;r)$ uses $r$ to sample the same key pair $(\pk, \sk)$ as $\calZ_1$. Then, it encrypts and outputs $\ct \gets \PKE.\Enc((\csh_j)_{j\in [n]})$ using as randomness a disjoint part of $r$ from the part used to sample the key pair.

    \item \textbf{Reduction Output Guess.} The reduction receives $r$.
    If $m^{(0)}_i = m^{(1)}_i$, the reduction guesses which hybrid it is at random. 

    Otherwise, it does the following. First, obtain $\cert$ from $\Del^*$. If $\cert$ is valid (which may be checked using the public $\verkey$), send $\sk$ to $D$ along with the residual view of the adversary, and output $D$'s guess (flipped if $m^{(0)}_i = 1$ and $m^{(0)}_i = 0$).
\end{enumerate}
First, observe that $\calZ$ has adaptive semantic hiding by the IND-CPA security of $\PKE$. Second, observe that when $m^{(0)}_i = m^{(1)}_i$, the two hybrids are identical and there can be no advantage. Finally, observe that when they are not identical, the ending distribution of hybrid $i-1$ is identically distributed to the $\calZ$ deletion game for $\SQSS$ with bit $b = m^{(0)}_i$. Similarly, the ending distribution of hybrid $i$ is identically distributed to the $\calZ$ deletion game with $b = m^{(1)}_i$.
Therefore if hybrids $i-1$ and $i$ were noticeably distinguishable by a QPT adversary and distinguisher, the reduction would noticeably distinguish in the $\calZ$ deletion game for $\SQSS$, contradicting it.

\end{proof}

An almost identical proof also shows security when replacing $\PKE$ with an attribute-based or identity-based scheme. The only two steps that change are $\calZ_1$ and $\calZ_2$. Both also incorporate the adversary so that it can query for secret keys which do not match the attributes or identity of the challenge ciphertext, and output the adversary's residual state at the end of this phase. Since ABE and IBE guarantee the semantic hiding of the challenge ciphertext even given keys not matching its attributes or identity, $\calZ$ is still semantically hiding.

Our construction in \Cref{thm:sq-cd} gives the required semi-quantum secret sharing scheme under the post-quantum hardness of LWE.
Thus, \Cref{thm:sq-pke-cd} gives the following corollary.
\begin{corollary}\label{cor:sq-pke-cd-lwe}
Assuming the post-quantum hardness of $\LWE$ and PKE, ABE, or IBE, there exists a semi-quantum PKE, ABE, or IBE scheme with publicly verifiable certified deletion.
\end{corollary}

\subsection{Commitments}

We adapt the commitment syntax and security requirements of \cite{C:BarKhu23} to the semi-quantum setting.
As in the encryption definition above, ordinary hiding and hiding after an accepted deletion certificate are required separately.

\begin{definition}\label{def:sq-commitment-cd}
A semi-quantum bit commitment scheme with (publicly verifiable) certified deletion consists of the following algorithms and protocol:
\begin{itemize}
    \item \textbf{$\Com\langle C_{\mathsf C}(b),R_{\mathsf Q}\rangle(1^\secpar)\rightarrow((\verkey,\decom),\ket{\com})$} is an interactive commitment protocol between a classical $\PPT$ committer $C_{\mathsf C}$ and a quantum $\QPT$ receiver $R_{\mathsf Q}$, using only classical communication.
    The committer inputs a bit $b\in\{0,1\}$, and the receiver has no private input.
    At the end of the protocol, the committer outputs a classical deletion verification key $\verkey$ and classical decommitment information $\decom$, while the receiver outputs a quantum commitment $\ket{\com}$.

    \item \textbf{$\Reveal(\decom,\ket{\com})\rightarrow\mu\in\{0,1,\bot\}$} is a $\QPT$ algorithm that takes as input the decommitment information and quantum commitment, then outputs a bit or $\bot$.

    \item \textbf{$\Del(\ket{\com})\rightarrow\cert$} is a $\QPT$ deletion algorithm that takes as input the quantum commitment, then outputs a classical deletion certificate $\cert$.

    \item \textbf{$\Ver(\verkey,\cert)\rightarrow\Accept/\Reject$} is a classical $\PPT$ verification algorithm that takes as input the verification key and certificate, then outputs accept or reject.
\end{itemize}
It satisfies the following properties:
\begin{itemize}
    \item \textbf{Correctness.}
    There exists a negligible function $\mu$ such that for every $b\in\{0,1\}$ and every $\secpar\in\mathbb{N}$, honest decommitments recover the committed bit:
    \[
        \Pr\left[b\gets\Reveal(\decom,\ket{\com}):((\verkey,\decom),\ket{\com})\gets\Com\langle C_{\mathsf C}(b),R_{\mathsf Q}\rangle(1^\secpar)\right]
        =1-\mu(\secpar).
    \]
Moreover, honestly generated deletion certificates are accepted with overwhelming probability:
    \[
        \Pr\left[\Accept\gets\Ver(\verkey,\cert):\begin{array}{c}
            ((\verkey,\decom),\ket{\com})\gets\Com\langle C_{\mathsf C}(b),R_{\mathsf Q}\rangle(1^\secpar)\\
            \cert\gets\Del(\ket{\com})
        \end{array}\right]
        =1-\mu(\secpar).
    \]

    \item \textbf{Computational Hiding.}
    For every $\QPT$ adversarial receiver $R^*$ outputting a register $\calR_{\com}$,
    \begin{gather*}
        \left\{(\verkey,\calR_{\com}):((\verkey,\decom),\calR_{\com})\gets\Com\langle C_{\mathsf C}(0),R^*\rangle(1^\secpar)\right\}
        \\
        \approx_c
        \\
        \left\{(\verkey,\calR_{\com}):((\verkey,\decom),\calR_{\com})\gets\Com\langle C_{\mathsf C}(1),R^*\rangle(1^\secpar)\right\}.
    \end{gather*}

    \item \textbf{Binding.}
    Following \cite{AC:Unruh16,EPRINT:DalSpo22}, define $\CollapseGame_{C^*}(1^\secpar)$ for a three-stage malicious committer $C^*=(C^*_{\com},C^*_{\mathsf{rev}},C^*_{\mathsf{guess}})$ as follows:
    \begin{enumerate}
        \item Sample $b\gets\{0,1\}$ uniformly.
        \item Run $C^*_{\com}$ in the commitment phase with the honest receiver, and let $(\calR_C,\calR_{\com})$ denote their resulting states.
        Then run $C^*_{\mathsf{rev}}$ coherently to obtain a decommitment register $\calR_{\decom}$ and residual state $\calR'_C$.
        \item Coherently apply $\Reveal$ to $(\calR_{\decom},\calR_{\com})$, writing its output into a register $\calR_M$, and measure whether $\calR_M\neq\bot$.
        If the measurement rejects, sample $b'$ uniformly and proceed to the final step.
        \item If $b=1$, measure $\calR_M$ in the computational basis; if $b=0$, do nothing.
        \item Give $(\calR'_C,\calR_{\decom},\calR_M)$ to $C^*_{\mathsf{guess}}$, which outputs a bit $b'$.
        \item Output $1$ if $b'=b$, and output $0$ otherwise.
    \end{enumerate}
    The scheme is \emph{computational collapse binding} (respectively, \emph{statistical collapse binding}) if for every $\QPT$ (respectively, possibly unbounded) three-stage malicious committer $C^*$, there exists a negligible function $\mu$ such that
    \[
        \Pr\left[\CollapseGame_{C^*}(1^\secpar)=1\right]
        \leq \frac{1}{2}+\mu(\secpar).
    \]

    \item \textbf{Publicly Verifiable Deletion Hiding.}{}
    For every $\QPT$ adversary $(R^*,\Del^*)$,
    \[
        \CDGame_{\mathrm{com}}(1^\secpar,0,(R^*,\Del^*))
        \approx_c
        \CDGame_{\mathrm{com}}(1^\secpar,1,(R^*,\Del^*)),
    \]
    where the game $\CDGame_{\mathrm{com}}(1^\secpar,b,(R^*,\Del^*))$ is played as follows:
    \begin{enumerate}
        \item Run $((\verkey,\decom),\calR_{\com})\gets\Com\langle C_{\mathsf C}(b),R^*\rangle(1^\secpar)$.
        \item Compute $(\calR_{\Del},\calR_{\cert})\gets\Del^*(\calR_{\com},\verkey)$.
        \item If $\Ver(\verkey,\calR_{\cert})$ accepts, output $(\calR_{\Del},\decom)$. Otherwise, output $\bot$.
    \end{enumerate}
    If $\Del^*$ does not receive $\verkey$, we omit ``publicly verifiable.''
\end{itemize}
Thus, after an accepted deletion certificate, hiding continues to hold even when the decommitment information is subsequently exposed.
\end{definition}

Let $\COM$ be a non-interactive commitment scheme for classical strings.
We require post-quantum computational hiding and either computational or statistical collapse binding.
Its commitment algorithm $\COM.\Commit(m;r)$ commits to $m$ using randomness $r$, and $\COM.\Ver$ verifies an opening.
Let $\SQSS$ be a 2-out-of-2 semi-quantum secret sharing scheme with publicly verifiable $\calZ$-deletion security for every semantically hiding distribution $\calZ$.

\begin{construction}[Semi-Quantum Commitment with Certified Deletion]\label{constr:sq-com-cd}
The commitment protocol and algorithms are defined as follows.
\begin{itemize}
    \item \underline{$\Com\langle C_{\mathsf C}(b),R_{\mathsf Q}\rangle(1^\secpar)$.}
    The committer and receiver proceed as follows:
    \begin{enumerate}
        \item Run the sharing protocol
        \[
            ((\verkey,\csh),\ket{\qsh})
            \gets
            \SQSS.\Share\langle C_{\mathsf C}(b),R_{\mathsf Q}\rangle(1^\secpar).
        \]
        \item The committer samples randomness $r$, computes
        \[
            \gamma\gets\COM.\Commit(\csh;r),
        \]
        and sends $\gamma$ to the receiver.
        \item The committer outputs $(\verkey,\decom)$, where $\decom\coloneqq(\csh,r)$, and the receiver outputs
        \[
            \ket{\com}\coloneqq(\ket{\qsh},\gamma).
        \]
    \end{enumerate}

    \item \underline{$\Reveal(\decom,\ket{\com})$.}{}
    Parse $\decom=(\csh,r)$ and $\ket{\com}=(\ket{\qsh},\gamma)$.
    If $\COM.\Ver(\gamma,\csh,r)$ rejects, output $\bot$; otherwise, output $\SQSS.\Reconstruct(\csh,\ket{\qsh})$.

    \item \underline{$\Del(\ket{\com})$.}
    Parse $\ket{\com}=(\ket{\qsh},\gamma)$ and output $\SQSS.\Del(\ket{\qsh})$.

    \item \underline{$\Ver(\verkey,\cert)$.}
    Output $\SQSS.\Ver(\verkey,\cert)$.
\end{itemize}
\end{construction}

\begin{theorem}\label{thm:sq-com-cd}
Assume there exist $\COM$ and $\SQSS$ satisfying the properties above.
Then the scheme in \Cref{constr:sq-com-cd} is a computationally hiding semi-quantum commitment scheme with publicly verifiable certified deletion as in \Cref{def:sq-commitment-cd}.
It has computational collapse binding if $\COM$ has computational collapse binding, and statistical collapse binding if $\COM$ has statistical collapse binding.
\end{theorem}
\begin{proof}
Opening correctness follows from the correctness of $\COM$ and $\SQSS$, while deletion correctness follows directly from that of $\SQSS$.

For computational hiding, the post-quantum hiding of $\COM$ lets us replace
\[
    \gamma=\COM.\Commit(\csh;r)
\]
by a commitment to $0^{|\csh|}$, even in the presence of the receiver's quantum auxiliary state.
The commitment component is then independent of $b$, so privacy of $\SQSS$ allows us to change $b$ from $0$ to $1$; reversing the first replacement proves computational hiding.

For either version of collapse binding, the corresponding property of $\COM$ lets us measure the $\csh$ component of any valid decommitment with computationally or statistically negligible disturbance, respectively.
Conditioned on this classical value and the classical commitment transcript, the receiver's quantum share is not entangled with the committer, and reconstruction fixes the opened bit using only the receiver's share.
Measuring that bit therefore causes no further disturbance to the committer's state, proving the corresponding collapse-binding property of the constructed scheme.

For deletion hiding, let $\calZ$ be the distribution that, on input $\csh$ and randomness $r$, outputs $\gamma=\COM.\Commit(\csh;r)$.
This distribution is semantically hiding by the post-quantum hiding of $\COM$.
After an accepted certificate, the publicly verifiable $\calZ$-deletion game reveals $(\calR_{\Del},\csh,r)$, which is exactly the output required in the commitment deletion-hiding game because $\decom=(\csh,r)$.
Publicly verifiable $\calZ$-deletion security therefore proves deletion hiding.
\end{proof}

Post-quantum $\LWE$ implies post-quantum computationally hiding, statistically binding non-interactive classical commitments.
Together with our LWE-based semi-quantum secret sharing scheme from \Cref{thm:sq-cd}, \Cref{thm:sq-com-cd} gives the following corollary.
\begin{corollary}\label{cor:sq-com-cd-lwe}
Assuming the post-quantum hardness of $\LWE$, there exists a computationally hiding and statistically collapse-binding semi-quantum commitment scheme with publicly verifiable certified deletion.
\end{corollary}

\subsection{Fully Homomorphic Encryption}

Following \cite{C:BarKhu23}, fully homomorphic encryption with certified deletion is obtained by augmenting the encryption syntax with homomorphic evaluation while retaining the encryption and deletion-security requirements.

\begin{definition}\label{def:sq-fhe-cd}
A semi-quantum fully homomorphic encryption scheme with (publicly verifiable) certified deletion is a semi-quantum public-key encryption scheme satisfying \Cref{def:sq-encryption-cd}, extended in the standard way to polynomial-length messages, together with the following algorithm:
\begin{itemize}
    \item \textbf{$\Eval(\pk,C,\ket{\ct})\rightarrow\ket{\widetilde{\ct}}$} is a $\QPT$ evaluation algorithm that takes as input a public key $\pk$, a polynomial-size classical circuit $C$, and a quantum ciphertext $\ket{\ct}$, then outputs a potentially quantum evaluated ciphertext $\ket{\widetilde{\ct}}$.
\end{itemize}
In addition to the correctness, ciphertext-privacy, and deletion-security properties of \Cref{def:sq-encryption-cd}, the scheme satisfies:
\begin{itemize}
    \item \textbf{Evaluation Correctness.}
    There exists a negligible function $\mu$ such that for every polynomial-length message $x$, every polynomial-size classical circuit $C$ defined on $x$, and every $\secpar\in\mathbb{N}$,
    \[
        \Pr\left[\Dec(\sk,\ket{\widetilde{\ct}})=C(x):\begin{array}{c}
            (\pk,\sk)\gets\KeyGen(1^\secpar)\\
            (\verkey,\ket{\ct})\gets\Enc\langle S_{\mathsf C}(\pk,x),R_{\mathsf Q}\rangle(1^\secpar)\\
            \ket{\widetilde{\ct}}\gets\Eval(\pk,C,\ket{\ct})
        \end{array}\right]
        =1-\mu(\secpar).
    \]

    \item \textbf{Compactness.}
    The size of $\ket{\widetilde{\ct}}$ and the running time of $\Dec(\sk,\ket{\widetilde{\ct}})$ are bounded by a polynomial in the security parameter and output length, independently of the size or depth of $C$.
\end{itemize}
The supported circuit class contains all polynomial-size classical circuits.
The inherited certified deletion guarantee applies to ciphertexts output by the interactive encryption protocol; as in the standard blind-delegation use of FHE with certified deletion, an evaluator may evaluate coherently, return the evaluated ciphertext for decryption, and uncompute the evaluation before deleting the restored original ciphertext.
\end{definition}

Let $\FHE=(\FHE.\Gen,\FHE.\Enc,\FHE.\Eval,\FHE.\Dec)$ be a post-quantum semantically secure classical fully homomorphic encryption scheme whose message space contains the classical shares output by the polynomial-message extension of $\SQSS.\Share$.
We additionally require that $\SQSS.\Reconstruct$ is a PPT algorithm diagonal in the computational basis.
The construction from \Cref{sec:sqss} satisfies this requirement with two minor changes: additionally apply $H$ to $\ket{\qsh}$ at the end of $\Share$, and also apply it at the start of $\Del$. $y_A$ should also be included in either of the two shares instead of computing it from $\ket{\qsh}$ in the computational basis.

\begin{construction}[Semi-Quantum Fully Homomorphic Encryption with Certified Deletion]\label{constr:sq-fhe-cd}
The algorithms and interactive encryption protocol are defined as follows.
\begin{itemize}
    \item \underline{$\KeyGen(1^\secpar)$.}
    Compute
    \[
        (\pk,\sk)\gets\FHE.\Gen(1^\secpar)
    \]
    and output $(\pk,\sk)$.

    \item \underline{$\Enc\langle S_{\mathsf C}(\pk,x),R_{\mathsf Q}\rangle(1^\secpar)$.}
    Run the polynomial-message extension of the sharing protocol to obtain
    \[
        ((\verkey,\csh),\ket{\qsh})
        \gets
        \SQSS.\Share\langle S_{\mathsf C}(x),R_{\mathsf Q}\rangle(1^\secpar).
    \]
    The sender computes
    \[
        \widetilde{\csh}\gets\FHE.\Enc(\pk,\csh),
    \]
    sends $\widetilde{\csh}$ to the receiver, and outputs $\verkey$.
    The receiver outputs the fresh ciphertext
    \[
        \ket{\ct}\coloneqq(\mathsf{fresh},\ket{\qsh},\widetilde{\csh}).
    \]

    \item \underline{$\Eval(\pk,C,\ket{\ct})$.}
    If $\ket{\ct}$ is tagged $\mathsf{fresh}$, parse it as $(\mathsf{fresh},\ket{\qsh},\widetilde{\csh})$ and, for every computational-basis value $z$ of the quantum share, define the polynomial-size classical circuit
    \[
        C_z(\csh)\coloneqq C\bigl(\SQSS.\Reconstruct(\csh,\ket{z})\bigr).
    \]
    Coherently and reversibly compute $\widetilde{y}_z\gets\FHE.\Eval(\pk,C_z,\widetilde{\csh})$ into a fresh output register controlled on $z$, and tag that register $\mathsf{eval}$.
    If $\ket{\ct}$ is tagged $\mathsf{eval}$, parse it as $(\mathsf{eval},\widetilde{y})$, coherently and reversibly compute
    \[
        \widetilde{y}'\gets\FHE.\Eval(\pk,C,\widetilde{y})
    \]
    into a fresh output register, and tag that register $\mathsf{eval}$.
    In either case, output the tagged register as $\ket{\widetilde{\ct}}$ and retain the input ciphertext and work registers so that evaluation can later be uncomputed.

    \item \underline{$\Dec(\sk,\ket{\ct})$.}
    If $\ket{\ct}$ is tagged $\mathsf{fresh}$, parse it as $(\mathsf{fresh},\ket{\qsh},\widetilde{\csh})$, compute $\csh\gets\FHE.\Dec(\sk,\widetilde{\csh})$, and output $\SQSS.\Reconstruct(\csh,\ket{\qsh})$.
    If $\ket{\ct}$ is tagged $\mathsf{eval}$, parse it as $(\mathsf{eval},\widetilde{y})$, run $\FHE.\Dec(\sk,\widetilde{y})$ coherently, and measure its plaintext output in the computational basis.

    \item \underline{$\Del(\ket{\ct})$.}
    Parse a fresh ciphertext as $(\mathsf{fresh},\ket{\qsh},\widetilde{\csh})${} and output $\SQSS.\Del(\ket{\qsh})$.
    On an evaluated ciphertext, output $\bot$.

    \item \underline{$\Ver(\verkey,\cert)$.}
    Output $\SQSS.\Ver(\verkey,\cert)$.
\end{itemize}
\end{construction}

\begin{theorem}\label{thm:sq-fhe-cd}
Assume there exists a post-quantum semantically secure classical fully homomorphic encryption scheme and a 2-out-of-2 semi-quantum secret sharing scheme with publicly verifiable $\calZ$-deletion security for every semantically hiding distribution $\calZ$ whose reconstruction algorithm is PPT and diagonal in the computational basis.
Then the scheme in \Cref{constr:sq-fhe-cd} is a semi-quantum fully homomorphic encryption scheme with publicly verifiable certified deletion as in \Cref{def:sq-fhe-cd}.
\end{theorem}
\begin{proof}
The fresh-ciphertext algorithms are exactly \Cref{constr:sq-pke-cd} instantiated with $\PKE=\FHE$.
Therefore correctness, ciphertext privacy, deletion correctness, and deletion security follow from \Cref{thm:sq-pke-cd}.

For evaluation correctness, write the quantum share in the computational basis.
Because $\SQSS.\Reconstruct$ is diagonal in this basis, the coherent controlled evaluation acts independently on each basis value $z$, and correctness of $\SQSS$ implies that the total weight on values satisfying $\SQSS.\Reconstruct(\csh,\ket{z})=x$ is overwhelming.
For each such value, correctness of $\FHE$ implies that decrypting $\FHE.\Eval(\pk,C_z,\widetilde{\csh})$ outputs $C_z(\csh)=C(x)$ except with negligible probability.
Consequently, coherent evaluation followed by decryption outputs $C(x)$ with overwhelming probability.
On an evaluated ciphertext, evaluation correctness follows directly from that of $\FHE$; iterating this argument handles any polynomial sequence of evaluations.
Compactness follows from the compactness of $\FHE$ because the evaluated ciphertext contains only a tag and one evaluated $\FHE$ ciphertext register, whose size is independent of the size and depth of $C$.
\end{proof}

Post-quantum $\LWE$ implies post-quantum semantically secure fully homomorphic encryption \cite{C:GenSahWat13}.
Our LWE-based semi-quantum secret sharing scheme from \Cref{thm:sq-cd}, with the basis changes described above, has a PPT reconstruction algorithm diagonal in the computational basis.
Thus, \Cref{thm:sq-fhe-cd} gives the following corollary.
\begin{corollary}\label{cor:sq-fhe-cd-lwe}
Assuming the post-quantum hardness of $\LWE$, there exists a semi-quantum fully homomorphic encryption scheme with publicly verifiable certified deletion.
\end{corollary}

\section{Proofs of No Intrusion with Classical Upload}\label{sec:poni}

In this section we recall the definition of PoNIs from \cite{EC:GoyRai26} and show how they can be generically added to our prior constructions, following the approach of ~\cite[Appendix~A]{DBLP:journals/iacr/KalaiKR26}. Thus, the auditing client does not need quantum capabilities at any point.

\begin{definition}[Semi-Quantum Encryption with Proofs of No Intrusion]\label{def:sq-encryption-poni}
A semi-quantum public-key encryption scheme with a proof of no intrusion consists of the following algorithms and protocols:
\begin{itemize}
    \item \textbf{$\KeyGen(1^\secpar)\rightarrow(\pk,\sk)$} is a classical $\PPT$ key-generation algorithm that outputs a classical public key $\pk$ and a classical secret key $\sk$.

    \item \textbf{$\mathsf{VKGen}(1^\secpar)\rightarrow (\verkey_s, \verkey_p)$} is a classical $\PPT$ verification-key-generation algorithm which outputs a secret verification key $\verkey_s$ and a public one $\verkey_p$.

    \item \textbf{$\Enc\langle S_{\mathsf C}(\pk,m,\verkey_s),R_{\mathsf Q}\rangle(1^\secpar)\rightarrow\ket{\ct}$} is an interactive encryption protocol between a classical $\PPT$ sender $S_{\mathsf C}$ and a quantum $\QPT$ receiver $R_{\mathsf Q}$, using only classical communication.
    The sender inputs $\pk$, $m$, and the secret verification key $\verkey_s$, and the receiver outputs a quantum ciphertext $\ket{\ct}$.

    \item \textbf{$\Dec(\sk,\ket{\ct})\rightarrow m$} is a $\QPT$ decryption algorithm.

    \item \textbf{$\mathsf{PoNI}\langle P(\rho),V(\verkey_p)\rangle(1^\secpar)\rightarrow(\rho',b)$} is an interactive protocol with classical communication between a $\QPT$ prover $P$ holding a state $\rho$ and a classical $\PPT$ verifier $V$ holding the public verification key $\verkey_p$.
    At the end of the protocol, the prover outputs a residual state $\rho'$, and the verifier outputs a decision bit $b\in\{\Accept,\Reject\}$.
\end{itemize}
The scheme satisfies the following properties:
\begin{itemize}
    \item \textbf{Correctness.}
    There exists a negligible function $\mu$ such that for every message $m$, every $(\pk,\sk)$ in the support of $\KeyGen(1^\secpar)$, and every $(\verkey_s, \verkey_p)$ in the support of $\mathsf{VKGen}(1^\secpar)$,
    \[
        \Pr\left[m\gets\Dec(\sk,\ket{\ct}):
        \ket{\ct}\gets\Enc\langle S_{\mathsf C}(\pk,m,\verkey_s),R_{\mathsf Q}\rangle(1^\secpar)\right]
        =1-\mu(\secpar).
    \]

    \item \textbf{Ciphertext Privacy.}
    For every pair of equal-length messages $m_0,m_1$ and every $\QPT$ adversarial receiver $R^*$,
    \begin{gather*}
        \left\{(\pk,\verkey,\calR_{\ct}):\begin{array}{c}
            (\pk,\sk)\gets\KeyGen(1^\secpar)\\
            \verkey\gets\mathsf{VKGen}(1^\secpar)\\
            \calR_{\ct}\gets\Enc\langle S_{\mathsf C}(\pk,m_0,\verkey_s),R^*\rangle(1^\secpar)
        \end{array}\right\}
        \\
        \approx_c
        \\
        \left\{(\pk,\verkey,\calR_{\ct}):\begin{array}{c}
            (\pk,\sk)\gets\KeyGen(1^\secpar)\\
            \verkey\gets\mathsf{VKGen}(1^\secpar)\\
            \calR_{\ct}\gets\Enc\langle S_{\mathsf C}(\pk,m_1,\verkey_s),R^*\rangle(1^\secpar)
        \end{array}\right\}.
    \end{gather*}

    \item \textbf{PoNI Correctness and State Preservation.}
    There exists a negligible function $\mu$ such that, for every honestly generated $(\pk,\sk)$ and $(\verkey_s, \verkey_p)$, every message $m$, and every honestly generated ciphertext state $\ket{\ct}$ that decrypts to $m$ with certainty, an honest execution
    \[
        (\rho'_{\ct},b)\gets\mathsf{PoNI}\langle P(\ket{\ct}),V(\verkey_p)\rangle(1^\secpar)
    \]
    has $b=\Accept$ and leaves $\rho'_{\ct}$ supported on the space of ciphertext states that decrypt to $m$, except with probability $\mu(\secpar)$.
    Thus, the ciphertext may change during the proof, but it remains decryptable to the same message.

    \item \textbf{PoNI Security.}
    The scheme satisfies search security as in \Cref{def:sq-poni-search} or decisional security as in \Cref{def:sq-poni-decision}.
\end{itemize}
\end{definition}

The verification key may be reused for many ciphertexts, including ciphertexts under different public keys.
The security notions permit the prover to complete polynomially many tests before its state is split between the server and an intruder, and require a final accepting proof to be incompatible with the intruder retaining useful plaintext information.

\begin{definition}[Search Security for Proofs of No Intrusion]\label{def:sq-poni-search}
Let the message space be $\{0,1\}^{\secpar}$.
The game $\mathsf{PoNI\text{-}Enc\text{-}S}_n(1^\secpar,A)$ is played by an adversary $A=(A_1,A_P,A_H)$, consisting of three algorithms with auxiliary quantum input, as follows:
\begin{enumerate}
    \item Sample $m\gets\{0,1\}^\secpar$, $(\pk,\sk)\gets\KeyGen(1^\secpar)$, and $\verkey\gets\mathsf{VKGen}(1^\secpar)$.
    \item Run the interactive encryption protocol with $A_1$ as the receiver:
    \[
        \calR\gets\Enc\langle S_{\mathsf C}(\pk,m,\verkey),A_1(\pk)\rangle(1^\secpar).
    \]
    In the publicly verifiable version, $A_1$ also receives $\verkey_p$, where $\verkey = (\verkey_s, \verkey_p)$.
    \item Starting from $\calR$, execute $\mathsf{PoNI}\langle A_1(\calR),V(\verkey)\rangle(1^\secpar)$ a total of $n$ times, updating $\calR$ after each execution.
    If the verifier rejects in any execution, output $0$.
    \item Compute $(\calH,\calP)\gets A_1(\calR)$.
    \item Run
    \[
        (\calP',b)\gets\mathsf{PoNI}\langle A_P(\calP),V(\verkey_p)\rangle(1^\secpar).
    \]
    \item Compute $m'\gets A_H(\calH,\sk,\verkey)$.
    \item Output $1$ if $b=\Accept$ and $m'=m$, and output $0$ otherwise.
\end{enumerate}
The proof of no intrusion has $n$-time search security if, for every $\QPT$ adversary $A$,
\[
    \Pr\left[\mathsf{PoNI\text{-}Enc\text{-}S}_n(1^\secpar,A)=1\right]
    =\negl.
\]
It has search security if it has $n$-time search security for every polynomially bounded $n$.
If $A_1$ and $A_P$ are required to be $\QPT$ but $A_H$ may be unbounded, then the corresponding guarantee is called everlasting search security.
\end{definition}

\begin{definition}[Decisional Security for Proofs of No Intrusion]\label{def:sq-poni-decision}
For a message bit $m\in\{0,1\}$ and a non-negative integer $n$, the game $\mathsf{PoNI\text{-}Sec}_n(1^\secpar,m,A)${} is played by a stateful adversary $A$ with auxiliary quantum input as follows:
\begin{enumerate}
    \item Sample $(\pk,\sk)\gets\KeyGen(1^\secpar)$ and $\verkey\gets\mathsf{VKGen}(1^\secpar)$.
    \item Run the interactive encryption protocol with $A$ as the receiver:
    \[
        \calR\gets\Enc\langle S_{\mathsf C}(\pk,m,\verkey),A(\pk)\rangle(1^\secpar).
    \]
    In the publicly verifiable version, $A$ also receives $\verkey_p$.
    \item Starting from $\calR$, execute $\mathsf{PoNI}\langle A(\calR),V(\verkey)\rangle(1^\secpar)$ a total of $n$ times, updating $\calR$ after each execution.
    If the verifier rejects in any execution, output $\bot$.
    \item Apply $A$ to split its state into registers $(\calR_{\Dec},\calR_P)$.
    \item Run
    \[
        (\calR'_P,b)\gets\mathsf{PoNI}\langle A(\calR_P),V(\verkey)\rangle(1^\secpar).
    \]
    If $b=\Reject$, output $\bot$; otherwise, output $(\calR_{\Dec},\sk,\verkey)$.
\end{enumerate}
The proof of no intrusion has $n$-time decisional security if, for every $\QPT$ adversary $A$,
\[
    \left\{\mathsf{PoNI\text{-}Sec}_n(1^\secpar,0,A)\right\}
    \approx_c
    \left\{\mathsf{PoNI\text{-}Sec}_n(1^\secpar,1,A)\right\}.
\]
It has decisional security if it has $n$-time decisional security for every polynomially bounded $n$.
If the two distributions are statistically indistinguishable, then the corresponding guarantee is called everlasting decisional security.
\end{definition}

Unlike certified deletion, a proof of no intrusion is non-destructive: after an accepting execution, the server retains a ciphertext that decrypts to the same message and can be tested again.

Appendix~A of \cite{DBLP:journals/iacr/KalaiKR26} gives the following generic compiler from public-key encryption with publicly verifiable certified deletion to public-key encryption with proofs of no intrusion.
Let
\[
    \begin{aligned}
        \PKE_{\mathsf{cd}}=\bigl(&\PKE_{\mathsf{cd}}.\KeyGen,\PKE_{\mathsf{cd}}.\Enc,\\
        &\PKE_{\mathsf{cd}}.\Dec,\PKE_{\mathsf{cd}}.\Del,\PKE_{\mathsf{cd}}.\Ver\bigr)
    \end{aligned}
\]
be a semi-quantum public-key encryption scheme with publicly verifiable certified deletion.
Let $\mathsf{Sig}=(\mathsf{Sig.KeyGen},\mathsf{Sig.Sign},\mathsf{Sig.Ver})$ be a post-quantum signature scheme, and let $\mathsf{SPA}_{\NP}$ be the state-preserving argument of knowledge for $\NP$ from \cite{DBLP:journals/iacr/KalaiKR26}.

\begin{construction}[Generic PoNI Compiler]\label{constr:sq-poni-compiler}
The compiled scheme is works as follows:
\begin{itemize}
    \item \underline{$\KeyGen(1^\secpar)$.}
    Output $(\pk,\sk)\gets\PKE_{\mathsf{cd}}.\KeyGen(1^\secpar)$.

    \item \underline{$\mathsf{VKGen}(1^\secpar)$.}
    Sample
    \[
        (\mathsf{pk}_{\mathsf{Sig}},\mathsf{sk}_{\mathsf{Sig}})
        \gets\mathsf{Sig.KeyGen}(1^\secpar)
    \]
    and output $(\verkey_s, \verkey_p) \coloneqq(\mathsf{sk}_{\mathsf{Sig}},\mathsf{pk}_{\mathsf{Sig}})$.

    \item \underline{$\Enc\langle S_{\mathsf C}(\pk,m,\verkey_s),R_{\mathsf Q}\rangle(1^\secpar)$.}
    Parse $\verkey_s \mathsf{sk}_{\mathsf{Sig}}$ and run
    \[
        (\verkey',\ket{\ct})
        \gets
        \PKE_{\mathsf{cd}}.\Enc\langle S_{\mathsf C}(\pk,m),R_{\mathsf Q}\rangle(1^\secpar).
    \]
    The sender computes $\sigma\gets\mathsf{Sig.Sign}(\mathsf{sk}_{\mathsf{Sig}},\verkey')$ and sends $(\verkey',\sigma)$ to the receiver.
    The receiver outputs $\ket{\ct'}\coloneqq(\ket{\ct},\verkey',\sigma)$.

    \item \underline{$\Dec(\sk,\ket{\ct'})$.}
    Parse $\ket{\ct'}=(\ket{\ct},\verkey',\sigma)$ and output $\PKE_{\mathsf{cd}}.\Dec(\sk,\ket{\ct})$.
\end{itemize}
The proof of no intrusion proceeds as follows:
\begin{enumerate}
    \item The prover parses $\ket{\ct'}=(\ket{\ct},\verkey',\sigma)$, sends $(\verkey',\sigma)$ to the verifier, and coherently computes $\PKE_{\mathsf{cd}}.\Del(\ket{\ct})$ to obtain
    \[
        \sum_{\cert}\alpha_{\cert}\ket{\cert,\mathsf{garbage}_{\cert}}.
    \]
    \item The verifier parses $\verkey_p = \mathsf{pk}_{\mathsf{Sig}}$ and rejects unless $\mathsf{Sig.Ver}(\mathsf{pk}_{\mathsf{Sig}},\sigma,\verkey')$ accepts.
    \item The parties execute $\mathsf{SPA}_{\NP}$ for the statement $\verkey'$ and the language
    \[
        \left\{\verkey':\exists\cert\text{ such that }\PKE_{\mathsf{cd}}.\Ver(\verkey',\cert)=\Accept\right\},
    \]
    using the certificate register as the prover's witness register.
    \item The prover uncomputes the coherent implementation of $\PKE_{\mathsf{cd}}.\Del$.
\end{enumerate}
\end{construction}

\begin{theorem}[Generic PoNI Compiler \cite{DBLP:journals/iacr/KalaiKR26}]\label{thm:generic-pvcd-to-poni}
Assume the existence of public-key encryption with computational (resp. everlasting) publicly verifiable certified deletion, post-quantum signatures, and the state-preserving argument of knowledge for $\NP$ from \cite{DBLP:journals/iacr/KalaiKR26}.

Then there exists a public-key encryption scheme with computational (resp. everlasting) search proofs of no intrusion. Moreover, there exists a scheme with decisional PoNIs under the same assumptions in the quantum random oracle model.
\end{theorem}

By plugging in the construction of semi-quantum PKE with publicly verifiable certified deletion from \Cref{sec:sqenc}, we obtain semi-quantum PKE with PoNIs. 
We also mention that the same transformation can be applied to the other primitives constructed in this section.

\begin{corollary}\label{cor:sq-pke-poni}
Assuming the post-quantum hardness of $\LWE$, there exists a semi-quantum public-key encryption scheme with computational search PoNIs.
Furthermore, there exists a semi-quantum public-key encryption scheme with decisional PoNIs under the same assumption in the quantum random oracle model.
\end{corollary}

\section{Decryption with Deletion}\label{sec:deldec}

Finally, we show how to equip our constructions with an interactive classical decryption protocol that also deletes the data. At the end of the protocol, the client learns the encrypted message while the server cannot decrypt even if they are later given the decryption key. Putting this together with classical upload and PoNIs, our constructions allow classical management of data throughout its entire useful lifespan.

\begin{definition}
    A semi-quantum encryption scheme has \textbf{decryption-with-deletion} if there is an interactive protocol \(\DecDel\), using only classical communication, between a classical \(\PPT\) receiver (or decryptor) \(R_C\) and a quantum \(\QPT\) sender (or server) \(S_Q\), with the following syntax and properties.
    \begin{itemize}
        \item \textbf{Syntax.}
        The decryptor inputs the decryption key \(\sk\) and the verification key \(\verkey\) associated with the ciphertext.
        The server inputs a quantum ciphertext \(\ket{\ct}\) and \(\verkey\).
        At the end, the decryptor outputs a decision \(d\in\{\Accept,\Reject\}\) and a message \(m'\in\{0,1,\bot\}\), with \(m'=\bot\) on rejection.
        We denote its outputs by
        \[
            (d,m')\gets\DecDel\langle S_Q(\ket{\ct},\verkey),R_C(\sk,\verkey)\rangle(1^\secpar).
        \]

        \item \textbf{Correctness.}
        For all sufficiently large \(\secpar\in\bbN\), all \((\pk,\sk)\) in the support of \(\KeyGen(1^\secpar)\), and all \(m\in\{0,1\}\),
        \[
            \Pr\left[\begin{array}{c}d=\Accept\\m'=m\end{array}:\begin{array}{l}
                (\verkey,\ket{\ct})\gets\Enc\langle S_C(\pk,m),R_Q\rangle(1^\secpar)\\
                (d,m')\gets\DecDel\langle S_Q(\ket{\ct},\verkey),R_C(\sk,\verkey)\rangle(1^\secpar)
            \end{array}\right]\geq1-\negl.
        \]
The quantum party first \emph{receives} the ciphertext, playing \(R_Q\), and later \emph{sends} the decryption information, playing \(S_Q\).

        \item \textbf{Forward Secrecy.}
        For every \(\QPT\) adversary \((R_Q^*,S_Q^*)\), allowing auxiliary quantum input,
        \[
            \mathsf{Del\text{-}Dec\text{-}Game}_{\mathrm{enc}}(1^\secpar,0,(R_Q^*,S_Q^*))
            \approx_c
            \mathsf{Del\text{-}Dec\text{-}Game}_{\mathrm{enc}}(1^\secpar,1,(R_Q^*,S_Q^*)),
        \]
        where \(\mathsf{Del\text{-}Dec\text{-}Game}_{\mathrm{enc}}(1^\secpar,m,(R_Q^*,S_Q^*))\) is the following experiment.
        \begin{enumerate}
            \item Sample \((\pk,\sk)\gets\KeyGen(1^\secpar)\).
            \item Run \((\verkey,\calR_{\ct})\gets\Enc\langle S_{\mathsf C}(\pk,m),R_Q^*\rangle(1^\secpar)\), where \(\calR_{\ct}\) includes the receiver's entire retained state.
            If encryption aborts, output \(\bot\).
            \item Run
            \[
                (d,m')\gets\DecDel\langle S_Q^*(\calR_{\ct},\pk,\verkey),R_C(\sk,\verkey)\rangle(1^\secpar),
            \]
            and let \(\calR_{\Del}\) be the server's entire residual state, including its classical transcript, \(\pk\), and \(\verkey\).\footnote{One may also define a version where \(\verkey\) is not revealed to \(S_Q^*\).
            We do not call these versions ``publicly verifiable'' or ``privately verifiable'': the verification is interactive, and its decision need not be publicly verifiable even when \(\verkey\) is public.}
            \item If \(d=\Accept\), output \((\calR_{\Del},\sk)\); otherwise, output \(\bot\).
        \end{enumerate}
        The decryptor's output \(m'\) and its private session state are not revealed in this experiment.
        In particular, exposure of the long-term decryption key does not include the fresh trapdoor used only during \(\DecDel\).
    \end{itemize}
\end{definition}

We augment the encryption scheme in \Cref{constr:sq-pke-cd}, instantiated with the concrete sharing scheme in \Cref{constr:sq-cd-keygen}, with the following protocol.
For an honestly generated ciphertext, write
\[
    \begin{aligned}
        \ket{\ct}&=(\ket{\qsh},\widetilde{\csh}),
        &\widetilde{\csh}&\gets\PKE.\Enc(\pk,(\widetilde m,\td_A)),\\
        \verkey&=(\pubparams_A,y_A),
        &\widetilde m&=p\oplus m,
    \end{aligned}
\]
where, up to a global phase,
\[
    \ket{\qsh}=\frac{1}{\sqrt2}\left(\ket{0,x_0^{(A)}}+(-1)^p\ket{1,x_1^{(A)}}\right).
\]
Here \(p=d_B\cdot(x_0^{(B)}\oplus x_1^{(B)})\) is the phase produced by sharing.
We use the same lossy-mode claw-state preparation and inversion guarantees as in \Cref{constr:sq-cd-keygen}.
In particular, the construction and proof below require each valid lossy-mode image to have one preimage on each branch.

\begin{construction}[Decryption with Deletion]
    The algorithms \(\KeyGen,\Enc,\Dec,\Del,\Ver\) are those of \Cref{constr:sq-pke-cd} with the concrete sharing scheme above.
    The additional protocol \(\DecDel\) proceeds as follows.
    \begin{enumerate}
        \item \textbf{Server.}
        Parse \(\ket{\ct}=(\ket{\qsh},\widetilde{\csh})\) and send the classical component \(\widetilde{\csh}\) to the decryptor.
        The decryptor uses its input \(\verkey=(\pubparams_A,y_A)\) throughout the protocol and postpones decryption of \(\widetilde{\csh}\) until the last step.

        \item \textbf{Decryptor.}
        Independently sample \((\pubparams_C,\td_C)\gets\Setup_{\calF}(1^\secpar,\lossy)\) and send \(\pubparams_C\) to the server.

        \item \textbf{Server.}
        Append \(\ket{\rand_{\pubparams_C}}\), evaluate \(f_{\pubparams_C}\) coherently using the branch bit of \(\ket{\qsh}\), and measure the function-output register to obtain \(y_C\).
        Send \(y_C\) to the decryptor.
        For an honest server, the remaining state is, up to negligible error and a global phase,
        \[
            \frac{1}{\sqrt2}\left(\ket{0,x_0^{(A)},x_0^{(C)}}+(-1)^p\ket{1,x_1^{(A)},x_1^{(C)}}\right),
        \]
        where \(f_{\pubparams_C}(0;x_0^{(C)})=f_{\pubparams_C}(1;x_1^{(C)})=y_C\).

        \item \textbf{Both parties.}
        Execute \(\SPNP\),{} the simulation-extractable witness-preserving argument from \Cref{coro:wp-np-eps-iext}, for the statement \((y_A,y_C)\) and language
        \[
            \lang_{\pubparams_A,\pubparams_C}
            =\left\{(y_A,y_C):\exists(b,x^{(A)},x^{(C)})\text{ such that }
            \begin{array}{l}
                f_{\pubparams_A}(b;x^{(A)})=y_A,\\
                f_{\pubparams_C}(b;x^{(C)})=y_C
            \end{array}\right\}.
        \]
        The server acts as the prover, using its three remaining registers as the witness register, and the decryptor acts as the verifier.
        If the verifier rejects, the decryptor outputs \((\Reject,\bot)\).

        \item \textbf{Server.}
        Measure the branch bit and the \(A\)- and \(C\)-preimage registers in the Hadamard basis, obtaining \(b'\concat d_A\concat d_C\), and send this string to the decryptor.

        \item \textbf{Decryptor.}
        Compute \(\csh\gets\PKE.\Dec(\sk,\widetilde{\csh})\) and parse \(\csh=(\widetilde m,\td_A)\).
        For \(b\in\{0,1\}\), compute
        \[
            x_b^{(A)}\gets\Inv(\td_A,b,y_A)
            \qquad\text{and}\qquad
            x_b^{(C)}\gets\Inv(\td_C,b,y_C).
        \]
        Reject with output \((\Reject,\bot)\) if a message is malformed, decryption or inversion fails, or a recovered preimage fails its corresponding evaluation equation.
        Here \(b'\) must be a bit, and \(d_A,d_C\) must have the corresponding preimage lengths.
        Otherwise, set
        \[
            p_A=d_A\cdot\left(x_0^{(A)}\oplus x_1^{(A)}\right)
            \qquad\text{and}\qquad
            p_C=d_C\cdot\left(x_0^{(C)}\oplus x_1^{(C)}\right),
        \]
        and output \((\Accept,\widetilde m\oplus b'\oplus p_A\oplus p_C)\).
        This output is local to the decryptor; no further message is sent to the server.
    \end{enumerate}
    The decryptor uses a fresh \(C\)-instance for every execution and erases \(\td_C\), its generation coins, and the recovered \(C\)-preimages after the execution.
\end{construction}

\begin{theorem}
    Assuming the post-quantum security of LWE, there exists a semi-quantum encryption scheme with publicly-verifiable certified deletion which also has decryption-with-deletion.
\end{theorem}
\begin{proof}
    We first show correctness.
    Witness preservation leaves the honest three-register state negligibly close to its state before the argument, and the verifier accepts with overwhelming probability.
    For the ideal claw state, the amplitude of a Hadamard outcome \(b'\concat d_A\concat d_C\) is proportional to
    \[
        (-1)^{d_A\cdot x_0^{(A)}+d_C\cdot x_0^{(C)}}
        \left(1+(-1)^{p+b'+p_A+p_C}\right).
    \]
    Consequently every outcome with nonzero amplitude satisfies
    \[
        p=b'\oplus p_A\oplus p_C.
    \]
    Decryption and trapdoor inversion therefore give
    \[
        \widetilde m\oplus b'\oplus p_A\oplus p_C
        =(p\oplus m)\oplus p=m.
    \]
    The negligible state-preparation, witness-preservation, and underlying decryption errors contribute only a negligible total failure probability.

    \paragraph{Forward Secrecy.}
    We show that the adversary's distinguishing advantage between $b = 0$ and $b=1$ in $\mathsf{Del\text{-}Dec\text{-}Game}_{\mathrm{enc}}(1^\secpar,b,(R_Q^*,S_Q^*))$ must be less than $\epsilon$ for every $\epsilon = 1/\poly$. Fix any such $\epsilon$.
    
    Consider the following modified decryption procedure $\DecDel'_{\epsilon/2}$ which uses the simulator-extractor of $\SPNP$.
    Let ($\Sim$, $\Ext_\epsilon$) be the QPT simulator and size $\poly[1/\epsilon]$ unitary guaranteed by $\epsilon$-indistinguishable-extraction. Then $\DecDel'_{\epsilon/4}$ does the following:
    \begin{enumerate}
        \item Run $\DecDel$ until the witness-preserving argument.
        \item Run $\Sim$.
        \item Set $\gamma = (\epsilon/4)^2$. Run $\Ext_\gamma$ to extract a witness register $\calW$, measure whether the witness register contains a valid witness, then run $\Ext_\gamma^\dagger$. If the extraction fails, abort.
        \item Complete $\DecDel$.
    \end{enumerate}
    Next, we use $\DecDel'_{\epsilon/2}$ to define the following hybrid experiments:
    \begin{itemize}
        \item $\Hyb_0(0)$ is the original forward security game played with message bit $b=0$.
        \item $\Hyb_1(0)$ is identical to $\Hyb_0(b)$ except it replaces $\DecDel$ with $\DecDel_{\epsilon/4}$.
        \item $\Hyb_2(0)$ is identical to $\Hyb_{2}(0)$ except it measures the extracted witness in the computational basis in between $\Ext_\gamma$ and $\Ext_\gamma^\dagger$.
        \item $\Hyb_{2}(1)$ is identical to $\Hyb_{2}(0)$ except the message bit is changed to $b=1$.
        \item $\Hyb_{1}(1)$ and $\Hyb_{0}(1)$ invert hybrids 2 and 1, so that $\Hyb_0(1)$ is the forward security game played with message bit $b=1$.
    \end{itemize}

    $\Hyb_{0}(b)$ is $(\epsilon/4 + (\epsilon/4)^2)$-indistinguishable from $\Hyb_1(b)$ by \Cref{lem:iext-close} for both $b \in \{0,1\}$. Note that $\epsilon/4 + (\epsilon/4)^2 < 1/3\epsilon$ for all $\epsilon \in (0,1]$.

    $\Hyb_{1}(b)$ is indistinguishable from $\Hyb_{2}(b)$ by the collapsing property of $\pubparams_C$, for both $b\in \{0,1\}$. Note that the forward security game \emph{only} uses $\pubparams_C$, and not any information dependent on its trapdoor.

    $\Hyb_{2}(0)$ is indistinguishable from $\Hyb_{2}(1)$ by the publicly-verifiable certified deletion security of $\Enc$. Explicitly, the reduction works as follows. Take as input $\Enc(b)$ and the verification key $\verkey = (y_A, \pubparams_A)$. Use these to run $\Hyb_2(b)$ until $\sk$ is required, obtaining $(b, x_{b}^{(A)}, x_{b}^{(C)})$ from measuring the extracted witness. Note that $\sk$ is not required if the extraction fails because the experiment aborts immediately and the adversary does not get to guess. If extraction succeeds, send $(b, x_{b}^{(A)})$ to the certified deletion challenger as the certificate. Since this is a valid certificate when extraction succeeds, receive $\sk$ in return. Use $\sk$ to complete $\Hyb_2(b)$ and output the distinguishing guess. Observe that the adversary's view in this reduction is identical to $\Hyb_{2}(b)$ and so the reduction's advantage in the certified deletion game is identical; it must therefore be negligible.

    Combining these facts, the adversary's advantage in distinguishing $\Hyb_0(0)$ from $\Hyb_0(1)$ must be $\leq 2\epsilon/3 + \negl < \epsilon$ for all $\epsilon = 1/\poly$. Thus, the advantage must be negligible. 

\end{proof}

\section{Acknowledgments}

We thank Dakshita Khurana for her helpful discussions during this project.

\section{AI Disclosure}

ChatGPT was used during the preparation of this paper. The technical ideas are wholly the human authors' own.

The primary usage of ChatGPT was for giving feedback on human-generated drafts. The largest contribution of ChatGPT was in the applications section, where it was asked to import definitions and prove specific applications based on the main technical sections. These proofs were heavily human- edited or rewritten for accuracy and exposition.

\clearpage
\bibliographystyle{alpha}
\bibliography{Bib/abbrev3,Bib/crypto,Bib/project}

@string{ieee =                  {IEEE}}

@string{springer =              "Springer"}

@string{mylncs =                "{LNCS}"}

@string{mylipics =              "{LIPIcs}"}

@string{asiacrypt16name2 =      asiacryptname # "~2016, Part~II"}

@string{asiacrypt16ed =         "Jung Hee Cheon and Tsuyoshi Takagi"}

@string{asiacrypt16vol2 =       "10032"}

@string{asiacrypt16addr =       ""}

@string{asiacrypt16month =      dec}

@string{asiacrypt16pub =        springer_berlin_heidelberg}

@string{asiacrypt19name1 =      asiacryptname # "~2019, Part~I"}

@string{asiacrypt19ed =         "Steven D. Galbraith and Shiho Moriai"}

@string{asiacrypt19vol1 =       "11921"}

@string{asiacrypt19addr =       ""}

@string{asiacrypt19month =      dec}

@string{asiacrypt19pub =        springer_cham}

@string{asiacrypt21name1 =      asiacryptname # "~2021, Part~I"}

@string{asiacrypt21ed =         "Mehdi Tibouchi and Huaxiong Wang"}

@string{asiacrypt21vol1 =       "13090"}

@string{asiacrypt21addr =       ""}

@string{asiacrypt21month =      dec}

@string{asiacrypt21pub =        springer_cham}

@string{asiacrypt25name8 =      asiacryptname # "~2025, Part~VIII"}

@string{asiacrypt25ed =         "Goichiro Hanaoka and Bo-Yin Yang"}

@string{asiacrypt25vol8 =       "16252"}

@string{asiacrypt25addr =       ""}

@string{asiacrypt25month =      dec}

@string{asiacrypt25pub =        springer_singapore}

@string{cryptoaddr =            ""}

@string{crypto13name1 =         cryptoname # "~2013, Part~I"}

@string{crypto13ed =            "Ran Canetti and Juan A. Garay"}

@string{crypto13vol1 =          "8042"}

@string{crypto13month =         aug}

@string{crypto13pub =           springer_berlin_heidelberg}

@string{crypto22name1 =         cryptoname # "~2022, Part~I"}

@string{crypto22ed =            "Yevgeniy Dodis and Thomas Shrimpton"}

@string{crypto22vol1 =          "13507"}

@string{crypto22month =         aug}

@string{crypto22pub =           springer_cham}

@string{crypto23name5 =         cryptoname # "~2023, Part~V"}

@string{crypto23ed =            "Helena Handschuh and Anna Lysyanskaya"}

@string{crypto23vol5 =          "14085"}

@string{crypto23month =         aug}

@string{crypto23pub =           springer_cham}

@string{crypto24name7 =         cryptoname # "~2024, Part~VII"}

@string{crypto24ed =            "Leonid Reyzin and Douglas Stebila"}

@string{crypto24vol7 =          "14926"}

@string{crypto24month =         aug}

@string{crypto24pub =           springer_cham}

@string{crypto25name2 =         cryptoname # "~2025, Part~II"}

@string{crypto25ed =            "Yael Tauman Kalai and Seny F. Kamara"}

@string{crypto25vol2 =          "16001"}

@string{crypto25month =         aug}

@string{crypto25pub =           springer_cham}

@string{eurocrypt23name1 =      eurocryptname # "~2023, Part~I"}

@string{eurocrypt23ed =         "Carmit Hazay and Martijn Stam"}

@string{eurocrypt23vol1 =       "14004"}

@string{eurocrypt23addr =       ""}

@string{eurocrypt23month =      apr}

@string{eurocrypt23pub =        springer_cham}

@string{eurocrypt24name3 =      eurocryptname # "~2024, Part~III"}

@string{eurocrypt24name4 =      eurocryptname # "~2024, Part~IV"}

@string{eurocrypt24ed =         "Marc Joye and Gregor Leander"}

@string{eurocrypt24vol3 =       "14653"}

@string{eurocrypt24vol4 =       "14654"}

@string{eurocrypt24addr =       ""}

@string{eurocrypt24month =      may}

@string{eurocrypt24pub =        springer_cham}

@string{eurocrypt25name3 =      eurocryptname # "~2025, Part~III"}

@string{eurocrypt25ed =         "Serge Fehr and Pierre-Alain Fouque"}

@string{eurocrypt25vol3 =       "15603"}

@string{eurocrypt25addr =       ""}

@string{eurocrypt25month =      may}

@string{eurocrypt25pub =        springer_cham}

@string{eurocrypt26name1 =      eurocryptname # "~2026, Part~I"}

@string{eurocrypt26ed =         "Joan Daemen and Emmanuel Thom{\'e}"}

@string{eurocrypt26vol1 =       "16541"}

@string{eurocrypt26addr =       ""}

@string{eurocrypt26month =      may}

@string{eurocrypt26pub =        springer_cham}

@string{focspub =               "{IEEE} Computer Society Press"}

@string{focs86name =            "27th " # focsname}

@string{focs86addr =            ""}

@string{focs86month =           oct}

@string{focs19name =            "60th " # focsname}

@string{focs19ed =              "David Zuckerman"}

@string{focs19addr =            ""}

@string{focs19month =           nov}

@string{focs22name =            "63rd " # focsname}

@string{focs22addr =            ""}

@string{focs22month =           oct # "~/~" # nov}

@string{icalppubv2 =            dagstuhl}

@string{icalp23name =           "ICALP 2023" # icalpname}

@string{icalp23ed =             "Kousha Etessami and Uriel Feige and Gabriele Puppis"}

@string{icalp23vol =            "261"}

@string{icalp23addr =           ""}

@string{icalp23month =          jul}

@string{itcspub_v3 =            mylipics}

@string{itcs25name =            "ITCS 2025" # itcsname_v4}

@string{itcs25ed =              "Raghu Meka"}

@string{itcs25vol =             "325"}

@string{itcs25addr =            ""}

@string{itcs25month =           jan}

@string{pkc13name =             "PKC~2013" # pkcname_v2}

@string{pkc13ed =               "Kaoru Kurosawa and Goichiro Hanaoka"}

@string{pkc13vol =              "7778"}

@string{pkc13addr =             ""}

@string{pkc13month =            feb # "~/~" # mar}

@string{pkc13pub =              springer_berlin_heidelberg}

@string{tcc20name3 =            "TCC~2020" # tccname  # ", Part~III"}

@string{tcc20ed =               "Rafael Pass and Krzysztof Pietrzak"}

@string{tcc20vol3 =             "12552"}

@string{tcc20addr =             ""}

@string{tcc20month =            nov}

@string{tcc20pub =              springer_cham}

@string{tcc23name4 =            "TCC~2023" # tccname  # ", Part~IV"}

@string{tcc23ed =               "Guy N. Rothblum and Hoeteck Wee"}

@string{tcc23vol4 =             "14372"}

@string{tcc23addr =             ""}

@string{tcc23month =            nov # "~/~" # dec}

@string{tcc23pub =              springer_cham}

@string{tcc24name2 =            "TCC~2024" # tccname  # ", Part~II"}

@string{tcc24name3 =            "TCC~2024" # tccname  # ", Part~III"}

@string{tcc24ed =               "Elette Boyle and Mohammad Mahmoody"}

@string{tcc24vol2 =             "15365"}

@string{tcc24vol3 =             "15366"}

@string{tcc24addr =             ""}

@string{tcc24month =            dec}

@string{tcc24pub =              springer_cham}

@InProceedings{AC:AbbKat25,
  author =       "Kasra Abbaszadeh and
                  Jonathan Katz",
  title =        "Non-Interactive Zero-Knowledge Arguments with Certified Deletion",
  pages =        "381--410",
  editor =       asiacrypt25ed,
  booktitle =    asiacrypt25name8,
  volume =       asiacrypt25vol8,
  address =      asiacrypt25addr,
  month =        asiacrypt25month,
  publisher =    asiacrypt25pub,
  series =       mylncs,
  year =         2025,
  doi =          "10.1007/978-981-95-5125-5_13",
}

@InProceedings{AC:HMNY21,
  author =       "Taiga Hiroka and
                  Tomoyuki Morimae and
                  Ryo Nishimaki and
                  Takashi Yamakawa",
  title =        "Quantum Encryption with Certified Deletion, Revisited: Public Key, Attribute-Based, and Classical Communication",
  pages =        "606--636",
  editor =       asiacrypt21ed,
  booktitle =    asiacrypt21name1,
  volume =       asiacrypt21vol1,
  address =      asiacrypt21addr,
  month =        asiacrypt21month,
  publisher =    asiacrypt21pub,
  series =       mylncs,
  year =         2021,
  doi =          "10.1007/978-3-030-92062-3_21",
}

@InProceedings{AC:CCKW19,
  author =       "Alexandru Cojocaru and
                  L{\'e}o Colisson and
                  Elham Kashefi and
                  Petros Wallden",
  title =        "{QFactory}: Classically-Instructed Remote Secret Qubits Preparation",
  pages =        "615--645",
  editor =       asiacrypt19ed,
  booktitle =    asiacrypt19name1,
  volume =       asiacrypt19vol1,
  address =      asiacrypt19addr,
  month =        asiacrypt19month,
  publisher =    asiacrypt19pub,
  series =       mylncs,
  year =         2019,
  doi =          "10.1007/978-3-030-34578-5_22",
}

@InProceedings{AC:Unruh16,
  author =       "Dominique Unruh",
  title =        "Collapse-Binding Quantum Commitments Without Random Oracles",
  pages =        "166--195",
  editor =       asiacrypt16ed,
  booktitle =    asiacrypt16name2,
  volume =       asiacrypt16vol2,
  address =      asiacrypt16addr,
  month =        asiacrypt16month,
  publisher =    asiacrypt16pub,
  series =       mylncs,
  year =         2016,
  doi =          "10.1007/978-3-662-53890-6_6",
}

@InProceedings{C:BarKhu25,
  author =       "James Bartusek and
                  Dakshita Khurana",
  title =        "On the Power of Oblivious State Preparation",
  pages =        "575--607",
  editor =       crypto25ed,
  booktitle =    crypto25name2,
  volume =       crypto25vol2,
  address =      cryptoaddr,
  month =        crypto25month,
  publisher =    crypto25pub,
  series =       mylncs,
  year =         2025,
  doi =          "10.1007/978-3-032-01878-6_19",
}

@InProceedings{C:BarRai24,
  author =       "James Bartusek and
                  Justin Raizes",
  title =        "Secret Sharing with Certified Deletion",
  pages =        "184--214",
  editor =       crypto24ed,
  booktitle =    crypto24name7,
  volume =       crypto24vol7,
  address =      cryptoaddr,
  month =        crypto24month,
  publisher =    crypto24pub,
  series =       mylncs,
  year =         2024,
  doi =          "10.1007/978-3-031-68394-7_7",
}

@InProceedings{C:BarKhuPor23,
  author =       "James Bartusek and
                  Dakshita Khurana and
                  Alexander Poremba",
  title =        "Publicly-Verifiable Deletion via Target-Collapsing Functions",
  pages =        "99--128",
  editor =       crypto23ed,
  booktitle =    crypto23name5,
  volume =       crypto23vol5,
  address =      cryptoaddr,
  month =        crypto23month,
  publisher =    crypto23pub,
  series =       mylncs,
  year =         2023,
  doi =          "10.1007/978-3-031-38554-4_4",
}

@InProceedings{C:BarKhu23,
  author =       "James Bartusek and
                  Dakshita Khurana",
  title =        "Cryptography with Certified Deletion",
  pages =        "192--223",
  editor =       crypto23ed,
  booktitle =    crypto23name5,
  volume =       crypto23vol5,
  address =      cryptoaddr,
  month =        crypto23month,
  publisher =    crypto23pub,
  series =       mylncs,
  year =         2023,
  doi =          "10.1007/978-3-031-38554-4_7",
}

@InProceedings{C:HMNY22,
  author =       "Taiga Hiroka and
                  Tomoyuki Morimae and
                  Ryo Nishimaki and
                  Takashi Yamakawa",
  title =        "Certified Everlasting Zero-Knowledge Proof for {QMA}",
  pages =        "239--268",
  editor =       crypto22ed,
  booktitle =    crypto22name1,
  volume =       crypto22vol1,
  address =      cryptoaddr,
  month =        crypto22month,
  publisher =    crypto22pub,
  series =       mylncs,
  year =         2022,
  doi =          "10.1007/978-3-031-15802-5_9",
}

@InProceedings{C:GenSahWat13,
  author =       "Craig Gentry and
                  Amit Sahai and
                  Brent Waters",
  title =        "Homomorphic Encryption from Learning with Errors: Conceptually-Simpler, Asymptotically-Faster, Attribute-Based",
  pages =        "75--92",
  editor =       crypto13ed,
  booktitle =    crypto13name1,
  volume =       crypto13vol1,
  address =      cryptoaddr,
  month =        crypto13month,
  publisher =    crypto13pub,
  series =       mylncs,
  year =         2013,
  doi =          "10.1007/978-3-642-40041-4_5",
}

@InProceedings{EC:GoyRai26,
  author =       "Vipul Goyal and
                  Justin Raizes",
  title =        "Proofs of No Intrusion",
  pages =        "510--540",
  editor =       eurocrypt26ed,
  booktitle =    eurocrypt26name1,
  volume =       eurocrypt26vol1,
  address =      eurocrypt26addr,
  month =        eurocrypt26month,
  publisher =    eurocrypt26pub,
  series =       mylncs,
  year =         2026,
  doi =          "10.1007/978-3-032-25291-3_18",
}

@InProceedings{EC:CGJL25,
  author =       "Orestis Chardouvelis and
                  Vipul Goyal and
                  Aayush Jain and
                  Jiahui Liu",
  title =        "Quantum Key Leasing for {PKE} and {FHE} with a Classical Lessor",
  pages =        "248--277",
  editor =       eurocrypt25ed,
  booktitle =    eurocrypt25name3,
  volume =       eurocrypt25vol3,
  address =      eurocrypt25addr,
  month =        eurocrypt25month,
  publisher =    eurocrypt25pub,
  series =       mylncs,
  year =         2025,
  doi =          "10.1007/978-3-031-91131-6_9",
}

@InProceedings{EC:HKMNPY24,
  author =       "Taiga Hiroka and
                  Fuyuki Kitagawa and
                  Tomoyuki Morimae and
                  Ryo Nishimaki and
                  Tapas Pal and
                  Takashi Yamakawa",
  title =        "Certified Everlasting Secure Collusion-Resistant Functional Encryption, and More",
  pages =        "434--456",
  editor =       eurocrypt24ed,
  booktitle =    eurocrypt24name3,
  volume =       eurocrypt24vol3,
  address =      eurocrypt24addr,
  month =        eurocrypt24month,
  publisher =    eurocrypt24pub,
  series =       mylncs,
  year =         2024,
  doi =          "10.1007/978-3-031-58734-4_15",
}

@InProceedings{EC:BGKMRR24,
  author =       "James Bartusek and
                  Vipul Goyal and
                  Dakshita Khurana and
                  Giulio Malavolta and
                  Justin Raizes and
                  Bhaskar Roberts",
  title =        "Software with Certified Deletion",
  pages =        "85--111",
  editor =       eurocrypt24ed,
  booktitle =    eurocrypt24name4,
  volume =       eurocrypt24vol4,
  address =      eurocrypt24addr,
  month =        eurocrypt24month,
  publisher =    eurocrypt24pub,
  series =       mylncs,
  year =         2024,
  doi =          "10.1007/978-3-031-58737-5_4",
}

@InProceedings{EC:AKNYY23,
  author =       "Shweta Agrawal and
                  Fuyuki Kitagawa and
                  Ryo Nishimaki and
                  Shota Yamada and
                  Takashi Yamakawa",
  title =        "Public Key Encryption with Secure Key Leasing",
  pages =        "581--610",
  editor =       eurocrypt23ed,
  booktitle =    eurocrypt23name1,
  volume =       eurocrypt23vol1,
  address =      eurocrypt23addr,
  month =        eurocrypt23month,
  publisher =    eurocrypt23pub,
  series =       mylncs,
  year =         2023,
  doi =          "10.1007/978-3-031-30545-0_20",
}

@InProceedings{FOCS:Zhang22,
  author =       "Jiayu Zhang",
  title =        "Classical Verification of Quantum Computations in Linear Time",
  pages =        "46--57",
  booktitle =    focs22name,
  address =      focs22addr,
  month =        focs22month,
  publisher =    focspub,
  year =         2022,
  doi =          "10.1109/FOCS54457.2022.00012",
}

@InProceedings{FOCS:LomMaSpo22,
  author =       "Alex Lombardi and
                  Fermi Ma and
                  Nicholas Spooner",
  title =        "Post-Quantum Zero Knowledge, Revisited or: How to Do Quantum Rewinding Undetectably",
  pages =        "851--859",
  booktitle =    focs22name,
  address =      focs22addr,
  month =        focs22month,
  publisher =    focspub,
  year =         2022,
  doi =          "10.1109/FOCS54457.2022.00086",
}

@InProceedings{FOCS:GheVid19,
  author =       "Alexandru Gheorghiu and
                  Thomas Vidick",
  title =        "Computationally-Secure and Composable Remote State Preparation",
  pages =        "1024--1033",
  editor =       focs19ed,
  booktitle =    focs19name,
  address =      focs19addr,
  month =        focs19month,
  publisher =    focspub,
  year =         2019,
  doi =          "10.1109/FOCS.2019.00066",
}

@InProceedings{FOCS:GolMicWig86,
  author =       "Oded Goldreich and
                  Silvio Micali and
                  Avi Wigderson",
  title =        "Proofs that Yield Nothing But their Validity and a Methodology of Cryptographic Protocol Design (Extended Abstract)",
  pages =        "174--187",
  booktitle =    focs86name,
  address =      focs86addr,
  month =        focs86month,
  publisher =    focspub,
  year =         1986,
  doi =          "10.1109/SFCS.1986.47",
}

@InProceedings{ICALP:GheMetPor23,
  author =       "Alexandru Gheorghiu and
                  Tony Metger and
                  Alexander Poremba",
  title =        "Quantum Cryptography with Classical Communication: Parallel Remote State Preparation for Copy-Protection, Verification, and More",
  pages =        "67:1--67:17",
  editor =       icalp23ed,
  booktitle =    icalp23name,
  volume =       icalp23vol,
  address =      icalp23addr,
  month =        icalp23month,
  publisher =    icalppubv2,
  series =       mylipics,
  year =         2023,
  doi =          "10.4230/LIPIcs.ICALP.2023.67",
}

@InProceedings{ITCS:Zhang25,
  author =       "Jiayu Zhang",
  title =        "Formulations and Constructions of Remote State Preparation with Verifiability, with Applications",
  pages =        "96:1--96:19",
  editor =       itcs25ed,
  booktitle =    itcs25name,
  volume =       itcs25vol,
  address =      itcs25addr,
  month =        itcs25month,
  publisher =    itcspub_v3,
  year =         2025,
  doi =          "10.4230/LIPIcs.ITCS.2025.96",
}

@InProceedings{PKC:KatThiZho13,
  author =       "Jonathan Katz and
                  Aishwarya Thiruvengadam and
                  Hong-Sheng Zhou",
  title =        "Feasibility and Infeasibility of Adaptively Secure Fully Homomorphic Encryption",
  pages =        "14--31",
  editor =       pkc13ed,
  booktitle =    pkc13name,
  volume =       pkc13vol,
  address =      pkc13addr,
  month =        pkc13month,
  publisher =    pkc13pub,
  series =       mylncs,
  year =         2013,
  doi =          "10.1007/978-3-642-36362-7_2",
}

@InProceedings{TCC:CGLR24,
  author =       "Alper {\c C}akan and
                  Vipul Goyal and
                  Chen-Da Liu-Zhang and
                  Jo{\~a}o Ribeiro",
  title =        "Unbounded Leakage-Resilience and Intrusion-Detection in a Quantum World",
  pages =        "159--191",
  editor =       tcc24ed,
  booktitle =    tcc24name2,
  volume =       tcc24vol2,
  address =      tcc24addr,
  month =        tcc24month,
  publisher =    tcc24pub,
  series =       mylncs,
  year =         2024,
  doi =          "10.1007/978-3-031-78017-2_6",
}

@InProceedings{TCC:AnaHuHua24,
  author =       "Prabhanjan Ananth and
                  Zihan Hu and
                  Zikuan Huang",
  title =        "Quantum Key-Revocable Dual-Regev Encryption, Revisited",
  pages =        "257--288",
  editor =       tcc24ed,
  booktitle =    tcc24name3,
  volume =       tcc24vol3,
  address =      tcc24addr,
  month =        tcc24month,
  publisher =    tcc24pub,
  series =       mylncs,
  year =         2024,
  doi =          "10.1007/978-3-031-78020-2_9",
}

@InProceedings{TCC:AnaPorVai23,
  author =       "Prabhanjan Ananth and
                  Alexander Poremba and
                  Vinod Vaikuntanathan",
  title =        "Revocable Cryptography from Learning with Errors",
  pages =        "93--122",
  editor =       tcc23ed,
  booktitle =    tcc23name4,
  volume =       tcc23vol4,
  address =      tcc23addr,
  month =        tcc23month,
  publisher =    tcc23pub,
  series =       mylncs,
  year =         2023,
  doi =          "10.1007/978-3-031-48624-1_4",
}

@InProceedings{TCC:CheHerVu23,
  author =       "C{\'e}line Chevalier and
                  Paul Hermouet and
                  Quoc-Huy Vu",
  title =        "Semi-quantum Copy-Protection and More",
  pages =        "155--182",
  editor =       tcc23ed,
  booktitle =    tcc23name4,
  volume =       tcc23vol4,
  address =      tcc23addr,
  month =        tcc23month,
  publisher =    tcc23pub,
  series =       mylncs,
  year =         2023,
  doi =          "10.1007/978-3-031-48624-1_6",
}

@InProceedings{TCC:BroIsl20,
  author =       "Anne Broadbent and
                  Rabib Islam",
  title =        "Quantum Encryption with Certified Deletion",
  pages =        "92--122",
  editor =       tcc20ed,
  booktitle =    tcc20name3,
  volume =       tcc20vol3,
  address =      tcc20addr,
  month =        tcc20month,
  publisher =    tcc20pub,
  series =       mylncs,
  year =         2020,
  doi =          "10.1007/978-3-030-64381-2_4",
}

@Misc{EPRINT:DalSpo22,
  author =       "Marcel Dall'Agnol and
                  Nicholas Spooner",
  title =        "On the necessity of collapsing",
  year =         2022,
  howpublished = "Cryptology ePrint Archive, Report 2022/786",
  url =          "https://eprint.iacr.org/2022/786",
}

@inproceedings{Mahadev18,
  author    = {Urmila Mahadev},
  title     = {Classical Verification of Quantum Computations},
  booktitle = {2018 IEEE 59th Annual Symposium on Foundations of Computer Science (FOCS)},
  year      = {2018},
  pages     = {259--267},
  publisher = {IEEE},
  address   = {Paris, France},
  doi       = {10.1109/FOCS.2018.00033}
}

@inproceedings{BCMVV18,
  author    = {Zvika Brakerski and
               Paul F. Christiano and
               Urmila Mahadev and
               Umesh V. Vazirani and
               Thomas Vidick},
  editor    = {Mikkel Thorup},
  title     = {A Cryptographic Test of Quantumness and Certifiable Randomness from
               a Single Quantum Device},
  booktitle = {59th {IEEE} Annual Symposium on Foundations of Computer Science, {FOCS}
               2018, Paris, France, October 7-9, 2018},
  pages     = {320--331},
  publisher = {{IEEE} Computer Society},
  year      = {2018},
  url       = {https://doi.org/10.1109/FOCS.2018.00038},
  doi       = {10.1109/FOCS.2018.00038},
  bibsource = {dblp computer science bibliography, https://dblp.org}
}

@article{DBLP:journals/iacr/KalaiKR26,
  author       = {Yael Tauman Kalai and
                  Dakshita Khurana and
                  Justin Raizes},
  title        = {How to Classically Verify a Quantum Cat without Killing It},
  journal      = {{IACR} Cryptol. ePrint Arch.},
  volume       = {2026},
  pages        = {210},
  year         = {2026},
  url          = {https://eprint.iacr.org/2026/210},
  bibsource    = {dblp computer science bibliography, https://dblp.org}
}

@article{eprint:LST26,
  author       = {Zhengnan Lai and
                  Nicholas Spooner and
                  Max Tromanhauser},
  title        = {How to Define Expected Quantum Polynomial-Time Zero Knowledge Simulation},
  journal      = {{IACR} Cryptol. ePrint Arch.},
  volume       = {2026},
  pages        = {1500},
  year         = {2026},
  url          = {https://eprint.iacr.org/2026/1500},
  bibsource    = {dblp computer science bibliography, https://dblp.org}
}

@article{TOC:Aar05,
  author       = {Scott Aaronson},
  title        = {Limitations of Quantum Advice and One-Way Communication},
  journal      = {Theory Comput.},
  volume       = {1},
  number       = {1},
  pages        = {1--28},
  year         = {2005},
  url          = {https://doi.org/10.4086/toc.2005.v001a001},
  doi          = {10.4086/TOC.2005.V001A001},
  bibsource    = {dblp computer science bibliography, https://dblp.org}
}

@article{TIT:Win99,
  author       = {Andreas J. Winter},
  title        = {Coding theorem and strong converse for quantum channels},
  journal      = {{IEEE} Trans. Inf. Theory},
  volume       = {45},
  number       = {7},
  pages        = {2481--2485},
  year         = {1999},
  url          = {https://doi.org/10.1109/18.796385},
  doi          = {10.1109/18.796385},
  bibsource    = {dblp computer science bibliography, https://dblp.org}
}

@inproceedings{CRYPTO:KLYY26,
  author       = {Fuyuki Kitagawa and
                  Jiahui Liu and
                  Shota Yamada and
                  Takashi Yamakawa},
  editor       = {Nadia Heninger and
                  Mike Rosulek},
  title        = {A Unified Approach to Quantum Key Leasing with a Classical Lessor},
  booktitle    = {Advances in Cryptology - {CRYPTO} 2026 - 46th Annual International
                  Cryptology Conference, Santa Barbara, CA, USA, August 17-20, 2026,
                  Proceedings, Part {V}},
  series       = {Lecture Notes in Computer Science},
  volume       = {16804},
  pages        = {90--121},
  publisher    = {Springer},
  year         = {2026},
  url          = {https://doi.org/10.1007/978-3-032-35409-9\_4},
  doi          = {10.1007/978-3-032-35409-9\_4},
  bibsource    = {dblp computer science bibliography, https://dblp.org}
}

\end{document}